%% file: MagnifierACompositionalAnalysisApproachforAutonomousTrafficControl.tex
\documentclass[10pt,journal,compsoc]{IEEEtran}

\ifCLASSOPTIONcompsoc
  \usepackage[nocompress]{cite}
\else
  \usepackage{cite}
\fi

\ifCLASSINFOpdf
\else
\fi

\usepackage{amsmath,amssymb,amsfonts}
\usepackage{algorithmic}
\usepackage{subfig}
\usepackage{textcomp}
\usepackage{longtable}
\usepackage{booktabs} % For formal tables
\usepackage{graphicx}
\usepackage{multirow}
\usepackage{comment}
\usepackage{color}
\usepackage{array}
\usepackage{hyperref}
\usepackage{multicol}
\usepackage{lipsum}
\usepackage{mwe}
\usepackage{cleveref}
\usepackage{enumerate}
\usepackage{amsthm}
\usepackage{mathrsfs}
\usepackage{fmtcount}
\usepackage{supertabular}
\usepackage{mathtools}
\usepackage{listings}
\usepackage{setspace}
\usepackage[linesnumbered,ruled]{algorithm2e}

\newcommand{\sosruleNormal}[3]{\frac{\normalsize{\parbox{#1}{\setstretch{1}\center$#2$\vspace{0.5mm}}}} {\parbox{#1}{\vspace{-0.8mm}\center$#3$}}}

\newcommand{\R}{\mathbb{R}_{\geq 0}}
\newcommand{\N}{\mathcal{N}}
\newcommand{\Z}{\mathbb{Z}}
\newcommand{\AID}{\mathit{AId}}

\newcommand{\MSG}{\mathit{Msg}}

\newcounter{example}
\newenvironment{example}[1][]{\refstepcounter{example}\par\medskip
	\noindent \textbf{Example~\theexample. #1} \rmfamily}{\medskip}

\newtheorem{definition}{Definition}[section]
\newtheorem{theorem}{Theorem}[section]
\newtheorem{lemma}{Lemma}[section]

\def\BibTeX{{\rm B\kern-.05em{\sc i\kern-.025em b}\kern-.08em
		T\kern-.1667em\lower.7ex\hbox{E}\kern-.125emX}}

\newcommand{\COM}[1]{\textcolor{red}{#1}}
\newcommand{\did}[1]{{#1}}

\usepackage[utf8]{inputenc}

\usepackage{listings}
\usepackage{xcolor}

\definecolor{codegreen}{rgb}{0,0.6,0}
\definecolor{codegray}{rgb}{0.5,0.5,0.5}
\definecolor{codepurple}{rgb}{0.58,0,0.82}
\definecolor{backcolour}{rgb}{0.95,0.95,0.92}

\begin{document}
	\sloppy

%\title{Runtime Compositional Verification of Self-adaptive Track-based Traffic Control Systems}
\title{Magnifier: A Compositional Analysis Approach for Autonomous Traffic Control}

\author{Maryam~Bagheri,~
		Marjan~Sirjani,~
		Ehsan~Khamespanah,~
		Christel~Baier,~
		and~Ali~Movaghar% <-this % stops a space
\thanks{ M. Bagheri and A. Movaghar are with the Department
of Computer Engineering, Sharif University of Technology, Iran.
 M. Sirjani is with the School of IDT,  M\"{a}lardalen University, Sweden, and the School of Computer Science, Reykjavik University, Iceland. E. Khamespanah is with the School of Electrical and Computer Engineering, University of Tehran, Iran, and the School of Computer Science, Reykjavik University. C. Baier is with the department of Computer Science, Technical University of Dresden, Germany.}% <-this % stops an unwanted space
\thanks{Manuscript received XX, 2020; revised XX, XXXX.}}

\markboth{
}%
{Shell \MakeLowercase{\textit{et al.}}: Bare Demo of IEEEtran.cls for Computer Society Journals}

\IEEEtitleabstractindextext{%
\input{00-Abstract}

\begin{IEEEkeywords}
Self-adaptive Systems, Model@Runtime, Compositional Verification, Track-based Traffic Control Systems, Ptolemy II
\end{IEEEkeywords}}

\maketitle

\IEEEdisplaynontitleabstractindextext
\IEEEpeerreviewmaketitle

\input{01-Introduction}
\input{02-ProblemDomain}
\input{03-Preliminaries}
\input{04-MultipleCOODA}
\input{05-MagnifierApproachForCompositionalVerificationOfSelfAdaptive}

\input{15-Evaluation}
\input{06-RelatedWork}

\input{07-Conclusion}

\ifCLASSOPTIONcompsoc
  \section*{Acknowledgments}
\else
  \section*{Acknowledgment}
\fi

The work on this paper has been supported in part by the project "Self-Adaptive Actors: SEADA" (nr. 163205-051) of the Icelandic Research Fund. 
The work of the second author is supported in part by KKS SACSys Synergy project (Safe and Secure Adaptive Collaborative Systems),  KKS DPAC Project (Dependable Platforms for Autonomous Systems and Control), and SSF Serendipity project at Malardalen University, and MACMa Project (Modeling and Analyzing Event-based Autonomous Systems) at Software Center, Sweden.
The authors would like to thank Professor Edward A. Lee for his valuable suggestions.

\ifCLASSOPTIONcaptionsoff
  \newpage
\fi

\bibliographystyle{IEEEtran}
\bibliography{IEEEabrv,mybibfile}

\appendix
\input{appendix}

\end{document}

%% file: 00-Abstract.tex
%Without
\begin{abstract} 
Autonomous traffic control systems are large-scale systems with critical goals. To satisfy expected properties, these systems adapt themselves to  changes in their environment %Following a change, a sequence of changes may happen; to satisfy expected properties of the system, it is adapted, 
while the adaptation may result in further changes propagated through the systems. For each change and its consequent adaptation,  %verifying the safety and quality
%properties of the system is necessary. Due to the uncertain changes in the environment of these systems, they adapt themselves to satisfy the expected properties. an autonomous TCS is adapted to a change in their surrounding world, while the adaptation may result in propagating the change through the system. 
%Due to the dynamic nature of the surrounding world of these systems, 
assuring the satisfaction of the properties of a system at runtime %and in the presence of a change  
is important. A prominent approach to assure the correct behavior of these systems is  verification at runtime, which  has strict time and memory limitations. To tackle these limitations,  we propose Magnifier, an iterative, incremental,  and compositional verification approach that operates on a component-based model. The Magnifier idea is zooming on the component affected by a change and verifying the correctness of properties of interest of the system after adapting the component to the change. The satisfaction of the properties indicates that the change does not propagate further. Magnifier zooms out and traces the change if it propagates. In this case, all components affected by the change are adapted and are composed to form a new component.  Magnifier repeats the same process for the new component. This iterative process terminates whenever the propagation of the change stops. %by checking whether the change propagates further than the new component.   This process of zooming out continues until the change does not propagate. In this case, all properties of the system are preserved. 
In Magnifier, we use the Coordinated Adaptive Actor  model (CoodAA) of traffic control systems. We present a formal semantics for CoodAA as  a network of Timed Input-Output Automata (TIOAs). The change does not propagate if TIOAs of the adapted component and its environment do not reach a deadlock state in their parallel product.  We implement our approach in Ptolemy II. The results of our experiments indicate that the proposed approach improves the verification time and the memory consumption compared to a non-compositional approach.

\end{abstract}

%% file: 01-Introduction.tex
\IEEEraisesectionheading{\section{Introduction}\label{sec:introduction}}

%\IEEEPARstart{A}{utonomous} response to context changes is a mechanism to prevent a failure.  A  self-adaptive system  is able to adjust its structure and behavior in response to changes in its environment and the system itself. For each change and its consequent adaptation, verifying the safety and quality properties of the system is necessary. As the behavior of a self-adaptive system changes during runtime its correctness needs to be checked during its execution.

\IEEEPARstart{M}{any} activities of the modern society are entirely managed by traffic control systems. These systems are large-scale, time and safety-critical systems that consist of numerous moving objects whose movements on a traveling space are adjusted and coordinated by controllers.  The application domain of traffic control systems is not only limited to air traffic control systems  or rail traffic control systems, but also includes more applications such as robotic systems, maritime transportation, smart hubs,  intelligent factory lines,  etc.  The traffic in such systems can pass through pre-specified tracks, that based on the minimum safe distance between the moving objects, are partitioned into a set of sub-tracks. A  system with this structural design is called a Track-based Traffic Control System (TTCS) \cite{cooda}.

Due to the dynamic nature of a TTCS and its  surrounding world, a TTCS is vulnerable to failures, threatening  human lives or causing intolerable costs. Autonomous response to context changes is a mechanism to prevent a failure in self-adaptive systems that %.  Self-adaptive systems 
are able to adjust their structures and behaviors in response to changes. The controller in an autonomous TTCS uses the track-based design to safely and efficiently manage the traffic whenever an unpredicted change happens. For each change and its consequent adaptation, verifying the safety and quality of the system is necessary, which should be performed during the execution of the system. For performing the analysis and verification at runtime, an abstract model of the system and its environment, the so-called \textit{model@runtime} \cite{Cheng2014}, is generated, updated, and  verified during the system execution. 

In  \cite{16_CoordinatedActorsForReliableSelfAdaptiveSystems}, we introduced the Coordinated Adaptive Actor model (CoodAA)    for constructing and analyzing self-adaptive track-based traffic
control systems.
CoodAA is an actor-based \cite{hew72,Agha:1986:AMC:7929} approach augmented with coordination policies.
In CoodAA each sub-track is modeled as an actor, the moving objects are considered as messages passed by the actors, and the controller is modeled as a coordinator.
 A TTCS is a large-scale system partitioned into a set of control areas where each area has its own controller, so, a  model  of  a  TTCS  can be  intrinsically  built  as  a  set  of components and is matched to CoodAA. The moving objects are sent and received at specified times through specified routes.

The coordinated adaptive actor model is designed based on the MAPE-K feedback loop \cite{Kephart:2003:VAC:642194.642200} for self-adaptive systems.  This control loop consists of the Monitor, Analyze, Plan, and Execute components. There is also a \textit{Knowledge base} where the model@runtime is kept.
The \textit{Knowledge base} is updated by the \textit{Monitor} component. The \textit{Analyze} and \textit{Plan} components are responsible for doing the analysis and providing  adaptation plans when a change happens. The new plan  is sent to the system through the \textit{Execute} component.

In this paper, our focus is on the analysis that is performed for adaptation at runtime, and we propose the \textit{Magnifier} idea.  
%We use the model@runtime in coordinated actor model to verify a set of correctness properties on a TTCS whenever it adapts to a change. 
%The correctness properties include: the moving objects have to arrive at their destinations at the pre-specified times, the collision of the moving objects should be avoided, the fuel of  the moving objects should not be less than a threshold, and the system should be deadlock-free.
%
%
%\marjan{Collision = Conflict????}
%
Magnifier uses an iterative and incremental process on a component-based model. When a change occurs Magnifier zooms-in on the affected component and checks if  properties of interest still hold. If not, it 
 adapts the component affected by the change by finding a new plan.
 Then, Magnifier 
checks if, because of the new plan, the change is \textit{propagated} through the model to other components, and it will continue the same process iteratively and incrementally.
The idea is checking the effects of the change on the affected area and then in the least number of neighborhood components (and trying to contain it) instead of  analysing the whole system for each change.
%}
%\marjan{
The general idea of Magnifier  is not specific for track-based traffic control systems  and can be applied for any autonomous control system.  But in our work, we focus on CoodAA and TTCSs,  provide formal semantics and necessary theorems for compositional verification of TTCSs, and illustrate the results by implementing the approach. 
%}

In Magnifier, we use a  compositional approach, we focus on the interface of each component, which in CoodAA means the inputs and outputs of each component at a specified time.
%\COM{The interface of a component in CoodAA describes what the component guarantees to provide on its outputs at specified times by making some assumptions on its inputs. --Marjan: I am not sure this is the case for our approach.}
If the adapted component, according to the new plan, generates new outputs, or generates outputs by  making  new  assumptions  on  its  inputs it means  that  the  effects  of  the  change  may propagate to the connected components. So, the connected components (or  the so-called environment components) are adapted  considering the new interface of the component. Then, Magnifier zooms-out and creates a new component by composing all components adapted to the change. The propagation of the change stops if the interface of the new (composite) component remains unchanged. 

%\COM{Can we use this paragraph instead of the above paragraph???: In Magnifier, we use a compositional approach and focus on the interface of each component. The interface of a component in CoodAA describes what the component guarantees to provide on its outputs at specified times by making some assumptions on its inputs. If the adapted component, according to the new plan, generates new outputs, or generates outputs with  making  new  assumptions  on  its  inputs, it means  that  the  effects  of  the  change  may propagate to the connected components. So, the connected components (or  the so-called environment components) are adapted  considering the new interface of the component. Then, Magnifier zooms out and creates a new component by composing all components adapted to the change. The propagation of the change stops if the interface of the new (composite) component remains unchanged.}

In this paper, we first present a compositional formal semantics for CoodAA as a network of Timed Input-Output Automata (TIOAs) \cite{kaynar2003timed}.
Each component is represented by TIOAs of its constituent actors and its coordinator.
We check the propagation of a change by checking the \textit{compatibility} of  TIOAs of the adapted component and  TIOAs of its environment components.
We call two (or more) TIOAs compatible if they do not reach a deadlock state in their parallel product.

To prove our incremental compositional approach, we adopt the compositional verification theorem of Clark et al. \cite{clarke1989compositional}.
In \cite{clarke1989compositional}, each component of the model is supplied with a correctness property. By composing a component with an abstraction of its environment components and verifying a property over the composition, the satisfaction of the property over the whole system is proved. 
Similar to \cite{clarke1989compositional}, we use abstractions of the environment components. To reduce the state space, instead of TIOAs of the environment components, we only consider  TIOAs of border actors that directly communicate with the adapted component. 
In contrast to \cite{clarke1989compositional}, 
we do not use any logical formula to express the properties, since it is enough to check whether the adapted component  interacts with its environment as expected (i.e. their compatibility).

Note that the verification of the \textit{propagation} of a change is checking whether the interface of a component remains unchanged after adapting to a new plan. 
The verification is performed on the model@runtime that is a static snapshot of the system at the moment of the change.

To illustrate the applicability of our approach, we implement it in Ptolemy II \cite{ptolemaeus2014system}. Ptolemy II is an actor-oriented open-source modeling and simulation framework. A Ptolemy model consists of actors that communicate via message passing. The semantics of communications of the actors in Ptolemy is defined by models of computation, implemented in a set of predefined director components. 
%In \cite{cooda}, we developed a Ptolemy template to model and analyze self-adaptive TTCSs. Our analysis in \cite{cooda} was based on simulation, since Ptolemy II with its deterministic models of computation does not support verification of a system. 
Here, to provide assertion-based verification  in Ptolemy II, we develop a Magnifier director. Our director generates the state space of the affected component, automatically extends its domain to include other components, and performs the reachability analysis over this extended domain. The results of our experiments for an example in the domain of air traffic control systems indicate a significant improvement in the verification time and  the memory consumption. 

\noindent\textbf{Novelty and importance.}
Magnifier can be seen as a decentralized adaptation mechanism.
%It is notable that the adaptation in our approach is performed in a decentralized manner. 
Adaptation in a decentralized setting is a well-known challenge \cite{KRUPITZER2015184,deLemos2013}.
It significantly improves  scalability and is a suitable option in hard real-time settings, when the reaction to a change should be performed in a negligible amount of time \cite{deLemos2013}. On the other hand, preserving  global goals in a decentralized setting is  difficult  \cite{deLemos2013}, as several components may need to reach a consensus about an adaptation policy to satisfy a global goal. Magnifier meets the global goals by first applying local adaptation to the component affected by a change. If it is not successful, it dynamically extends its adaptation (verification) domain to consider more components. The Magnifier approach relies upon the assumption that the environment components of a component are recognisable at the analysis time.%, and if there is an adaptation policy within the adaptation scope to prevent the change to be propagated, Magnifier is always able to find it.
%We enumerate the assumptions upon which  Magnifier relies as follows:
%\begin{itemize}
%	\item The adaptation using Magnifier is performed in a decentralized setting.
%	\item  The environment components of a component are diagnosable at the analysis time.  
 %Adaptation in a decentralized setting is a well-known challenge \cite{KRUPITZER2015184,deLemos2013}, and has some pros and cons. It significantly improves the scalability and is a suitable option in hard real-time settings, when the reaction to a change should be performed in a negligible amount of time \cite{deLemos2013}. On the other hand, preserving the global goals in a decentralized setting is a difficult task \cite{deLemos2013}, as several components may need to reach a consensus about an adaptation policy to satisfy a global goal. Magnifier tries to meet global goals by first applying local adaptation to the component affected by a change. If it was not successful, it dynamically extends its adaptation (verification) domain to consider more components. 
%	\item  Magnifier is always able to adapt the component using a policy that prevents the change to be propagated if such a policy exists. %If there is an adaptation policy within the adaptation scope to prevent the change to be propagated, Magnifier is always able to adapt the component based on that policy. 
%	This assumption is realized if Magnifier tries all adaptation plans and selects the one with the mentioned aim.
%\end{itemize}

\noindent\textbf{Contribution.} 
CoodAA is introduced in \cite{16_CoordinatedActorsForReliableSelfAdaptiveSystems} and 
	its applicability in modeling TTCSs is shown by implementing a case study. In \cite{cooda},  CoodAA is explained using activity and sequence diagrams and the mapping between different applications of TTCSs and CoodAA is illustrated.
We briefly presented the Magnifier idea as a future work in  a short work-in-progress paper \cite{bagheri2016runtime}.
In the current paper, we present the formal foundation of CoodAA and Magnifier, and support the idea of effectiveness of Magnifier by an implementation of Magnifier in Ptolemy II and experimental results. The Summary of contributions are as follows:
\begin{itemize}
\item 
Formal compositional semantics of coordinated adaptive actor model as Timed Input-Output Automata (TIOAs)
\item
Compositional and incremental verification of  model@runtime in CoodAA using Magnifier, and proof of correctness of the compositional approach
\item
Abstraction technique for  environment components in Magnifier for reducing the state space %\COM{while verification} of compatibility of components
\item
Implementation of Magnifier as a director in Ptolemy II and supporting experimental results on an air traffic control system as an example to show the efficiency of Magnifier compared to a non-compositional approach
\end{itemize}

The rest of the paper is organized as follows.  We provide a general overview of TTCSs in Section~\ref{section::problemDef}.  We recall the definitions of a TIOA and the parallel composition of several TIOAs  in Section~\ref{sec::background}. %Furthermore, we introduce the coordinated actor model and the Ptolemy framework in Section~\ref{sec::background}.
In Section~\ref{section::sem},  the formal compositional semantics of CoodAA is described in terms of TIOAs. Section~\ref{section::magnifier} describes the details of the Magnifier approach. The implementation of Magnifier in Ptolemy II and the results of our experiments are shown in Section~\ref{sec::evaluation}. We describe the related work in Section~\ref{sec::relatedwork}, and conclude the paper in Section~\ref{sec::conclude}.

%% file: 02-ProblemDomain.tex
\section{Problem Definition and an  Example } \label{section::problemDef}
\begin{comment}
\begin{figure}
\centering
\includegraphics[width=5cm,height=5cm,keepaspectratio]{Fig/NAT}
\caption{North Atlantic Organized Track System \cite{ATCFig}}
\label{both}
\end{figure}
\end{comment}

Track-based Traffic Control Systems (TTCSs), introduced in \cite{cooda}, are safety-critical systems. A TTCS works
based on the track-based design of the traveling space; to reduce the risk of collision between moving objects, they move on certain tracks in the traveling space instead of moving around freely. Based on the safe distance
between two moving objects, each track is divided into a set of
sub-tracks.  %that build on the idea of moving objects traveling across  safe regions called sub-tracks. To reduce the risk of collision between moving objects in a TTCS, the traveling space is divided into a set of tracks. Based on the safe distance between two moving objects, each track is divided into smaller safe regions that are called sub-tracks.
Each sub-track is a critical section that accommodates only one moving object in-transit. A large-scale TTCS is divided into a set of areas, while the traffic of each area is controlled by a centralized controller. The  controller uses the  track-based infrastructure of the area to safely navigate the moving objects considering congestion and environmental changes.  As explained in \cite{cooda}, the application domain of TTCSs ranges from Air Traffic Control Systems (ATCs), rail traffic control systems, maritime transportation, to centralized robotic systems and intelligent factory lines. For instance, ATC in the North Atlantic follows a track-based structure that is called an organized track system \cite{TOS}. The North Atlantic organized track system consists of a set of nearly parallel tracks positioned in light of the prevailing winds to suit the traffic between Europe and North America.

%Transportation systems are safety-critical systems that use traffic control systems to monitor and control movements of a large number of moving objects along a traveling space.  One class of transportation systems is track-based transportation systems in which the traveling space is divided into a set of tracks. Based on the safe distance between two moving objects, each track is divided into a set of sub-tracks. Each sub-track is a critical section that accommodates only one moving object in-transit. Considering the congestion and environmental changes, the traffic control systems use this track-based structure to guarantee the safety and improve the performance. In \cite{cooda}, we introduced Track-based Traffic Control Systems (TTCSs) as all traffic control systems that follow the track-based infrastructure. 

In the real-world applications of TTCSs, each moving object has an initial traveling plan that is generated prior to the departure of the moving object from its source. A traveling plan consists of a route, time schedule decisions, and depending on the application, fuel, etc. The  route of a moving object is a sequence of  sub-tracks traveled by the moving object from its source  to its destination. The time schedule decisions of a moving object consist of its departure time from its source, assumed arrival time at each sub-track in its route, and assumed arrival time at its destination. %The initial traveling plans are generated prior to the departure of  moving objects from their sources, 
TTCSs are sensitive to  unforeseen changes in their context. As a consequence of a dynamic environmental change,  the traveling plans of  moving objects may require to be modified. Therefore, following a change in the context, a sequence of changes might happen. For instance in an ATC, the aircraft flight plans are changed if a storm happens in a part of their flight routes. While changing traveling plans, several safety issues should be considered; i.e. loss of the separation between two moving objects should be avoided, and the remaining fuel should be checked. To avoid conflicts, changing the traveling plan of a moving object may result in changing the traveling plans of  other moving objects. These changes can be propagated to the whole system.  Besides the safety concerns, performance metrics such as arrival times of the moving objects at their destinations or sub-tracks in their routes are important. In a TTCS, the controller is in charge of coordinating the moving objects by  rerouting/rescheduling them. 

\begin{figure}
	\centering
	\includegraphics[width=8cm,height=9cm,keepaspectratio]{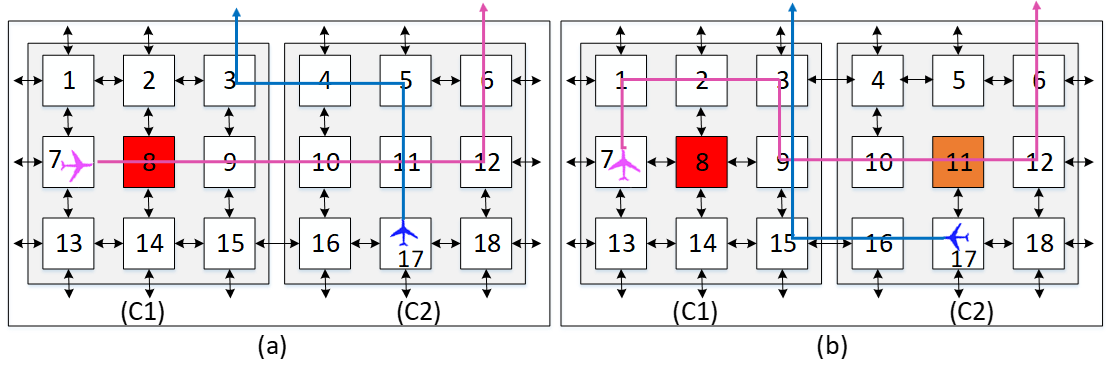}
	\caption{A TTCS with 18 sub-tracks. The effect of the change in sub-track 8 is propagated to the component $c_2$. To avoid the collision in sub-track 11, the blue moving object is rerouted.%, e.g. two aircraft arrive at sub-track 11 at the same time. %The colored dotted lines are the new parts of the moving objects routes which are selected by a rerouting algorithm and the solid lines are their initial routes. The black dotted lines, marked with an aircraft symbol, show the forward and backward change propagations (circular dependency between two components which is explained in Section.??). 
	}
	\label{ttcs}
\end{figure}

%\noindent\textbf{Example.} 
\begin{example}
An example of the change propagation is described for a TTCS as follows. Assume Fig.~\ref{ttcs}(a) and Fig.~\ref{ttcs}(b) show a TTCS with two control areas ($C1,C2$), where each area has  nine sub-tracks. The traffic flows from the west to the east and vice versa. %(eastbound traffic) and from the east to the west (westbound traffic). 
Each moving object of the eastbound traffic is able to travel towards a sub-track in the north, south, and east. The initial routes of the moving objects are shown in Fig.~\ref{ttcs}(a). The moving object with an unavailable sub-track in its route is rerouted and its new route is shown in Fig.~\ref{ttcs}(b). The red sub-track is an unavailable sub-track through which no moving object can travel. For instance, if a storm happens in a part of the airspace in an ATC, the aircraft cannot cross over the sub-tracks affected by the storm and are rerouted. Suppose that the traveling times of the moving objects through each sub-track are the same and are equal to one. %The initial traveling routes and time schedule decisions of  
The initial traveling planes of the purple and blue moving objects in Fig.~\ref{ttcs}(a) are $\{(0,7),(1,8),(2,9),(3,10),(4,11),(5,12),(6,6)\}$ and $\{(5,17),(6,11),(7,5),(8,4),(9,3)\}$, respectively. The first entry of each tuple shows the arrival time of the moving object at the sub-track mentioned in the second entry. For instance,  two subsequent tuples $(0,7),(1,8)$ mean that the purple moving object arrives at sub-track 7 at time zero and arrives at sub-track 8 at time 1 (which is the same time that it exits sub-track 7).

Suppose that a change happens to sub-track 8 and it becomes unavailable. %At this moment, 
As a consequence, the traveling plan of the purple moving object is changed to $\{(0,7),(1,1),(2,2),(3,3),(4,9),(5,10),(6,11),(7,12),(8,6)\\\}$, shown in Fig.~\ref{ttcs}(b). With the new plan,  the purple moving object enters into sub-track 10 (next area) at time 5 instead of 3, and this way the change propagates from $C1$ to $C2$. Now, the purple moving object arrives at sub-track 11 at time 6. At this time, the blue moving object has to enter into sub-track 11 based on its initial traveling plan. To prevent the collision between two moving objects, the controller employs a rerouting algorithm (adaptation policy) and changes the plan of the blue moving object to $\{(5,17),(6,16),(7,15),(8,9),(9,3)\}$.  As can be seen, by the occurrence of  a change, e.g. a storm, a sequence of changes happens, e.g. rerouting a set of moving objects. This example also shows a situation in which the change circulates between two areas. Based on the new traveling plan obtained for the blue moving object, it enters into $C1$ at time 7 instead of 9, and this way the change propagates back to $C1$.
\end{example}

%The track-based structure of ATC is followed in many applications, i.e. the rail traffic control systems, maritime transportation, smart hubs, unmanned vehicles, and centralized robotic systems with a mesh structure. In these systems, to reduce the risk of collision between the moving objects, the traveling space is divided into smaller safe regions, and a centralized controller manages the traffic flow. In \cite{16_CoordinatedActorsForReliableSelfAdaptiveSystems}, we introduced these systems as TTCSs, where the small safe regions are called tracks. Each track is divided into several sub-tracks.  
%The change propagation is a common phenomenon among TTCSs. 
As a change in the  context of a TTCS and its consequent adaptations in the system happen at runtime, the satisfaction of properties of interest should be checked at runtime. The properties include: the moving objects have to arrive at their destinations at the pre-specified times, the collision of the moving objects should be avoided, the fuel of  the moving objects should not be less than a threshold, and the system should be deadlock-free.
%
%For instance, how is the separation between two moving objects guaranteed (safety property)? Regarding the designed adaptation (rerouting algorithm), does a moving object arrive at a sub-track in its route at an expected time  (qualitative property)? Can the controller design an adaptation plan (select a rerouting algorithm among its algorithms) to satisfy the given properties (synthesis)? 
These properties are checked by verification. %However, TTCSs are large scale and verifying given properties at runtime faces state space explosion. The approach of this paper is our first step toward addressing this problem using Magnifier. 

%\COM{This part is only for describing the problem not clarifying different parts of the solution. It says that change propagates and for ... we verify. Our solution is to adapt when the change propagates... So, the adaptation of an area is triggered the moment the change propagates to that area..Somewhere we have to say that it can compose two components if it cannot find a local solution to the problem and efficiently find the new routes}

%% file: 03-Preliminaries.tex
\section{Background: Timed Input-Output Automata} \label{sec::background}
%In this section, we present the background to our research by giving an overview of a TIOA, the coordinated actor model, and the Ptolemy II Framework. 
%\subsection{Timed Input-Output Automata} \label{section::pre}
In this section, we briefly recall the definitions of a TIOA,  and the parallel product of several TIOAs. We also recall the definition of a deadlock state in a TIOA that is used to define the compatibility of two TIOAs in Section~\ref{section::magnifier}.  %Furthermore, we recall the definition of composable TIOAs that explains the necessary conditions to drive the parallel product of TIOAs.
\begin{comment}
	
 Finally, we introduce the coordinated actor model and multiple interactive coordinated actor models 
%\COM{the coordinated actors and the coordinated components }
as the high-level models, whose
semantics will be defined based on TIOAs.
%\subsection{Timed Input Output Automata}
content...
\end{comment}

A timed automaton with a set of input actions and a set of  output actions is called a TIOA. A TIOA  with integer variables \cite{zbrzezny2007sat} is defined as follows. 

\begin{definition} (TIOA) A Timed Input-Output Automaton is a tuple $\mathit{TA}=(Q, q_0,\mathit{Var},\mathit{Clk},\mathit{Act_{\mathit{in}}},\mathit{Act_{out}},T,I)$ where $Q$ is a finite set of locations, $q_0 \in Q$ is the initial location, $\mathit{Var}$ is the set of integer variables, $\mathit{Clk}$ is a finite set of clocks, $\mathit{Act_{\mathit{in}}}$ is a set of input actions, $\mathit{Act_{out}}$ is a set of output actions, $T \in Q \times (\mathcal{B}(\mathit{Clk}) \cup \mathcal{B}(\mathit{Var})) \times (\mathit{Act_{\mathit{in}}} \cup \mathit{Act_{out}} \cup \{\tau\})  \times 2 ^ \mathit{Clk} \times 2 ^\mathit{Ass} \times Q $ is a set of edges, and $I$ is an invariant-assignment function. Let $\# \in \{ \leq, <, =, \geq, >\}$ and $c \in \mathbb{N}$. The sets of conjunctions of constraints of the form $x \# c$ or $x-y \# c$ for $x ,y \in \mathit{Clk}$, and $v \# c$ or $v-w \# c$ for $v ,w \in \mathit{Var}$ are respectively denoted by $\mathcal{B}(\mathit{Clk} )$ and $\mathcal{B}(\mathit{Var} )$. The set of all variable assignments is denoted by $\mathit{Ass}$. The function $I: Q \rightarrow \mathcal{B}(\mathit{Clk})$ assigns invariants to locations. \qed
\end{definition}
Based on the above definition, the edge $e=(q,\psi,a, r,u,q') \in T$, besides action $a$, is labeled with a guard $\psi$, a sequence $u$ of assignments, and a set $r$ of clocks. Let $v_{C},v'_{C}:\mathit{Clk} \rightarrow \R$ and $v_{V},v'_{V}: \mathit{Var} \rightarrow \Z$ be clock and variable valuations, respectively. A state of the system modeled by a TIOA is in the form of $(q,v_{C},v_{V})$.  There is a discrete transition $(q,v_{C},v_{V}) \xrightarrow{a} (q',v'_{C},v'_{V})$ for an edge $e=(q,\psi,a, r,u,q')$ such that $v_{C}$ and $v_{V}$ satisfy $\psi$, $v'_{C}$ is reached by resetting the clocks in  the set $r$ to zero, and $v'_{V}$ is obtained as a subset of variables are set to their new values in the assignment set $u$. The clocks and variables not mentioned in $r$ and $u$ remain unchanged. Furthermore, $v'_{C}$ satisfies $I(q')$. The TIOA can stay in the location $q$ as long as the invariant $I(q)$ is valid. Let for $x \in \mathit{Clk}$ and $d \in \R$, $(v_{C}+d)(x)=y+d$ iff $v_{C}(x)=y$. For each delay $d \in \R$ there is a timed transition $(q,v_{C},v_{V}) \xrightarrow{d} (q,v_{C}+d,v_{V})$ such that $v_{C}+d$ satisfies $I(q)$.   %It performs action $a$, and moves from the state $q$ to the state $q'$ the moment the guard $\psi$ holds. By this transition, the clocks in the set $r$ are reset to zero and a subset of variables are set to their new values in the assignment set $u$. The state of the system modeled by a TIOA is $(q,v_{\mathit{Clk}},v_{\mathit{vars}})$, where $q\in Q$, $v_{\mathit{Clk}}:\mathit{Clk} \rightarrow \R$ and $v_{\mathit{vars}}: \mathit{vars} \rightarrow \Z$ are clock and variable valuations, respectively. There is a discrete transition $(q,v_{\mathit{Clk}},v_{\mathit{vars}}) \xrightarrow{a} (q',v'_{\mathit{Clk}},v_{\mathit{vars}})$
A state of the system can be a deadlock state that, based on \cite{10.1007/3-540-48778-6_18}, is a state from which no outgoing discrete transition is enabled, even after letting time progress.
\begin{definition}
	(Deadlock State) A state $s$ is a deadlock state if there is no delay $d \in \R$ and action $a \in (\mathit{Act_{\mathit{in}}} \cup \mathit{Act_{out}} \cup \{\tau\})$ such that $s \xrightarrow{d} s' \xrightarrow{a} s''$. \qed
\end{definition}
%Let $\mathit{TA_1}=(Q_1, q_{0_1},\mathit{Var_1},\mathit{Clk_1},\mathit{Act_{{\mathit{in}_1}}},\mathit{Act_{out_1}},T_1,I_1)$ and $\mathit{TA_2}=(Q_2, q_{0_2},\mathit{Var_2},\mathit{Clk_2},\mathit{Act_{{\mathit{in}_2}}},\mathit{Act_{{out}_2}},T_2,I_2)$ be two TIOAs. %, and $\mathit{shared(TA_1,TA_2)}=(\mathit{Act_{\mathit{in}_1}} \cap \mathit{Act_{\mathit{out}_2}}) \cup (\mathit{Act_{\mathit{in}_2}} \cap \mathit{Act_{\mathit{out}_1}})$ be the actions shared between $\mathit{TA_1}$ and $\mathit{TA_2}$. 
%Based on the definition taken from \cite{10.1007/3-540-45828-X_9},  two TIOAs are  composable if they do not have shared output actions or shared clocks.%their output actions and their clocks are disjoint. %Based on the definition taken from  \cite{Krichen:2009:CTR:1541661.1541692,kaynar2010theory}, two composable TIOAs are defined as follows.

%\begin{definition} \label{com::prec}
%	(Composable TIOAs) Two TIOAs $\mathit{TA_1}$ and $\mathit{TA_2}$ are composable if $\mathit{Clk_1} \cap \mathit{Clk_2}=\emptyset$ and $\mathit{Act_{\mathit{out}_1}} \cap \mathit{Act_{\mathit{out}_2}}=\emptyset$. \qed
%$\mathit{Act_{\mathit{out}_1}} \setminus \mathit{shared(TA_1,TA_2)}$, $\mathit{Act_{\mathit{in}_2}} \setminus \mathit{shared(TA_1,TA_2)}$, and $\mathit{Act_{\mathit{out}_2}} \setminus \mathit{shared(TA_1,TA_2)}$ are pairwise disjoint. \qed
%\end{definition}
 
Consider the network $\mathcal{N}=\{\mathit{TA_i} | i=1,\cdots,n\}$ of TIOAs, where TIOAs run in parallel and communicate through shared variables. Furthermore, TIOAs synchronize over time and shared common actions. %Several TIOAs are composable, if each pair of them are composable. 
We assume that when two edges (transitions) of two TIOAs synchronize over an action, their variables are updated by first executing the variable assignments of the output transition, and then by executing the variable assignments of the input transition. We also assume that the input transitions do not update the shared variables. Let $\mathit{shared(\mathcal{N})}={\bigcup\limits_{i,j=1}^{n}}_{j\neq i}(\mathit{Act_{in_i}} \cap \mathit{Act_{out_j}})$ be the set of actions shared between two or more TIOAs in the network $\mathcal{N}$. In some cases, we use $a!$ and $a?$ instead of $a$ to label an output and an input transition, respectively. For a set $L=\{l_1,\cdots,l_m\}$, $l_i < l_{i+1}$, we define $\bigsqcup\limits_{l \in L}u_l$ as a sequence of variable assignments $u_{l_1} \cdots u_{l_m}$. Based on \cite{zbrzezny2007sat}, the parallel product of TIOAs in the network $\mathcal{N}$ is defined as follows.

\begin{definition}
	(Parallel Product) Let TIOAs $\mathit{TA_i}, i=1,\cdots,n$, do not have shared output actions or shared clocks. %be pairwise composable. 
	The parallel product of $\mathit{TA_1},\cdots,\mathit{TA_n}$, denoted by $\mathit{TA_1}\otimes\cdots\otimes\mathit{TA_n}$, is $\mathit{TA}=(Q, q_{0},\mathit{Var},\mathit{Clk},\mathit{Act_{\mathit{in}}},\mathit{Act_{out}},T,I)$, where
\begin{gather*}
	Q=Q_1 \times \cdots \times Q_n, q_0=(q_{0_1}, \cdots, q_{0_n}),\\ \mathit{Var}= \bigcup\limits_{i=1}^{n} \mathit{Var_i},
	\mathit{Clk}= \bigcup\limits_{i=1}^{n} \mathit{Clk_i},\\ \mathit{Act_{\mathit{in}}}=\bigcup\limits_{i=1}^{n} \mathit{Act_{in_i}}\setminus \mathit{shared(\mathcal{N})},\\
	\mathit{Act_{\mathit{out}}}=\bigcup\limits_{i=1}^{n} \mathit{Act_{out_i}}\setminus \mathit{shared(\mathcal{N})},\\I((q_1,\cdots,q_n))=\bigwedge\limits_{i=1}^{n} I_i(q_i),
\end{gather*}
	 and $T$ is defined in the following way:

	 \begin{itemize}
	 	\item  $((q_1,\cdots,q_n),\psi,a, r,u,(q'_1,\cdots,q'_n)) \in T$ iff
	 	
	 	there exists $i \in \{1,\cdots,n\}$ such that $a \in \mathit{Act_{in_i}} \cup \mathit{Act_{out_i}} \setminus \mathit{shared(\mathcal{N})}$, $(q_i,\psi,a, r,u,q'_i) \in T_i$, and for all $j \in \{1,\cdots,n\} \setminus \{i\}$,  $q'_j=q_j$ holds;
	 \end{itemize}
 	 
 \begin{itemize}
 		\item $((q_1,\cdots,q_n),\psi,\tau, r,u,(q'_1,\cdots,q'_n)) \in T$ iff
 		
 		there exists $a \in \mathit{shared(\mathcal{N})}$ and $i \in L$ for $L=\{k|(q_k,\psi_k,a, r_k,u_k,q'_k) \in T_k\}$ such that $(q_i,\psi_i,a!, r_i,u_i,q'_i) \in T_i$ and for all $j \in L \setminus \{i\}$, $(q_j,\psi_j,a?, r_j,u_j,q'_j) \in T_j$, and $\psi=\bigwedge\limits_{k \in L}\psi_k$, $r=\bigcup\limits_{k \in L} r_i$, $u=u_i\bigsqcup\limits_{k \in L\setminus \{i\}}u_k$, and for all $k \in \{1, \cdots, n\} \setminus L,  q'_j=q_j$ holds. \qed 
 \end{itemize}
 
	 \begin{comment}
	 
	\begin{itemize}
		\item For $(q_1,\cdots,q_n)\in Q_1\times\cdots\times Q_n$, $a \in \mathit{Act_{in_i}} \cup \mathit{Act_{o_i}} \cup \{\tau\} \setminus \mathit{shared(\mathcal{N})}: (q_i,\psi_i,a, r_i,u_i,q'_i) \in T_i \Rightarrow ((q_1,\cdots,q_n),\psi_i,a, r_i,u_i,(q'_1,\cdots,q'_n)) \in T \wedge \forall j=1,\cdots,n \cdot j\neq i, q'_j=q_j$;
		
		\item For $(q_1,\cdots,q_n)\in Q_1\times\cdots\times Q_n$, $a \in \mathit{shared(\mathcal{N})}$, $L=\{k|(q_k,\psi_k,a?, r_k,u_k,q'_k) \in T_k\}$: $(q_i,\psi_i,a!, r_i,u_i,q'_i) \in T_i \wedge L \neq \emptyset  \Rightarrow ((q_1,\cdots,q_n),\bigwedge\limits_{k \in L \cup \{i\}}\psi_k,\tau, \bigcup\limits_{k \in L \cup \{i\}} r_i,u_i\bigsqcup\limits_{k \in L}u_k,\\(q'_1,\cdots,q'_n)) \in T \wedge \forall j=1,\cdots,n \cdot j\notin L, q'_j=q_j$. \qed
	\end{itemize}
	\end{comment}
\end{definition} 

Note that the state of the system modeled by a network of TIOAs is obtained by clock values, values of all variables, and the locations of all TIOAs in the network. In the rest of the paper, we benefit from the syntax of the \textsc{Uppaal} modeling language \cite{behrmann2006tutorial} to use functions as macros for  expressions in  guards and updates in TIOAs. %The same as the \textsc{Uppaal} modeling language \cite{behrmann2006tutorial} we use functions to define invariants, guards, and variable assignments in the rest of this paper. %An \textsc{Uppaal} model is usually composed of several TIOAs executed in parallel.  The \textsc{Uppaal} modeling language extends TIOAs with a set of features such as discrete variables in order to provide a more natural description of the system using Timed Automata.  The shared variables in \textsc{Uppaal} are defined in global declarations. The same as the \textsc{Uppaal} tool, 
Furthermore, we use different colors such as purple, green, light blue, and dark blue to respectively distinguish invariants, guards,  synchronization actions, and clock reset and variable assignments in the figures related to TIOAs. %use different colors such as purple, green, light blue, and dark blue to respectively distinguish invariants, guards,  synchronization actions, and clock reset and variable assignments. 

%% file: 04-MultipleCOODA.tex
\section{Formal Compositional Semantics of CoodAA} \label{section::sem}
In this section, we first review the coordinated adaptive actor model. We then provide a formal  compositional  semantics for  CoodAA %by focusing on track-based systems.  We  present the semantics 
in terms of TIOAs. 

\subsection{Background: Coordinated Adaptive Actor Model} \label{subsec::cooda}
%\COM{One of the reviewer comments is to describe MAPE-K loop such that every body understand it... }
%The coordinated actor model encapsulates a set of actors and a coordinator.  As shown in Fig.~\ref{fig::architecture},  the coordinated actor model is designed based on the MAPE-K feedback loop \cite{Kephart:2003:VAC:642194.642200}. 

%We introduced the coordinated adaptive actor model (CoodAA) in \cite{16_CoordinatedActorsForReliableSelfAdaptiveSystems}. This model  is composed of a set of components. Each component encapsulates  a set of actors and a set of coordination policies that are described in a  coordinator. CoodAA is designed based on the MAPE-K feedback loop to realize a self-adaptive system.
We introduced the coordinated adaptive actor model in \cite{16_CoordinatedActorsForReliableSelfAdaptiveSystems}. CoodAA consists of  a set of coordination policies, which are described in a coordinator, and a set of actors. Actors, as units of computations, communicate via message passing, and   
can be categorized into a set of components. Each component, shown in Fig.~\ref{fig::architecture}, besides actors has its own coordinator that describes the coordination policies relevant to the actors of the component.   %The structure of a component is shown in Fig.~\ref{fig::architecture}.
   A component can be a nested component, meaning that it can consist of several components.  %the actors of the component are categorized into a set of components.
   The structure of each component (and as a consequence  CoodAA itself) conforms to the MAPE-K feedback loop. %The coordinated adaptive actor model is designed based on the MAPE-K feedback loop to realize a self-adaptive system.

 The managed system in a self-adaptive system  is  controlled by a feedback control loop. A well-known approach to realize a self-adaptive system is by means of the MAPE-K feedback control loop \cite{Kephart:2003:VAC:642194.642200}. This loop consists of \textit{Monitor}, \textit{Analyze}, \textit{Plan}, and \textit{Execute} components together with the \textit{Knowledge base}. The model@runtime, as an abstraction of the system and its environment, is kept in the \textit{Knowledge base}. The \textit{Monitor} component monitors the system and its environment, and updates the model@runtime. In the case of detecting a change, the \textit{Analyze} component analyzes the model@runtime. The analysis results are given to the \textit{Plan} component, and the \textit{Plan} component makes an adaptation plan.  The adaptation plan is  applied to the model@runtime and  the model@runtime is again analyzed. If the requirements of the system are satisfied, the adaptation plan is sent to the system through the \textit{Execute} component. Otherwise, the \textit{Plan} component makes another adaptation plan.  

%\COM{Marjan: The figure and the explanation is confusing.}

As shown in Fig.~\ref{fig::architecture}, %each component of CoodAA (and as a consequence  CoodAA itself) is mapped to a MAPE-K feedback loop. T
the actors of the component construct the model@runtime and the coordinator consists of the Analyze and Plan activities  \cite{16_CoordinatedActorsForReliableSelfAdaptiveSystems}. %The coordinator is able to analyze the model@runtime by executing it to investigate the future behavior of the system. The coordinator dispatches messages among actors to execute the actor-based model@runtime. 
The coordinator dispatches messages among actors to execute the actor-based model@runtime and analyze it by investigating the future behavior of the system.

\begin{figure}
	\centering
	\fbox{\includegraphics[width=8cm,keepaspectratio]{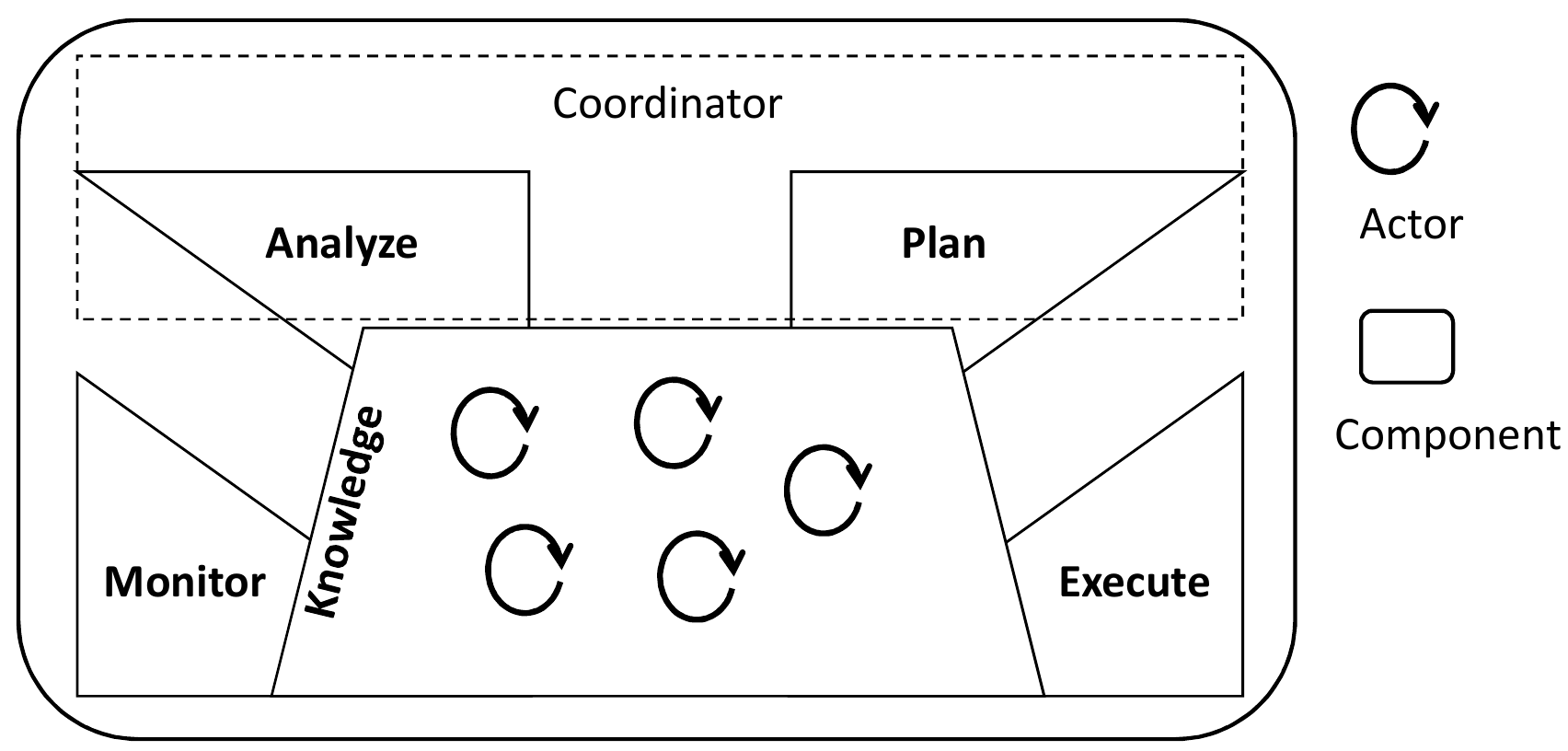}}
	\caption{The mapping between a component and the  MAPE-K feedback loop. An actor-based model@runtime is kept in the \textit{Knowledge base}.
	%\marjan{I am not sure what this figure shows.}
	%For our compositional analysis we abstract this state away. The edges and locations related to error handling or applying adaptations are removed from part (a), and the resulting TIOA is shown in part (b)
	}
	\label{fig::architecture}%
\end{figure}

CoodAA  is aligned with the structure of track-based systems, each sub-track is modeled by an actor,  the controller is modeled by a coordinator, and the  moving objects are modeled as messages passing among the actors \cite{cooda}. Each message carries information such as the identifier of a moving object, its traveling plan (route and time schedule decisions), speed, fuel, etc. The  coordinator is able to adapt the system by rerouting/rescheduling the moving objects considering the congestion and environmental conditions. Upon occurring a change, the model@runtime is updated based on a snapshot taken from the system.  Then, the coordinator obtains new  routing/scheduling plans. 
The coordinator operates either in the regular phase or in the adaptation phase. %The concern of the coordinator in the regular phase is dispatching the messages. 
In the regular phase, the coordinator dispatches the messages, and the routes/schedules are given in the messages passed to the actors. 
 In the adaptation phase, the coordinator makes decisions to adapt the system.
 %by rerouting/rescheduling the moving objects. 
 After the decision making, the new routes/schedules are passed to the actors while the coordinator moves to the regular phase. 
 The messages are passed between the actors based on the 
 %static 
 plans given to the actors.
 %Then, moving objects travel based on their static plans given to the actors.

%\COM{ we were supposed that to say that we implemented different adaptation policies in our previous work, and therefore, adaptation is not a concern in this work, because the reviewers have questions about the adaptation policies, and how you know which one is optimum and .. }
In \cite{cooda}, the coordinator is augmented with different rerouting/rescheduling algorithms. The coordinator is able to predict the behavior of the system through executing the model@runtime,  measure several metrics using simulation, and based on the calculated results,  select the best algorithm for rerouting/rescheduling purpose. The analysis in \cite{cooda} is based on simulation not verification.

%\COM{ADD A Paragraph on PobSAM and Rebeca and Timed Rebeca}
CoodAA is initially inspired from
an adaptive actor-based framework  proposed by Khakpour et al.~in \cite{KHAKPOUR20123}. The framework is called PobSAM (Policy-based Self-Adaptive Model) and is an integration of algebraic formalisms
and actor-based Rebeca models. A hierarchical extension of PobSAM is proposed by Khakpour et al.~in
\cite{KHAKPOUR20122770}.
The Reactive Object Language, Rebeca \cite{DBLP:journals/fuin/SirjaniMSB04}, is an actor-based \cite{hew72,Agha:1986:AMC:7929} modeling language. 
%supported by a model checking tool Afra \cite{Afra}.
Rebeca is used for modeling and formal verification of concurrent and distributed systems.
The model of computation in  Rebeca is event-driven and the communication is asynchronous.
Timed Rebeca \cite{DBLP:conf/facs2/KhamespanahSVK15,DBLP:journals/scp/KhamespanahSSKI15} extends  Rebeca
to model the timing features.
In PobSAM there is no explicit notion of a coordinator, no timing constraint, and no focus on verification at runtime.

%It also can be augmented with several rerouting/rescheduling algorithms, and by predicting the behavior of the system through executing the model$@$runtime (analyze), selects the best algorithm for rerouting/rescheduling purpose (plan) \cite{16_CoordinatedActorsForReliableSelfAdaptiveSystems}.

%Beside dispatching the messages through actors, the coordination is about rerouting/rescheduling of the moving objects. The coordinator operates either in the regular phase or in the adaptation phase. The concern of the coordination in the regular phase is dispatching the messages. The dispatching is about routing/scheduling of the moving objects where  routes/schedules are given in plans passed to the actors. In the adaptation phase, the coordinator makes decisions to adapt the system by rerouting/rescheduling the moving objects. After the decision making, the new plans are passed to the actors while the coordinator moves to the regular phase. Then, actors work based on their static plans. 

%\input{041-Syntax}
%\subsection{Abstract Syntax of the Coordinated Actor Model}
\subsection{Summary of Definitions}

\begin{table*}
	\begin{center}
		\begin{tabular}{|m{\textwidth}|} 
			\hline
			\begin{itemize}
				
				\item $a_i$ is an actor with the unique identifier $i$, defined as $a_i=(\mathit{status_i}, \mathit{movingPlan_i},\mathit{Mtds},\mathit{Change},\mathit{P_{I_i}},$
				$ \mathit{P_{O_i}})$ 
				\begin{itemize}
					\item $\mathit{status_i}$ with a value of $\{\mathit{Free, Occupied, Error/Adapt}\}$ denotes the state of the actor
					\item $\mathit{movingPlan_i}$ stores the traveling plan  given to the actor by a message
					%\item $\mathit{Mtds}=\{\mathit{interact_{i,j}(\langle transP \rangle)}\{\langle body \rangle\}|j=1,\cdots,\mathit{directions}\}$	is the set of message handlers of the actor $a_i$, where \textit{direction} is the number of arrival and departure directions for a sub-track, \COM{$\langle transP \rangle=\langle objectId,travelPlan \rangle$, and $\langle body \rangle$ is a sequence of statements}
					\item $\mathit{Mtds}= \mathit{interact_{i,j}(transP) \{body\}}$ is the set of message handlers of the actor $a_i$,  $ transP=(objectId,travelPlan)$,  $ body$ is a sequence of statements, and  $j=1,\cdots,\mathit{directions}$ %} 
						%\marjan{(Note: \textit{direction} is the number of arrival and departure directions for a sub-track)}
					%}
					\item $\mathit{Change} = \mathit{change_i(args)}\{body\}$ is a special message handler where the sequence of its input arguments is $args= ( \mathit{errorORadapt},\mathit{adapted},\mathit{error},\mathit{errorOver}, \mathit{newPlan})$
					\item  $\mathit{P_{I_i}}=\{p_{I_{i,j}}|j=1,\cdots,\mathit{directions}\}\cup \{\mathit{change_{i}}\}$ is the set of input ports of the actor $a_i$
					\item $\mathit{P_{O_i}}=\{p_{O_{i,j}}|j=1,\cdots,\mathit{directions}\}$ is the set of output ports of the actor $a_i$
				\end{itemize}
				\item $\mathit{ch_i}$ is a channel with the unique identifier $i$,  defined as $\mathit{ch_i}=(s_i,d_i)$
				\begin{itemize}
					\item $s_i$ is the source end of the channel $\mathit{ch_i}$
					\item $d_i$ is the sink (destination) end of the channel $\mathit{ch_i}$
					%\item $\mathit{msgs_i}$ is a sequence of messages in the channel $\mathit{ch_i}$
				\end{itemize}
				\item 	$C_i$ is a component with the unique identifier $i$, defined as $C_i=(A_i,\mathit{CH_{i},\mathit{B_i}, \mathit{F_i} })$
				\begin{itemize}
					\item $A_i=\parallel_{j \in I_{C_i}} a_j$ is a set of actors $\{a_j |j \in I_{C_j}\}$ concurrently executing, where $I_{C_i}$ is the set of identifiers of the actors belonging to $C_i$
					\item $\mathit{CH_{i}}$ is the set of all channels of $C_i$
					\item $\mathit{F_i}:A_i \rightarrow A_i$ is the coordination function of $C_i$
					
					\item $\mathit{B_{i}}:P_i \rightharpoonup \mathit{CH_i}$ is the binding function of $C_i$, where $P_i=\bigcup_{a_j \in A_i}(P_{I_j} \cup P_{O_j})$ is a set of  ports of all actors of $C_i$
				\end{itemize}
				\item 	$C_k=C_i \parallel C_j$ is a composite component with the unique identifier $k$, defined as $C_k=(\mathit{A}_k , \mathit{CH_{k}},\mathit{{B_k}}, \mathit{F_k})$ 
				\begin{itemize}
					\item $A_k=A_i \cup A_j$ is a set of actors concurrently executing
					\item $\mathit{CH_k}=\mathit{CH_{i}} \cup \mathit{CH_{j}} \cup \mathit{NewCH}$ is the set of channels of $C_k$, where $\mathit{NewCH}$ is a set of channels to connect boundary output ports of a component to boundary input ports of the other component. The sets of boundary input and output  ports of the component $C_i$ are respectively defined as $P_{I_{C_i}}=\{p|a_l \in A_i \wedge p \in P_{I_l} \wedge  \nexists \mathit{ch} \in \mathit{CH_i}, (p,\mathit{ch}) \in B_i\}$ and $P_{O_{C_i}}=\{p| a_l \in A_i \wedge p \in P_{O_l} \wedge  \nexists \mathit{ch} \in \mathit{CH_i}, (p,\mathit{ch}) \in B_i\}$
					\item $\mathit{F_k}:A_k \rightarrow  A_k$ is the coordination function of $C_k$
					
					\item $\mathit{B_{k}}=B_i \cup B_j \cup NewB$ is the binding function of $C_k$, where $\mathit{NewB}:P_k \rightharpoonup \mathit{NewCH}$ and $P_k$ is the set of all boundary input and output ports of the components  $C_i$ and $C_j$.%$\mathit{B_{k}}:P_k \rightharpoonup \mathit{CH_k}$ is the binding function of $C_k$, where $P_k=\{p|  a_l \in (A_i \cup A_j) \wedge p \in (P_{I_l} \vee P_{O_l}) \wedge \nexists \mathit{ch} \in (\mathit{CH_i}\cup \mathit{CH_j}), (p,\mathit{ch}) \in (B_i \cup B_j)\}$ is the set of all boundary ports of the components  $C_i$ and $C_j$
				\end{itemize}
				
				\item A coordinated adaptive actor model is a composite component, which is a component itself.%$\mathcal{M}$ is the coordinated actor model, defined as the component $\mathcal{M}=(S, \mathit{CH},B, F)$, where $S=\{C_i| i \in \CAMID\}$ and $\CAMID$ is the set of all component identifiers
				%\item $\mathit{msg}$ is a message in a model of a track-based system,  defined as $\mathit{msg}=(\mathit{objectId},\mathit{route},\mathit{schedule},\mathit{fuel},\mathit{speed})$
				%\item $\MSG$ is the set of messages communicated in the model
				
				%	\item $\mathit{msg} \in \MSG$ in a model of a track-based system is defined as $\mathit{msg}=(\mathit{id},\mathit{route},\mathit{schedule},\mathit{fuel},\mathit{speed})$
				%\item $\MSG$ is the set of all messages, $\AID$ is the set of all actor identifiers, $\CAMID$ is the set of all coordinated actor model identifiers, and $\mathit{Chan}$ is the set of all channels
				
			\end{itemize}
			\\
			\hline
		\end{tabular}
	\end{center}
	\caption{Summary of definitions and notations for the coordinated adaptive actor model (CoodAA)
	%\COM{Shall we add the brief explanation of body here?}
	}
	\label{tab::summery}
\end{table*}

In this section, the basic elements of CoodAA, including actors, components, coordinators, and channels are formally described. The summary of definitions and notations are provided in Table~\ref{tab::summery}. %The coordinated adaptive actor model of a track-based system is composed of a set of components. Each component contains a set of actors, a coordinator, and a set of channels (). 
An actor has a variable  \textit{status} that shows the status of the actor being free, occupied or in an adaptation phase. It also has a variable to keep the plan (\textit{movingPlan}); the plan specifies the direction and the time to send out the message.
	An actor has several input and output ports (modeling several directions that a moving object can arrive at or depart from a sub-track).
	Input and output ports  are  communication interfaces of the actor with other actors. 
	An actor has a set of message handlers and each message handler corresponds to an input port of the actor.
	The actor can be informed about an adaptation or an environmental change by receiving a message  over a special input port, the \textit{change} port. 

%if the moving objects can arrive at the sub-track corresponding to the actor from several directions.
%Similarly, the actor has several output ports if the moving objects can depart form the sub-track from several directions. 
%The actor can be informed about an adaptation or an environmental change by receiving a message  over a special input port, the \textit{change} port.  %Besides input and output ports, a
%An actor has two variables and a set of message handlers; each message handler corresponds to an input port of the actor. 
%Let \textit{directions} be the number of arrival or departure directions for a sub-track, e.g, the sub-track has 4 arrival directions and 4 departure directions if \textit{directions=4}. %The main computation of the actor is defined in its $\mathit{handler}$ method.

%\COM{Make sure to keep the order of the components in a definition the same everywhere, in the table, here, ....}
\begin{definition} \label{def::actor}
	(Actor) An actor, $a_i$, with the unique identifier $i$, is defined as $(\mathit{status_i}, \mathit{movingPlan_i},
	\mathit{Mtds}_i,\mathit{change_i(args)}\{body\},
	\mathit{P_{I_i}},$ $\mathit{P_{O_i}})$,
	where $\mathit{status_i}$ having a value of $\{\mathit{Free, Occupied, Error/Adapt}\}$ denotes the state of the actor, $\mathit{movingPlan_i}$ stores the traveling plan of the moving object,  $\mathit{P_{I_i}}=\{p_{I_{i,j}}|j=1,\cdots,directions\}\cup \{\mathit{change_{i}}\}$ and $\mathit{P_{O_i}}=\{p_{O_{i,j}}|j=1,\cdots,directions\}$ are respectively the sets of input and output ports of the actor,  $\mathit{Mtds}_i=\{\mathit{interact_{i,j}(transP)}\{body\}|j=1,\cdots,directions\}$ is the set of message handlers of the actor with the sequence of input arguments $ transP= (objectId,travelPlan)$,  $\mathit{change_i(args)}$ is a special message handler with the sequence of input arguments  $\mathit{args}=( \mathit{errorORadapt},\mathit{adapted},\mathit{error},\mathit{errorOver},\\ \mathit{newPlan})$, and $ body=\mathit{stm}^*$ as the body of a message handler is a sequence of statements. 
	\qed
\end{definition}

%\COM{you need to add the body of the mtd in the definition, you only have the header now. 	Wher you put the interact????}
The \textit{status} variable of the actor has the initial value of \textit{Free}, and the \textit{movingPlan} variable is null.  The main computation of the actor is performed in its message handlers. A statement in the body of a message handler can be an assignment statement ($\mathit{assign}$),  a conditional statement ($\mathit{cond}$), a send statement, or a delay statement, i.e. $\mathit{stm}::= \mathit{assign}| \mathit{cond} | \mathit{send(p, msg)} | \mathit{delay(t)}$,  where $\mathit{p}$ is an output port, $\mathit{msg}$ is a sequence of values, and $t$ is a time variable. Definitions of  assignment and conditional statements are like in regular programming languages. The actor can send a sequence of values as a message over an output port using the $\mathit{send}$ statement.  
The actor is also able to introduce a delay during the execution of its message handlers. The $\mathit{delay}(t)$ statement models the passage of $t$ units of time for the actor, where $t$ is derived from  the plan of the actor.

%\COM{We cannot say If there is a message in an input port of the actor.. because the input port does not have memory, I think  we can sey: }

Two actors are connected via a primitive medium that is called a channel. A channel has a source and a sink (destination) end. Each port of an actor can be connected to an end of at most one channel. 
%A channel contains a sequence of messages that are sent from an actor to another actor over the channel.  %, through which the actors communicate. 

\begin{definition} (Channel) 
	A channel, $\mathit{ch_i}$, with the unique identifier $i$, is defined as  $(s_i,d_i)$, where $s_i$ and $d_i$ are respectively the source and the sink of the channel. \qed
\end{definition}

The source of a channel is connected to at most one output port and the sink of a channel is connected to at most one input port. The bindings between channels and ports of the actors are defined through the binding function of the component. A component is defined as follows.

\begin{definition} (Component) A component, $C_i$, with the unique identifier $i$, is defined  as $C_i=(A_i, \mathit{CH_{i}}, \mathit{{B_i}}, \mathit{F_i},)$, where $A_i$ is the set of internal actors of $C_i$, $\mathit{CH_i}$ is the set of channels belonging to $C_i$, $\mathit{F_i}:A_i \rightarrow  A_i$ is the coordination function of $C_i$, and $\mathit{B_{i}}:P_i \rightharpoonup \mathit{CH_i}$ is the binding function of $C_i$, where $P_i=\bigcup_{a_j \in A_i}(P_{I_j} \cup P_{O_j})$ is a set of  ports of all actors of $C_i$. \qed
	\label{def::component}
\end{definition}

%\COM{MARJAN: Separate the story of TTCS and CoodAA. First explain TTCS and what happens in that, then explain what happens in actors and components. In between you can explain what is modeled by what.}
The coordinator of a component is defined in the form of a function. In the rest of the paper, we use the terms coordinator and coordination function interchangeably. The coordination function of $C_i$ is able to adapt the behaviors of  actors by putting  messages on the \textit{change} ports of the actors.  The binding function of $C_i$ defines the topology of the component by connecting ports of actors to ends of channels.

Each component has  sets of boundary input and output  ports through which the component communicates with other components. These sets are defined based on the sets of input and output ports of the constituent actors of the component. The input port  $p \in \bigcup_{a_l \in A_i}P_{I_l}$ is a boundary input port of the component $C_i$ if $p$ is not bound to any channels of the component $C_i$ using the binding function $B_i$. Similarly, the output port $p \in \bigcup_{a_l \in A_i}P_{O_l}$ is a boundary output  port of the component $C_i$ if $p$ is not bound to any channels of the component $C_i$ using the binding function $B_i$. Therefore, the sets of boundary input and output  ports of the component $C_i$ are respectively defined as $P_{I_{C_i}}=\{p| p \in P_{I_l} \wedge a_l \in A_i \wedge \nexists \mathit{ch} \in \mathit{CH_i}, (p,\mathit{ch}) \in B_i\}$ and $P_{O_{C_i}}=\{p| p \in P_{O_l} \wedge a_l \in A_i \wedge \nexists \mathit{ch} \in \mathit{CH_i}, (p,\mathit{ch}) \in B_i\}$.

The composition of two (or more) components forms a composite component that is defined as follows.

\begin{definition} 
	(Composite Component) The composite component $C_k=(\mathit{A}_k , \mathit{CH_{k}},\mathit{{B_k}}, \mathit{F_k})$ %with the unique identifier $k$, 
	is a component such that there exist two components $C_i=(A_i, \mathit{CH_{i}},\mathit{{B_i}},\mathit{F_i} )$ and $C_j=(A_j, \mathit{CH_{j}},\mathit{{B_j}}, \mathit{F_j})$ whose composition results in $C_K$, where $A_k=A_i \cup A_j$,  $\mathit{CH_k}=\mathit{CH_{i}} \cup \mathit{CH_{j}} \cup \mathit{NewCH}$, %is the set of channels of $C_k$
	 $\mathit{F_k}:A_k \rightarrow  A_k$, %is the coordination function of $C_k$
	 and $\mathit{B_{k}}=B_i \cup B_j \cup NewB$. %is the binding function of $C_k$
	 The set $\mathit{NewCH}$ is the set of channels that connect boundary output ports of a component to boundary input ports of the other component. The links between boundary ports and the new channels are defined using   $\mathit{NewB}:P_k \rightharpoonup \mathit{NewCH}$, where $P_k$ is the set of all boundary ports of the components  $C_i$ and $C_j$. The composition of $C_i$ and $C_j$ is denoted by $C_i \parallel C_j$.  \qed
%	(Composite Component) Let $C_i=(A_i, \mathit{CH_{i}},\mathit{{B_i}},\mathit{F_i} )$ and $C_j=(A_j, \mathit{CH_{j}},\mathit{{B_j}}, \mathit{F_j})$ be two components. The composition of $C_i$ and $C_j$, denoted by $C_i \parallel C_j$, results in the composite component $C_k=(\mathit{A}_k , \mathit{CH_{k}},\mathit{{B_k}}, \mathit{F_k})$ with the unique identifier $k$, where $A_k=A_i \cup A_j$,  $\mathit{CH_k}=\mathit{CH_{i}} \cup \mathit{CH_{j}} \cup \mathit{NewCH}$ is the set of channels of $C_k$, $\mathit{F_k}:A_k \rightarrow  A_k$ is the coordination function of $C_k$, and $\mathit{B_{k}}=B_i \cup B_j \cup NewB$ is the binding function of $C_k$. The set $\mathit{NewCH}$ is the set of channels that connect boundary output ports of a component to boundary input ports of the other component. The links between boundary ports and the new channels are defined using   $\mathit{NewB}:P_k \rightharpoonup \mathit{NewCH}$ , where $P_k$ is the set of all boundary ports of the components  $C_i$ and $C_j$. \qed
	\label{def::compC}
\end{definition}

Based on Definition.~\ref{def::compC}, each composite component is itself a component. The coordinated adaptive actor model is a component composed of all components of the model.

%The set $\mathit{CH_i}$ as the set of all channels of  $C_i$ is defined as $\mathit{CH_i}=\{(s_j,d_j,\mathit{msgs_j})| a_s \in A_i \vee a_d \in A_i \}$. The sets $\mathit{CH_{I_i}}$ and  $\mathit{CH_{O_i}}$ are defined as $\mathit{CH_{I_i}}=\{\mathit{ch_j}|\exists a_l \in A_i \cdot \mathit{ch_j} \in \mathit{Ch_{I_l}} \wedge \nexists a_{k} \in A_i \cdot \mathit{ch_j} \in \mathit{Ch_{O_k}}\}$, and $\mathit{CH_{O_i}}=\{\mathit{ch_j}|\exists a_l \in A_i \cdot \mathit{ch_j} \in \mathit{Ch_{O_l}} \wedge \nexists a_k \in A_i \cdot \mathit{ch_j} \in \mathit{Ch_{I_k}}\}$. In other words, the channel $\mathit{ch_j}$ is an input channel of $\mathit{C_i}$ if an actor of $\mathit{C_i}$ has $\mathit{ch_j}$ as its input channel and none of the actors of $\mathit{C_i}$ has $\mathit{ch_j}$ as its output channel. Similarly, the channel $\mathit{ch_j}$ is an output channel, if an actor  of $\mathit{C_i}$ has $\mathit{ch_j}$ as its output channel and none of the actors of $\mathit{C_i}$ has $\mathit{ch_j}$ as its input channel. %The set of channels in $\mathit{C_i}$ is the union of input and output channels of its constituent actors. %This set is denoted by $\CH_i$.

\subsection{Compositional Semantics of CoodAA}
In this section, we present the compositional semantics of the coordinated adaptive actor model of a track-based system. Each component is presented in the form of a network of TIOAs of actors and a coordinator, which are shown in  Fig.~\ref{fig::automata1}. 
Each actor is specified by a separate TIOA. There is no separate TIOA for a channel, since here we have zero-capacity  channels, where each channel connects two ports and   synchronises the communications of two actors.  
Fig.~\ref{fig::automata1}\protect\subref{fig::automata::abstract} shows an abstract view of an actor where changing the state between free, occupied and adaptation is clear. In Fig.~\ref{fig::automata1}\protect\subref{fig::automata::detailed} we show a more detailed view with more details on the interactions and transitions.
We present an abstract TIOA of the coordinator in Fig.~\ref{fig::automata1}\protect\subref{fig::automata::coordinator} which shows the  operations of the coordinator on  actors.
We include an abstract view of the coordinator to present  a complete formal semantics of CoodAA.
This abstract view is enough for our discussions in this paper because our focus in 
this paper is on the verification of compatibility of components. 
%that Magnifier performs in each iteration.
%
In Magnifier, the coordinator analyzes and adapts the model@runtime, and the verification of  compatibility of components  that is performed in each iteration is 
on  the model@runtime, which is a snapshot of the system and only consists of the actors.
The TIOA of the coordinator is not needed for that verification. 
\input{041-CAMSemantics}

%% file: 041-CAMSemantics.tex
\begin{figure*}
	\centering
	\subfloat[An abstract representation of TIOA of an actor $\mathit{a_j}$ of the component $\mathit{C_i}$. The arrival of the moving object into the sub-track is modeled by the transition from \emph{Free} to \emph{Occupied}. The sub-track remains occupied during the traveling time of the moving object, formulated in the \textit{Occupied} location, where $\mathit{clock}$ shows the time elapsed since the arrival of the moving object. The departure of the moving object from the sub-track is modeled by the transition from \emph{Occupied} to \emph{Free}. An actor goes into its  \emph{Error/Adapt} state  when a change happens to the sub-track itself and/or the plan of the moving object needs an adaptation.
	%represents that a change happens to the sub-track, and/or the plan of the moving object needs an adaptation.
	 %This figure is described with more details in Fig.~ \protect\subref{fig::automata::detailed}.
	]{{\fbox{\includegraphics[width=0.25\linewidth,keepaspectratio]{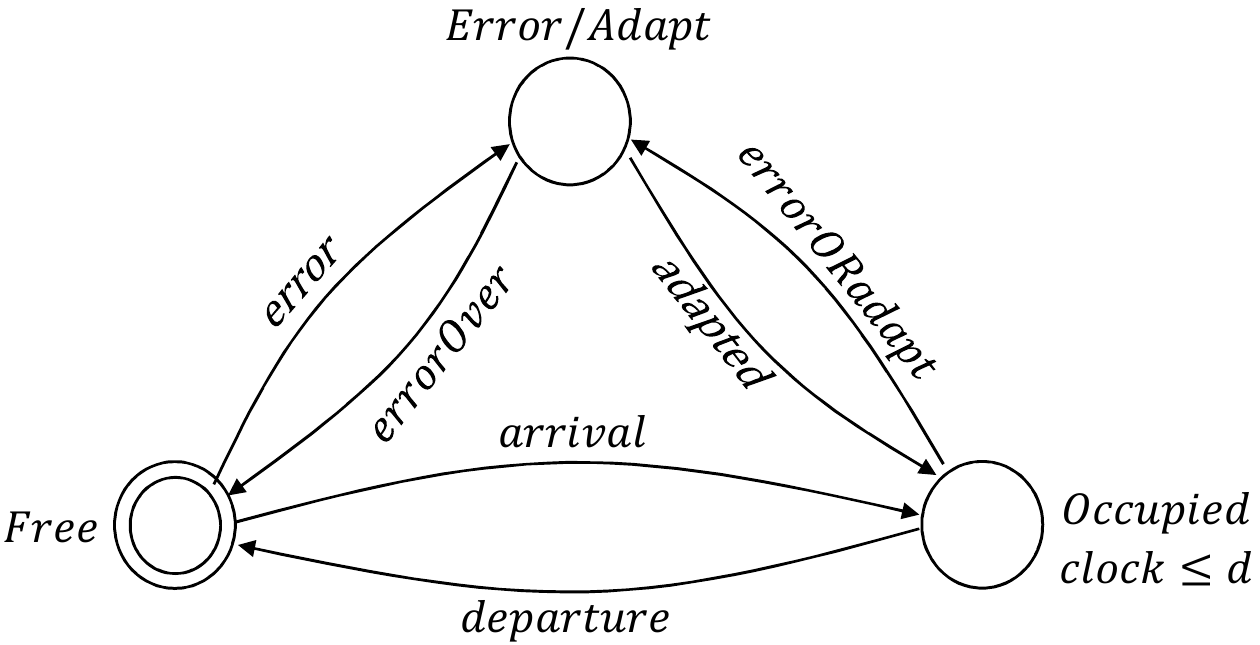}}\label{fig::automata::abstract}}}
	\quad
	\subfloat[The detailed model of TIOA of an actor that is shown in Fig.~\ref{fig::automata1}\protect\subref{fig::automata::abstract}. The automaton moves from \textit{Free} to \textit{Occupied} by synchronizing with $\mathit{interact_{ch}}?$, where \emph{ch} is an input channel of the actor. The variable \textit{movingPlan} is the variable of the actor. The traveling plan in \textit{movingPlan} is updated using the the \emph{update} function (the route of the moving object is updated by removing the first entry from the route).  %The sub-track remains occupied during the traveling time of the moving object, formulated in the \textit{Occupied} location, where $x$ shows the time elapsed since the arrival of the moving object, 
	The function \emph{travT} calculates the delay for the actor. The automaton moves from \textit{Occupied} to \textit{Free} by synchronizing with $\mathit{interact_\mathit{ch}}!$.  The \textit{transP} variable is used to transfer a message between TIOAs of two actors. %This transition describes  that the moving object departs from the sub-track. 
	%This transition is enabled if the target actor is available, checked by the function \emph{isDesFree}. Otherwise, the automaton moves to \emph{Error/Adapt}.
	The functions $\mathit{inChS(j)}$  returns the set of input channels  of the actor. The function  $\mathit{outCh}$ returns the output channel over which the message \textit{movingPlan} is sent. 
	The automaton moves to \emph{Error/Adapt} by synchronizing with $\mathit{error_{j}}$ or $\mathit{errorORadapt_{j}}$, and moves back by synchronizing with $\mathit{errorOver_{j}}$ or $\mathit{adapted_{j}}$ of the coordination TIOA in Fig.~\ref{fig::automata1}\protect\subref{fig::automata::coordinator}.  ]{{\fbox{\includegraphics[width=0.4\linewidth,keepaspectratio]{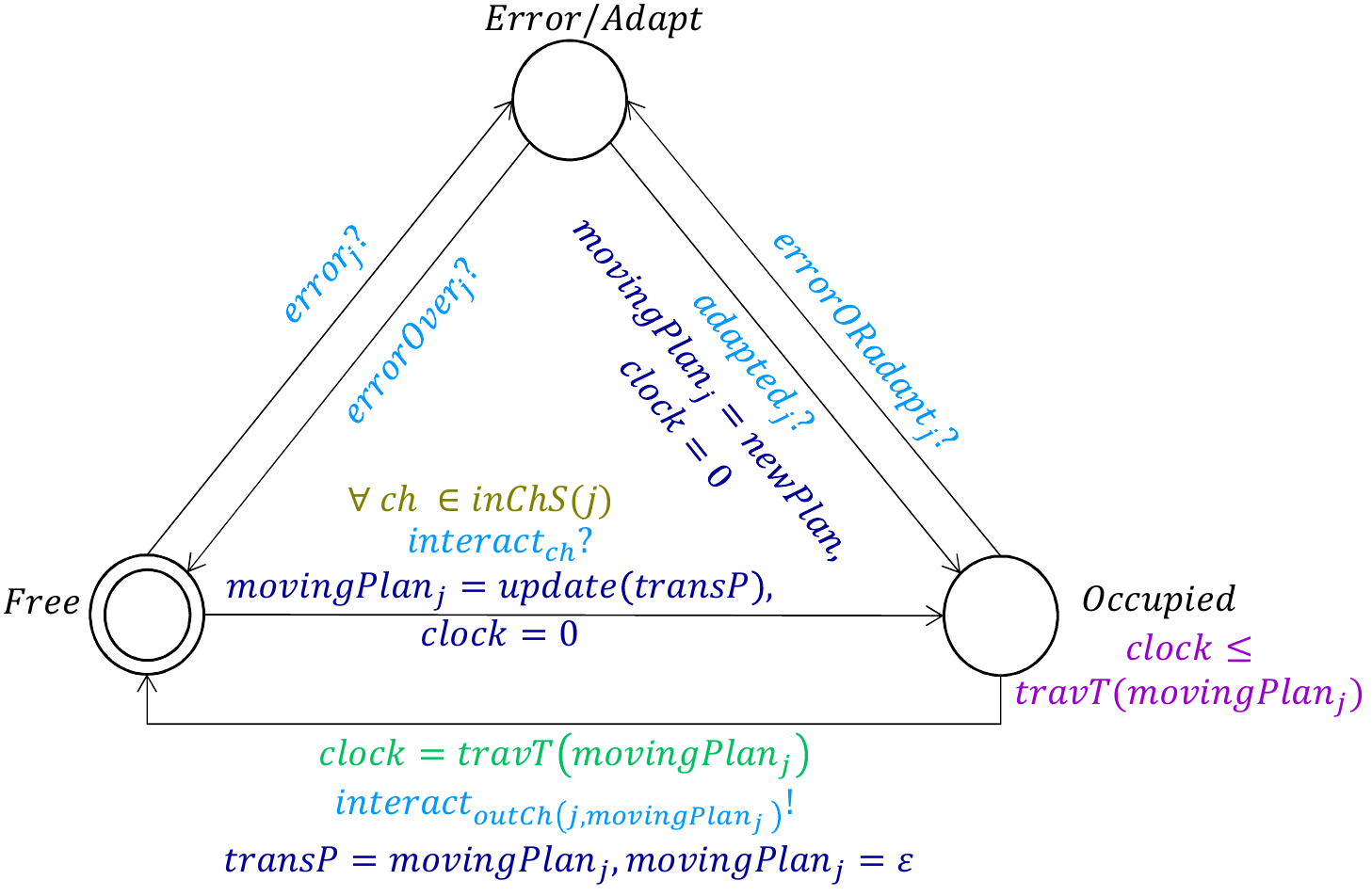}}\label{fig::automata::detailed}}}
	\quad
	\subfloat[The abstract representation of TIOA of the coordinator of the component $C_i$,  where $I_{C_i}$ is the set of identifiers of actors of $C_i$. The automaton changes the locations of the actors using a set of output actions, which are defined for each actor of the component $C_i$. The function \textit{adapt} develops a new traveling plan that is sent to the actor $a_j$ using the global variable \textit{newPlan}.  ]{{\fbox{\includegraphics[width=0.25\linewidth,keepaspectratio]{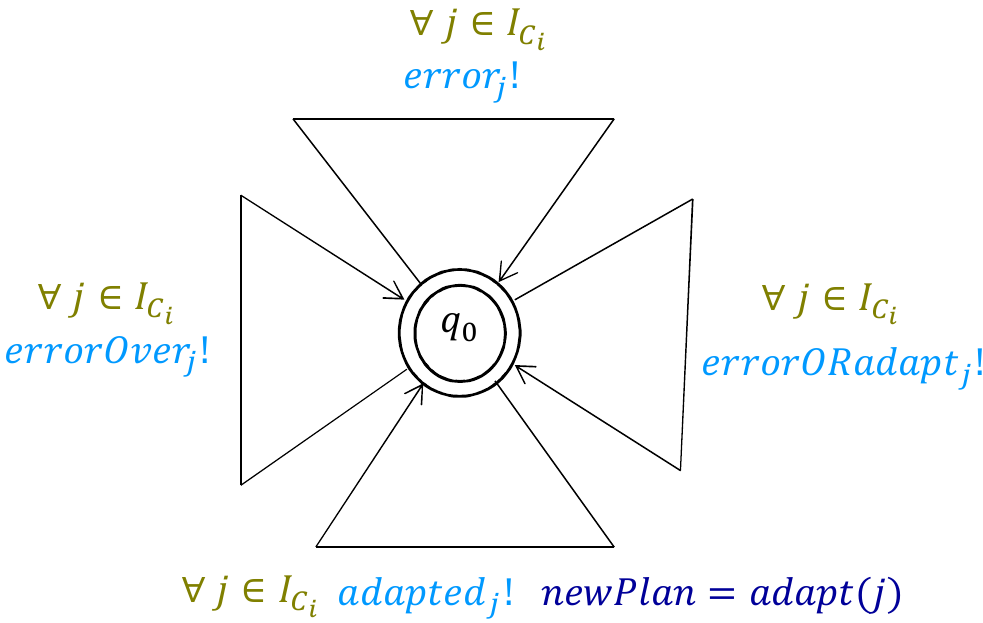}}\label{fig::automata::coordinator}}}
	\caption{An abstract TIOA of an actor $\mathit{a_j}$ of  component  $\mathit{C_i}$ is shown in  \protect\subref{fig::automata::abstract} and a more detailed view is shown in \protect\subref{fig::automata::detailed}. An abstract representation of TIOA of the coordinator of the component $C_i$ is shown in \protect\subref{fig::automata::coordinator}.
	}
	\label{fig::automata1}%
\end{figure*}

\begin{figure*}
	\centering
	%\qquad
	\subfloat[The TIOA of an actor $a_i$ of the component $C_i$ in the model@runtime. This automaton is obtained from the automaton of  Fig.~\ref{fig::automata1}\protect\subref{fig::automata::detailed} by removing the adaptation mode of the actor. This is because the model@runtime is analyzed after the coordinator applies its adaptation decision to the actors. ]{{\fbox{\includegraphics[width=0.47\linewidth,keepaspectratio]{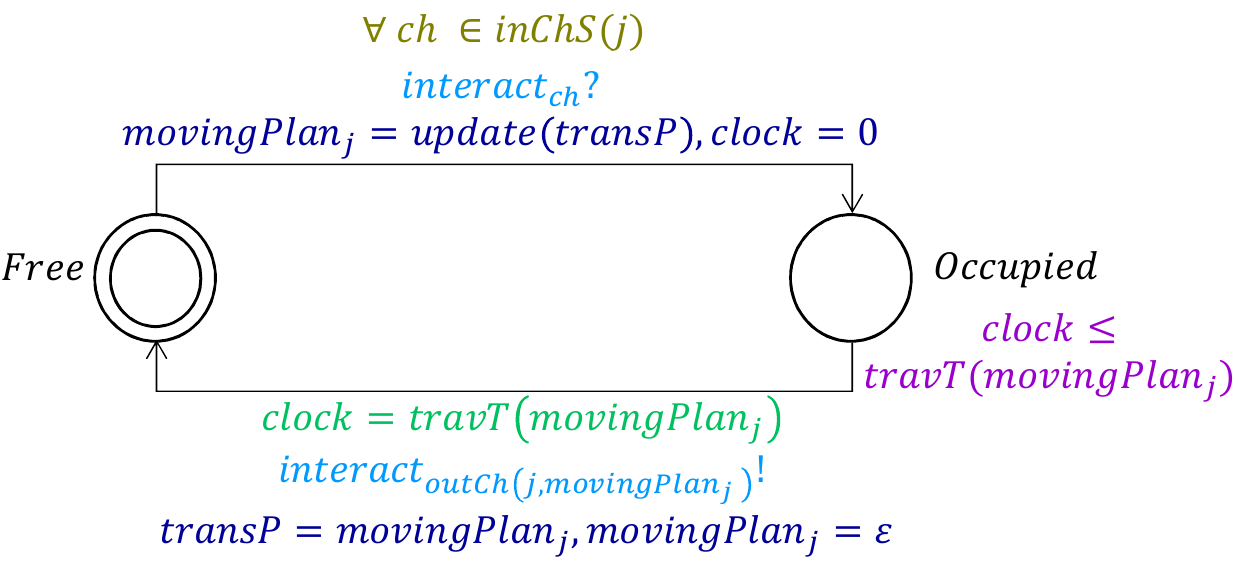}}\label{fig::automata::simplified}}} 
	\quad
	\subfloat[The TIOA of an augmented environment actor $\mathit{\mathit{aa}_j}$ belonging to $C_i$. %The TIOA shown in part (d) is described in Section~\ref{section::magnifier}.  
	The augmented environment actor $\mathit{\mathit{aa}_j}$ has a list  $\mathit{ERS}$ of messages that have to be sent or received by the actor at the pre-specified times and over the pre-specified channels.
	%\marjan{ETRS stands for  Expected To Receive and Send, and is basically the component overall plan.}
	The actor stays at state $q_0$ until the time progresses up to a time at which $\mathit{\mathit{aa}_j}$ sends or receives a message. 
	The time required to be elapsed to send or receive a message is calculated by the function \emph{timeProg}. 
	This automaton synchronizes on the action $\mathit{interact_{expCh(ERS_j)}}$  to receive or send a message  over an expected channel \emph{expCh(ERS)}.
	Functions \emph{in} and \emph{out} determine if a message should be received or sent by  $\mathit{\mathit{aa}_j}$, respectively.
	The top edge represents receiving, and it is only enabled if the message to be received matches the message that $\mathit{\mathit{aa}_j}$ is expecting. This condition is checked by the function \emph{hasMsg}.
	On the lower, edge $\mathit{\mathit{aa}_j}$ sends a message, and using the function \emph{expMes}, the first message in the list is added to \emph{transP} which denotes passing the message by $\mathit{\mathit{aa}_j}$ to the receiver actor.
	%} 
	%This automaton synchronizes on $\mathit{interact_{expCh(ETRS)}}$ with the  automaton of an actor, if this is the time for $a_j$ to receive a message from the actor over an expected channel (function \emph{expCh} returns this channel). This edge (top edge) is enabled if the actor contains the message the environment actor $a_j$ expects, checked by the function \emph{hasMsg}. %This condition is checked by the function \emph{hasMsg}.
	%If this is the time for $a_j$ to send a message  to an actor, the automaton of $a_j$ synchronizes with the automaton of the actor over the action $\mathit{interact_{\mathit{expCh(ETRS)}}}$. 
	Finally, the function \emph{updateL} removes the first entry from the \emph{ERS} list.  
	%The function \emph{expMes} returns the message of the first entry of the list. %The message of the first entry of the list is returned by the function \emph{expMes}.
	]{{\fbox{\includegraphics[width=0.47\linewidth,keepaspectratio]{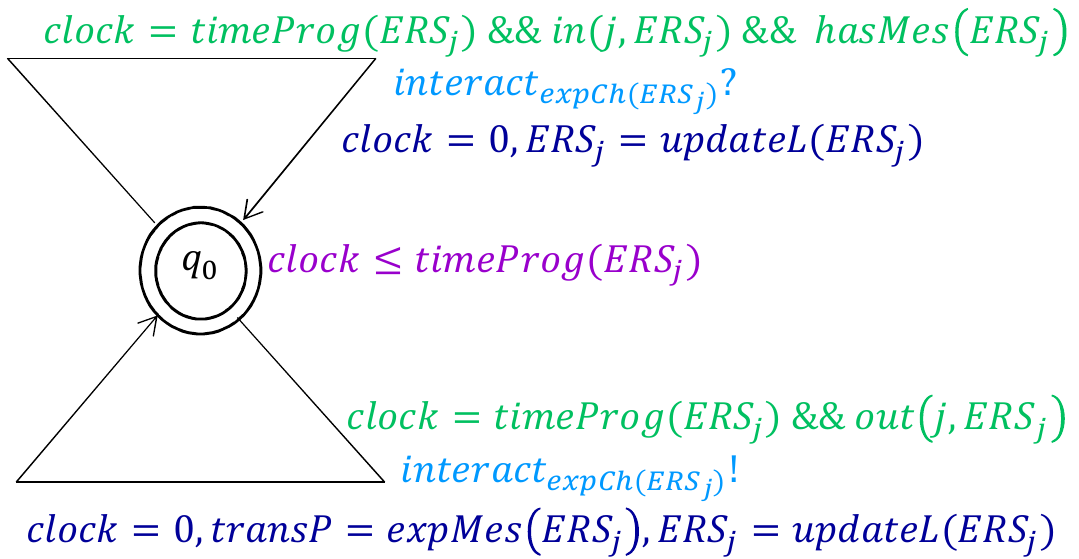}}\label{fig::automata::environment}}} 
	\caption{The TIOAs of an actor and an augmented environment actor in the model@runtime are respectively shown in \protect\subref{fig::automata::simplified} and \protect\subref{fig::automata::environment}. The automaton presented in  \protect\subref{fig::automata::environment} is described in Section.~\ref{section::magnifier}. %the model of a track-based system is shown in Figures \protect\subref{fig::automata::abstract}, \protect\subref{fig::automata::detailed}, and \protect\subref{fig::automata::simplified}. The TIOA of an environment actor $\mathit{a_j}$ belonging to $C_i$ is shown in part \protect\subref{fig::automata::environment}, described in Section.~\ref{section::magnifier}. %Figure \protect\subref{fig::automata::abstract} is an abstract TIOA that is detailed in Figure \protect\subref{fig::automata::detailed} and \protect\subref{fig::automata::simplified} is the simplified version of Figure \protect\subref{fig::automata::detailed} which is generated by removing the adaptation mode of the actor.
		%	The variable \emph{plan} is the local variable, and \emph{transP} is a global variable to transfer messages between two TIOAs. An empty message is denoted by $\epsilon$. The synchronization action $\mathit{interact_{i,ch}}$ is defined for every channel $\mathit{ch}$ of $\mathit{C_i}$.
	}%
	\label{fig::automata2}%
\end{figure*}

The TIOA of an actor, modeling a sub-track, without any details about the edges is shown in Fig.~\ref{fig::automata1}\protect\subref{fig::automata::abstract}. This automaton has three locations that correspond to the values of the \textit{status} variable of the actor. %As can be seen, the TIOA of an actor models the arrival of a moving object at a sub-track. It also models the departure of a moving object from a sub-track after traveling for a while. 
The \textit{Free} location that is also the initial location represents that the sub-track  is empty. The automaton moves from \textit{Free} to \textit{Occupied} whenever a moving object arrives at the sub-track. We suppose that the route of the moving object is updated over this transition. The route of a moving object is a sequence of sub-tracks traveled by the moving object. By passing a moving object from a sub-track, the first entry in the route of the moving object, referring to the current sub-track, is removed. The clock $\mathit{clock}$ %associated with TIOA of an actor %measures the amount of time a moving object spends in a sub-track
represents the time elapsed since the arrival of the moving object at the sub-track. The traveling time of the moving object in Fig.~\ref{fig::automata1}\protect\subref{fig::automata::abstract} is denoted by $d$.  The sub-track remains occupied during the traveling time of the moving object, formulated in the \textit{Occupied} location. When a sub-track is occupied another moving object can not enter into it, because a sub-track is a critical section. %Besides the arrival and departure of a moving object, this automaton reflects the situation where an adversarial environmental condition happens to the sub-track. %The \textit{Error/Adapt} location represents two situations: an adversarial environmental condition presents in the sub-track, and/or the traveling plan of the moving object in the sub-track needs an adaptation. 
The automaton moves from \textit{Occupied} to \textit{Free} whenever the moving object leaves the sub-track. A change can happen to a sub-track at any time, while the sub-track is empty or occupied. If the  sub-track is empty and an adversarial event  happens to the sub-track, the automaton moves from \textit{Free} to \textit{Error/Adapt}.  %The automaton remains in \textit{Error/Adapt} until the event is removed from the sub-track.  
The automaton moves back to the \textit{Free} location the moment this adversarial event is removed. The automaton moves from \textit{Occupied} to \textit{Error/Adapt} if the sub-track is occupied and an adversarial event happens to the sub-track. The other case through which the automaton moves to \textit{Error/Adapt} is when the sub-track is occupied and the traveling plan of the moving object should be adapted.  The automaton moves back to the \textit{Occupied} location whenever an adaptation decision in the form of rerouting/rescheduling for the traveling plan is made. %The adaptation in the form of rerouting/rescheduling is performed over this transition by changing the local variable of the actor

%\COM{MARJAN: This paragraph also should move to section 4.3 where you explain the TIOAs  of the coordinator and their composition.} Suppose the case in which the target sub-track of a moving object is unavailable. A sub-track is unavailable if it is occupied or is affected by an adversarial environmental condition. The coordinator can adapt the source actor by giving a new plan to the actor. This way, the moving object traveling across the sub-track corresponding to the source actor may stay more or less in its current sub-track or may change its departure direction. Furthermore, the coordination informs an actor about an adversarial environmental effect through the \textit{error} argument of the $\mathit{chang}$ method.  %Therefore, whenever a message is sent over an output port,

The TIOA of an actor $a_j$ of the component $C_i$ is shown in Fig.~\ref{fig::automata1}\protect\subref{fig::automata::detailed} with more details. This automaton has a \textit{movingPlan} variable that corresponds to the \textit{movingPlan} variable of the actor in Definition.~\ref{def::actor}.  The automaton has a set of input actions $\mathit{interact_{ch}}$, where \textit{ch} is the channel bound to an input port of the actor. These actions correspond to the interact handlers of the actor. Similarly, for each channel \textit{ch} bound to an output port of the actor an output action $\mathit{interact_{ch}}$ is defined. Regarding the \textit{change} port of the actor, the automaton has a set of input actions $\mathit{errorORadapt_{j}}$, $\mathit{adapted_{j}}$, $\mathit{error_{j}}$, $\mathit{errorOver_{j}}$,  where each action has the same name as an input argument of the change handler. This automaton has access to the global variable \textit{transP} that corresponds to the input argument of the interact handlers and is used to transfer a value between TIOAs of two actors. The global variable \textit{newPlan} transfers a value from  TIOA of the coordinator to  TIOA of the actor. This variable corresponds to the \textit{newPlan} argument of the change handler.  The TIOA of an actor is defined as follows.

\begin{definition} \label{actor::dec}
	(TIOA of an Actor) The TIOA associated with an actor $a_j$ of $\mathit{C_i}$ is $TA=(\{\mathit{Free},\mathit{Occupied},\mathit{Error/Adapt}\}, \mathit{Free},\mathit{Var},\{\mathit{clock}\},\mathit{Act_{\mathit{in}}},\\\mathit{Act_{out}},T,I)$, where $\mathit{Var}=\{\mathit{movingPlan_j}\}$, $\mathit{Act_{\mathit{in}}}=\{\mathit{interact_{ch}}|\mathit{ch} \in \mathit{inChS(j)}\} \cup \{\mathit{error_{j}}, \mathit{errorOver_{j}},\mathit{errorORadapt_{j}},\mathit{adapted_{j}}\}$, $\mathit{Act_{out}}=\{\mathit{interact_{ch}}| \mathit{ch} \in \mathit{outChS(j)} \}$, $I(\mathit{Occupied})=\mathit{clock} \leq travT(movingPlan_j)$, and $T$ is defined as follows.
	\begin{multline*}
	\tag{Receive}
	\forall \mathit{ch} \in \mathit{inChS(j)}: (\mathit{Free},\mathit{true},\mathit{interact_{ch}},\{\mathit{clock}\},\\\{\mathit{movingPlan_j=}\mathit{update(transP)}\},\mathit{Occupied})
	\end{multline*}
	\begin{multline*}
	\tag{Send}
	(\mathit{Occupied},\mathit{clock}=travT(movingPlan_j), \\\mathit{interact_{\mathit{outCh(j,movingPlan_j)}}},\emptyset,\\\{\mathit{transP=movingPlan_j},\mathit{movingPlan_j=\epsilon}\},\mathit{Free}) 
	\end{multline*}  
	\begin{multline*} \tag{Error}
	(\mathit{Free},\mathit{true},\mathit{error_{j}},\emptyset,\emptyset,\mathit{Error/Adapt})
	\end{multline*} 
	\begin{multline*} \tag{NoError}
	(\mathit{Error/Adapt},\mathit{true},\mathit{errorOver_{j}},\emptyset,\emptyset,\mathit{Free})
	\end{multline*}
	\begin{multline*} \tag{Adapt}
	(\mathit{Occupied},\mathit{true},\mathit{errorORadapt_{j}},\emptyset, \emptyset,\\\mathit{Error/Adapt})
	\end{multline*}
	%	\begin{multline*} \tag{Adapt.2}
	%	(\mathit{Occupied},\mathit{clock}=travT(movingPlan_j) \wedge \mathit{isDesFree(j,movingPlan_j)},\\\tau,\emptyset,\{\mathit{transP=movingPlan_j}\},\mathit{Error/Adapt})
	%	\end{multline*}
	\begin{multline*} \tag{Adapted}
	(\mathit{Error/Adapt},\mathit{true},\mathit{adapted_{j}},\{\mathit{clock}\},\\\{\mathit{movingPlan_j=newPlan}\},\mathit{Occupied})
	\end{multline*} \qed
\end{definition}

In the following, we describe edges of  TIOA of the actor $a_j$. For the sake of convenience, we call the channel bound to an input port of an actor an input channel of the actor. Similarly, we call the channel bound to an output port of an actor the output channel of the actor.

An actor is always waiting to receive a message over an input port.  If a message is present, the actor receives the message  and the message handler  corresponding to that input port is triggered. By triggering $\mathit{interact_{j,l}(transP)}$ of the actor $\mathit{a_j}$ on its input port $p_{j,l}$, the actor receives the message \textit{transP} sent from another actor. When the message is received, the \textit{status} variable of the actor  is set to \textit{Occupied}. %The message is in the \textit{transP} argument of the interact handler. 
%By triggering $\mathit{interact_{j,l}(transP)}$ of actor \textit{j} on its port $p_{j,l}$, %the arrival of a  moving object into a sub-track  is modeled. (Note that $j=1,\cdots,\mathit{directions}$ and  \textit{direction} is the number of arrival and departure directions for a sub-track.) 
The message \textit{transP} includes the object id of the message ($objectId$) and the travel plan ($travelPlan$);  	\textit{travelPlan} includes  the traveling route of the moving object, its schedule, the amount of  fuel, and its speed.  This information is stored in the \textit{movingPlan} variable of the actor. 
%We suppose that 
%\COM{Why you say "we suppose that"?}
The actor updates the traveling plan (by removing the first entry in the route) before storing it in \textit{movingPlan}. These operations are specified through the \textit{Receive} edge that is defined for every input channel of the actor.  The actor $a_j$ of $\mathit{C_i}$ receives a message from the port connected to the  channel $\mathit{ch}$ whenever TIOA of the actor is synchronized with the input action $\mathit{interact_{ch}}$.   The auxiliary functions used over this edge are described as follows. Let $\AID$ be the set of all actor identifiers,  $\mathit{CH}$ be the set of all channels in the model, and $\MSG$ be the set of all messages in the form of (\textit{objectId}, \textit{travelPlan}).  The function $\mathit{inChS(j)}$, where $\mathit{inChS:\AID \rightarrow 2^{\mathit{CH}}}$, returns the set of all input channels of the actor $a_j$. The function $\mathit{update: \MSG \rightarrow \MSG}$ receives a message and returns a new message in which the traveling plan is updated. %The received message is stored in the \textit{movingPlan} variable of the actor.

%\COM{Do not switch back and forth between the CoodAA terminology and TTCS terminology, I put parenthesis to avoid switching}
%\marjan{When a message is received,} the \textit{status} variable of the actor  is set to \textit{Occupied}. 

The actor introduces a delay that is derived based on the traveling plan included in  the \textit{movingPlan} variable. This operation is formulated in the \textit{Occupied} location of the automaton, where the function $\mathit{travT: \MSG \rightarrow \R}$ receives a message and calculates the amount of the delay. After passing the delay time, the actor  sends out the message. The output port of the actor over which the message is sent is also determined using the traveling plan. After sending the message, the \textit{status} variable of the actor  is set to \textit{Free} and \textit{movingPlan} is set to \textit{null}. These operations are specified through the \textit{Send} edge that is defined for every output channel of the actor. The function $\mathit{outChS(j)}$, where $\mathit{outChS}:\AID \rightarrow 2^\mathit{CH}$, returns the set of  all output channels of the actor $a_j$.  The automaton of the actor is synchronized with the output action $\mathit{interact_{\mathit{outCh(j,movingPlan_j)}}}$ over this edge. The function $\mathit{outCh(j,movingPlan_j)}$, where $\mathit{outCh}:\AID \times \MSG \rightarrow \mathit{CH}$, returns  the output channel of the actor $a_j$ over which channel the message $\mathit{movingPlan_j}$ is sent. %Let $\CAMID$ be the set of all component identifiers. The function $\mathit{outCh(j,msg)}$, where $\mathit{outCh}:\AID \times \MSG \rightarrow \mathit{CH}$, returns  the output channel of the actor $a_j$ over which channel the message $\mathit{msg}$ is sent. The function $\mathit{desCOM(j,msg)}$, where $\mathit{desCOM}:\AID \times \MSG \rightarrow \CAMID$, returns the identifier of the component whose one of actors will receive the message $\mathit{msg}$ from the actor $a_j$. 
The message is transferred between two actors whenever their TIOAs are synchronized over an interact action.  The message is delivered to the receiver actor using the \textit{transP} variable. %The \textit{Send} edge also describe the case in which a message is sent to the actor of another component, because this  edge is  synchronized over an action having the identifier of the other component.

The execution of an interact handler of the actor cannot be preempted  by other interact handlers (the sub-track is a critical section), but the change handler has the highest priority for execution and can preempt interact handlers. By triggering the change handler, the message (\textit{errorORadapt}, \textit{adapted}, \textit{error}, \textit{errorOver}, \textit{newPlan}) sent from the coordinator is received.  The \textit{status} variable of the actor is set to \textit{Error/Adapt} under the following conditions. First,   \textit{error} is true and  \textit{status} is \textit{Free}, second, \textit{errorORadapt} is true and  \textit{status} is \textit{Occupied}. The value of \textit{status} changes from \textit{Error/Adapt} to \textit{Free} if \textit{errorOver} is true. Furthermore, the value of \textit{status} changes from \textit{Error/Adapt} to \textit{Occupied} if \textit{adapted} is true. If the latter case holds, \textit{movingPlan} variable of the actor is set to the new plan stored in \textit{newPlan} of the message. These operations are specified through the \textit{Error}, \textit{NoError}, \textit{Adapt}, and \textit{Adapted} edges. The  automaton respectively is synchronized with the input actions $\mathit{error_{j}}$ and $\mathit{errorOver_{j}}$ over the \textit{Error} and \textit{NoError} edges. Furthermore, the automaton respectively is synchronized with the input actions  $\mathit{errorOradapt_{j}}$ and $\mathit{adapted_{j}}$ over the \textit{Adapt} and \textit{Adapted} edges.

An abstract TIOA for the coordinator is presented in Fig.~\ref{fig::automata1}\protect\subref{fig::automata::coordinator}. As can be seen, this automaton is synchronized with the automaton of an actor over the actions  $\mathit{error_{j}}$, $\mathit{errorOver_{j}}$, $\mathit{errorOradapt_j}$, and $\mathit{adapted_j}$.  This shows the role of the coordinator to update the  state and/or the traveling plan of an actor. %whenever an environmental event happens to the corresponding sub-track or is removed from the corresponding sub-track. The coordinator is also able apply an adaptation decision to an actor whenever it is needed. %Furthermore, the coordinator checks whether the actor corresponding to the next sub-track in the route of the moving object is processing another message. 
\textbf{TIOA of an Actor in the Model@runtime.} 
The adaptation mechanism of coordinator  obtains new plans for the actors. In Magnifier, the model@runtime is verified after an adaptation decision is applied to the model. 
%based on which the messages are dispatched (the moving objects are routed/scheduled). 
An actor with  \textit{movingPlan}=null and \textit{status=Error/Adapt} (an empty sub-track with an adversarial environmental condition) does not contribute in communications. Therefore,  the \textit{Error/Adapt} location is not needed in TIOA of an actor in the model@runtime. As shown in Fig.~\ref{fig::automata2}\protect\subref{fig::automata::simplified}, we simplify TIOA of the  actor $a_j$ of the component $C_i$.  Compared to Definition.~\ref{actor::dec}, this simple TIOA does not contain the   \textit{Error/Adapt} location of the actor. 

%% file: 05-MagnifierApproachForCompositionalVerificationOfSelfAdaptive.tex
\section{Verification of Model@runtime Using Magnifier} \label{section::magnifier}
In this section, we develop a compositional approach to verify the system in the case of a change occurring and applying adaptation to components. We first provide insight into the Magnifier approach. We then present the definitions used to describe the approach. Finally, we formally explain the Magnifier approach and prove its correctness. 

\input{021-Magnifier}

\begin{table*}
	\begin{center}
		\begin{tabular}{|m{\textwidth}|} 
			\hline
			\begin{itemize}
				\item  $C_j$ is an environment component of $C_i$, $C_j \in \mathit{Env(C_i)}$, if  $\exists \mathit{ch} \in \mathit{CH_{\mathcal{M}}} \cdot p_i \in P_{C_i}, p_j \in P_{C_j}, \{(p_i,\mathit{ch}),(p_j,\mathit{ch})\} \subseteq B_\mathcal{M}$, where $C_i$ and $C_j$ are components of $C_\mathcal{M}=(A_\mathcal{M},\mathit{CH_{\mathcal{M}}},\mathit{B_\mathcal{M}},  \mathit{F_\mathcal{M}})$ and  $P_{C_i}=\{p|  a_l \in A_i \wedge p \in (P_{I_l} \vee P_{O_l}) \wedge \nexists \mathit{ch} \in \mathit{CH_i} \cdot (p,\mathit{ch}) \in B_i\}$ is the set of boundary ports of the component $C_i$. The set of all environment components of $C_i$ is denoted by $\mathit{Env(C_i)}$ 
				%\item  \COM{Cj is an environment component of Ci where ...} $\mathit{Env(C_i)}$ denotes the set of environment components of the component $C_i$
				\item $a_k$ of $C_j$ is an environment actor of $C_i$, $a_k \in C_j |_{C_i}$,  iff  $ C_j \in \mathit{Env(C_i)} \wedge \exists \mathit{ch} \in \mathit{CH_{\mathcal{M}}}, p_j \in (P_{I_j} \vee P_{O_j}), p_i \in P_{C_i} \cdot \{(p_i,\mathit{ch}),(p_j,\mathit{ch})\} \subseteq B_\mathcal{M}$. The set of actors of $C_j$ that are environment actors  of $C_i$ is denoted by $C_j |_{C_i}$

				\item $aa_k$ is the augmented environment actor, corresponding to the environment actor $a_k \in C_j |_{C_i}$, defined as   $(\mathit{ERS_k},\mathit{init_k}()\{\mathit{body}\},\mathit{Mtds_k},P_{I_k}, P_{O_k})$
				\begin{itemize}
					\item $\mathit{ERS_k}$ is an ordered list. Each entry of  $\mathit{ERS_k}$ is defined as $(\mathit{transP},\mathit{t},\mathit{ch})$, where   $transP=(objectId,travelPlan)$ is a message, $\mathit{t}$ is a delay value, and $\mathit{ch}$ is a channel identifier
					\item $\mathit{init_k}()\{\mathit{body}\}$ is an initialisation method where \textit{body} is a sequence of statements
					\item $\mathit{Mtds_k}=\mathit{interact_{k,l}(transP)}\{\mathit{body}\}$ is the set of message handlers of $aa_k$ where $l=1, \cdots,\mathit{Num}$ and \textit{Num} is the number of input ports of $a_k$ through which  $a_k$ interacts with $C_i$
					%with the sequence of input arguments $transP$, and $l=1, \cdots,\mathit{Num}$,
					\item $P_{I_k}$ is the set of input ports of $aa_k$. Each port is an input port of $a_k$ 
					through which $a_k$ communicates with $C_i$ 
					\item $ P_{O_k}$ is the set of output ports of $aa_k$. Each port is an output port of $a_k$ 
					through  which $a_k$ communicates with $C_i$
				\end{itemize}
				\item  $C_j \downarrow_{C_i}$ is the set of augmented environment actors of $C_i$ where each actor of this set corresponds to an actor of $C_j |_{C_i}$
				%\item $a_k \in C_j$, where $C_j \in \mathit{Env(C_i)}$, is an environment actor of the component $C_i$ if there is a channel connecting a port of $a_k$ to a boundary port of $C_i$
				\item $\N_{C_i}$ denotes the network of TIOAs of the component $C_i$
				\item $\N_{C_{a,i}}$ denotes the network of TIOAs of the adapted component $C_i$
				\item $\N_1 \bowtie \cdots \bowtie \N_n$ denotes that the networks $\N_1,\cdots,\N_n$ are compatible, where the parallel product of their TIOAs does not reach a deadlock state
				
			\end{itemize}
			\\
			\hline
		\end{tabular}
	\end{center}
	\caption{Summary of definitions and notations for the Magnifier approach
	%\COM{You dropped the definition of	Env(Ci), what is the relation of Env(Ci) and Ci? likewise for Environment Actors, which actors in Cj are aj?.	It may be easier to say Cj is an enviroment component of Ci, and Cj /in Env(Ci) where ...also ak is an environment actor ... like when you say $aa_k$ ...}
	}
	\label{tab::summery2}
\end{table*}

\subsection{Preliminary Definitions for Magnifier }

In this section, we present the definitions on which the Magnifier approach relies. The notations and the summary of definitions are given in Table.~\ref{tab::summery2}.  Let $C_\mathcal{M}=(A_\mathcal{M},\mathit{CH_{\mathcal{M}}},\mathit{B_\mathcal{M}},  \mathit{F_\mathcal{M}})$, composed of a set of components, be the coordinated adaptive actor model of a track-based system, where $A_\mathcal{M}$ is the set of actors, $\mathit{CH_{\mathcal{M}}}$ is the set of channels, $\mathit{B_\mathcal{M}}$ is the binding function, and $\mathit{F_\mathcal{M}}$ is the coordination function of $C_\mathcal{M}$.  The component $C_i$ of $C_\mathcal{M}$ models an area of the system, and interacts with a set of components called environment components of $C_i$. Let $P_{C_i}=\{p|  a_l \in A_i \wedge p \in (P_{I_l} \vee P_{O_l}) \wedge \nexists \mathit{ch} \in \mathit{CH_i} \cdot (p,\mathit{ch}) \in B_i\}$ be the set of boundary ports of the component $C_i$. An environment component is defined as follows.

%Let $\mathit{shared_c(C_i,C_j)}$ be the set of channels shared between two components $C_i$ and $C_j$, i.e. $\mathit{shared_c(C_i,C_j)=(\mathit{CH_{I_i}} \cap \mathit{CH_{O_j}}) \cup (\mathit{CH_{I_j}} \cap \mathit{CH_{O_i}})}$. An environment component is defined as follows.
\begin{definition} \label{def::envC}
	(Environment Component) The component $C_j$ is called an environment component of the component $C_i$ if  there exists a channel  $\mathit{ch}$ connecting a boundary port $p_i$ of the component $C_i$ to a boundary port $p_j$ of the component $C_j$, i.e. $\mathit{ch} \in \mathit{CH_{\mathcal{M}}}, p_i \in P_{C_i}, p_j \in P_{C_j}, \{(p_i,\mathit{ch}),(p_j,\mathit{ch})\} \subseteq B_\mathcal{M}$. The set of all environment components of $C_i$  is denoted by $\mathit{Env(C_i)}$. \qed
\end{definition}

\begin{example}
A model with five components $C_1$, $C_2$, $C_3$, $C_4$, and $C_5$ is shown in Fig.~\ref{interactiveComponents}(a). The connections between the actors are denoted by arrows. The components $C_2$, $C_3$, and $C_4$ are environment components of $C_1$.
\end{example}

\begin{figure}
	\centering
	\includegraphics[width=.45\textwidth,keepaspectratio]{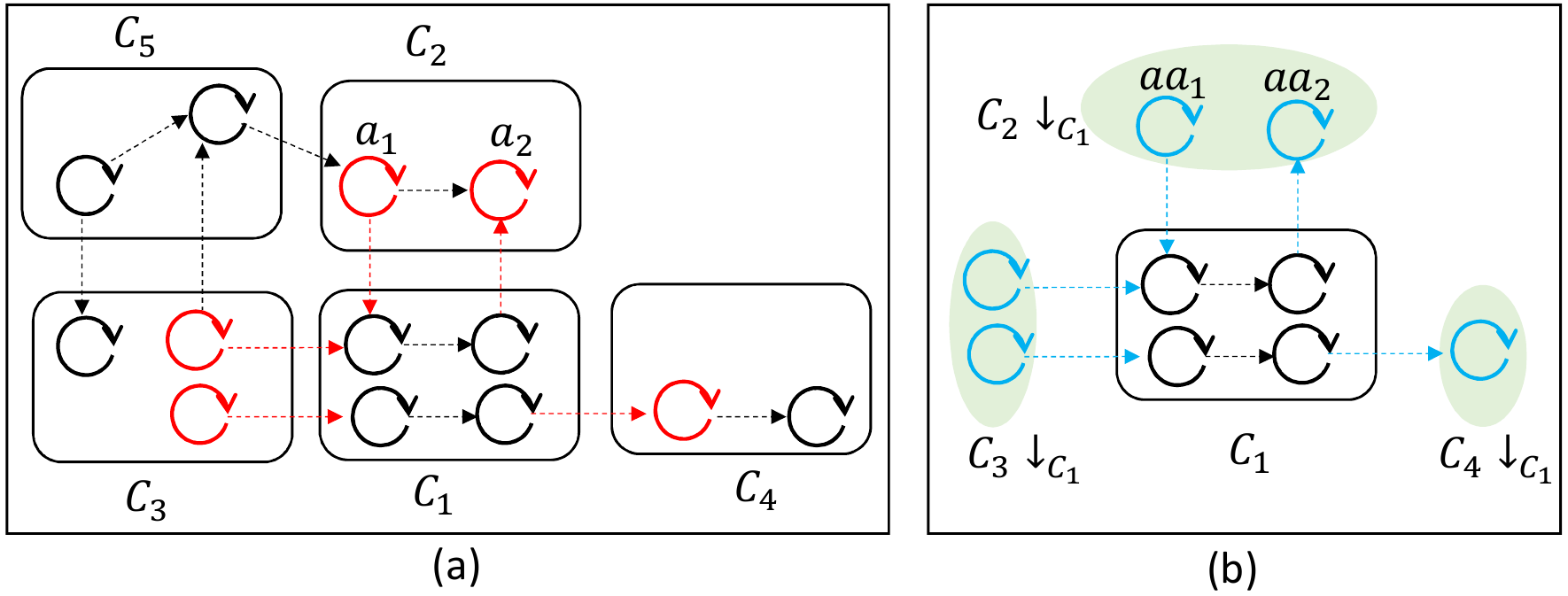}
	\caption{A model consisting of 5 components is shown in (a). The connections between actors are shown with dashed arrows. The red actors are environment actors of the component $C_1$. The interface components of $C_1$, i.e. $C_j \downarrow_{C_1}$, $j=1,2,3$, are shown in (b). The augmented environment actors of $C_1$ along with their ports are shown in blue.}
	\label{interactiveComponents}
\end{figure}

Using Definition~\ref{def::envC}, we define an environment actor. An environment actor of a component is an actor 
whose input and output ports are bound to the input and output ports of the component by a set of channels (and
hence directly sends or receives messages to/from the component). 

\begin{definition}
	(Environment Actor) Let $C_j \in \mathit{Env(C_i)}$ be an environment component of $C_i$. The actor $a_j$ of $C_j$ is an environment actor of $C_i$ iff  $\exists \mathit{ch} \in \mathit{CH_{\mathcal{M}}}, p_j \in (P_{I_j} \vee P_{O_j}), p_i \in P_{C_i} \cdot \{(p_i,\mathit{ch}),(p_j,\mathit{ch})\} \subseteq B_\mathcal{M}$. \qed
\end{definition}

%We use the definition of an environment actor to define an augmented environment actor. 
%In Magnifier, an environment actor is augmented with a list that is named \textit{ETRS}. An entry of \textit{ETRS} contains a message, a delay value, and a channel identifier. % where the channel bound to an input or an output port of the actor. %that the environment actor expects to send or receive a message that the environment actor over a specific port after an amount of time.%, or specifies a message that the environment actor intends to send over an output port after an amount of time.

%\begin{definition}
%	(\textit{ETRS}) \textit{ETRS} is an ordered list. Each entry, \textit{e}, of  \textit{ERTS} is defined as $e=(\mathit{transP},\mathit{t},\mathit{ch})$, where the message $\mathit{transP}=(objectId,travelPlan)$, $\mathit{t}$ is a time, and $\mathit{ch}$ is a channel identifier. \qed
%\end{definition}
We use $C_j |_{C_i}$ to denote the set of actors of $C_j$ that are environment actors  of $C_i$.

\begin{example}
The red actors shown in Fig.~\ref{interactiveComponents}(a) are environment actors of $C_1$. For the environment component $C_2$, $C_2 |_{C_1}=\{a_1,a_2\}$.
\end{example} 

Suppose that $a_k \in A_j$ is an environment actor of  the component $C_i$.  We use $P_{I_{k,C_i}}$ and $P_{O_{k,C_i}}$ to respectively denote the sets of input and output ports of the actor $a_k$ over which it communicates with the component $C_i$, i.e.  $P_{I_{k,C_i}}=\{p \in P_{I_k}| \exists p_i \in P_{C_i},  \mathit{ch} \in \mathit{CH}_\mathcal{M} \cdot \{(p,\mathit{ch}),(p_i,\mathit{ch})\} \subseteq B_\mathcal{M}\}$, $P_{O_{k,C_i}}=\{p \in P_{O_k}| \exists p_i \in P_{C_i},  \mathit{ch} \in \mathit{CH}_\mathcal{M} \cdot \{(p,\mathit{ch}),(p_i,\mathit{ch})\} \subseteq B_\mathcal{M}\}$.

\vspace{0.5cm}
\noindent\textbf{Abstraction of the environment.} In order to abstract the environment of a component in Magnifier, corresponding to each environment actor an augmented environment actor is defined for the component. 
We augment all the significant information of an environment component to the augmented environment actor.
An augmented environment actor has a list of expected  receives and sends,  called  \textit{ERS}, where each entry of the list contains a message, a delay value, and the identifier of a channel. %abstract sends or receive messages  at different times.   
%We use the definition of an environment actor to define an augmented environment actor. 
%An environment actor augmented with the \textit{ETRS} list is called an augmented environment actor.
Besides, this actor has a set of input ports, a set  of output ports, an init method, and a set of message handlers, where each  handler corresponds to an input port of the actor. The augmented environment actor corresponding to the environment actor $a_k$ has  $P_{I_{k,C_i}}$ and $P_{O_{k,C_i}}$ as the sets of its input and output ports, respectively.

\begin{definition}
	(Augmented Environment Actor) An augmented environment actor, $aa_k$, corresponding to the environment actor $a_k$, %of the component $C_j$, where $a_k \in A_i$ and $C_j \in \mathit{Env(C_i)}$, 
	is defined as $(\mathit{ERS_k},\mathit{init_k}()\{\mathit{body}\},\mathit{Mtds_k},P_{I_k}, P_{O_k})$, where $\mathit{ERS_k}$ is an ordered list, $\mathit{init_k}$ is a method, $\mathit{Mtds_k}=\{\mathit{interact_{k,l}(transP)}\{\mathit{body}\}|l=1, \cdots,\mathit{Num}\}$ is the set of message handlers of the actor with the sequence of input arguments $transP=(objectId,travelPlan)$, and $P_{I_k}=P_{I_{k,C_i}}$ and $ P_{O_k}=P_{O_{k,C_i}}$ are the sets of input and output ports of the actor, respectively. The number of input ports is denoted by  $\mathit{Num}$. Each entry of  $\mathit{ERS_k}$ is defined as $(\mathit{transP},\mathit{t},\mathit{ch})$, where  $\mathit{t}$ is a delay value and $\mathit{ch}$ is a channel identifier. \qed
	
\end{definition}

 We use $C_j \downarrow_{C_i}$ to denote the set of augmented environment actors of $C_i$ where each actor of this set corresponds to an actor of $C_j |_{C_i}$. %where $C_j \in \mathit{Env(C_i)}$, to denote the sets of augmented environment actors of the component $C_i$, where each actor corresponds to an actor of $C_i$ that is an environment actor of $C_j$.  The set of actors of $C_j$ that are environment actors of $C_i$ is denoted by $C_j \downarrow_{C_i}$.

\begin{example}
The augmented environment actors of $C_1$ along with their ports are shown in blue in Fig.~\ref{interactiveComponents}(b),  i.e. $C_2 \downarrow_{C_1}=\{\mathit{aa_1, aa_2}\}$ where $\mathit{aa_1}$ and $\mathit{aa_2}$ correspond to the actors $a_1$ and $a_2$ in Fig.~\ref{interactiveComponents}(a), respectively.
\end{example} 

We now define an interface component of a component. %%Therefore, the interface components of a component are sets of environment actors of its environment components. 
The interface components of a component (or visible parts of its environment)  
are sets of its augmented environment actors. 

\begin{definition} \label{def::intCom}
	(Interface Component) For each $C_j \in \mathit{Env(C_i)}$, $C_j \downarrow_{C_i}$ is called an interface component of the component $C_i$. \qed
\end{definition}

%Based on Definition.~\ref{def::intCom}, interface components of a component are  abstractions of its environment components. %. an environment component $C_j \in \mathit{Env(C_i)}$ can be abstracted to a set of actors that is  an interface component. 
The definition of the interface component is inspired from the approach of \cite{clarke1989compositional}, where it defines interface processes. For two processes $P_1$ and $P_2$, $P_1 \downarrow \Sigma_{P_2}$ is an interface process of $P_2$, where $\Sigma_{P_2}$ is the set of symbols (i.e. atomic propositions) associated with $P_2$. The interface process $P_1 \downarrow \Sigma_{P_2}$ is the process $P_1$ in which all symbols that do not belong to $\Sigma_{P_2}$ are hidden. %In the next section, we explain that interface components in our approach are obtained based on the channels shared between two components. Furthermore,  the approach of this paper using interface components is formally described.

Finally, we  define compatible TIOAs. In the Magnifier approach, two components can interact if their TIOAs are compatible.  
%As described in Section~\ref{section::sem}, the semantics of a coordinated actor model as a component is defined in terms of a network of TIOAs. We define two or more  compatible TIOAs as follows.% Two or more compatible TIOAs are defined as follows.

%We  define a deadlock state in the parallel product of two or more TIOAs. A deadlock state is the same as the deadlock state defined in the \textsc{Uppaal} tool \cite{behrmann2006tutorial}. We modify the definition of compatible TIOAs as follows.
\begin{definition} \label{def::compT}
	(Compatible TIOAs) Two or more TIOAs are compatible if the parallel product of them does not reach a deadlock state. \qed
\end{definition}

Our definition of compatibility is inspired from the approach of \cite{10.1007/3-540-45828-X_9}, in which two components (timed interfaces) are compatible if there is an environment to avoid the parallel product of the components from reaching an error state (the environment makes the components work together). 
%In our approach, the compatibility is checked by ignoring a helpful environment. 
In our approach, for checking the compatibility we do not consider any helpful environment.
%It is notable that a deadlock state in Magnifier 
Note that a deadlock state in the product of two TIOAs 
is different from a deadlock in a track-based system.
%A deadlock in \did{a track-based system} happens whenever the moving objects are stuck in a traffic blockage and cannot find an available route towards their destinations.
We use Definition.~\ref{def::compT} to define the compatibility between networks of TIOAs. We call the networks $\N_1$ and $\N_2$ compatible if all TIOAs in $\N_1$ and $\N_2$ are compatible. Similarly, the networks $\N_1,\cdots,\N_n$ of TIOAs are compatible, denoted by $\N_1 \bowtie \cdots \bowtie \N_n$, if the networks are pairwise compatible.

\subsection{Magnifier Approach}
In this section, we use $\N_{C_i}$ to denote the network of TIOAs of the component $C_i$ such that all TIOAs in $\N_{C_i}$ are compatible. We also use $\N_{C_{a,i}}$ to denote the network of TIOAs of the adapted component $C_{a,i}$ such that all TIOAs in $\N_{C_{a,i}}$ are compatible. We first describe the Magnifier approach for two components.   
Consider a system whose model consists of only two interacting components $C_1$ and $C_2$. %, and $\mathit{Env(C_i)}$ denotes the set of environment components of $C_i$. 
To ensure that  correctness properties of the system are satisfied, it is checked whether $\N_{C_1}$ and $\N_{C_2}$ are compatible in the absence of a change. None of the correctness  properties are violated and there is a safe execution for the model if  $\N_{C_1} \bowtie \N_{C_2}$. By detecting a change, the  component affected by the change is adapted. Consequently, a new network for the adapted component is obtained.  Suppose that a change in $C_1$ is detected. If  $\N_{C_{a,1}} \bowtie \N_{C_2}$, the provided adaptation in $C_1$ does not result in a  change propagation to its environment component ($C_2$) and no more adaptation is required. Otherwise, the change is propagated to $C_2$. This case shows that the provided adaptation changes the observable behavior of $C_1$, and $C_2$ has to be adapted to consider the new behavior of $C_1$.

%As described, adapting a component results in a new network of TIOAs for the component. Upon occurring a change, the model@runtime is updated based on a snapshot taken from the system.  Then, the adaptation mechanism contributes to obtaining new static plans based on which the moving objects are routed/scheduled. We can abstract away the rerouting/rescheduling of the moving objects, and present a simple TIOA to model an actor. This simplification is because the compositional analysis proposed here only employs the model@runtime that consists  of the actors and the scheduler of the coordinator (dispatching part). Therefore, the semantics of a component is presented based on a network of TIOAs, where TIOAs of the actors are simplified as shown in Fig.~\ref{fig:distributedTA2}. Compared to Definition.~\ref{actor::dec}, this simple TIOA of the actor does not model the adaptation mode of the actor. We assume that the traveling time of a moving object through a sub-track is a fixed value denoted by $d$ in Fig.~\ref{fig:distributedTA2}.

Although the proposed approach  works effectively for the small systems, checking the compatibility in a system with several components is an expensive process since the product of TIOAs of the adapted component and  TIOAs of its environment components may result in a large state space. To reduce the state space in our analysis, %we propose considering the observable parts of the environment components to the component. 
we propose that instead of an environment component, its observable part to the component, which is the set of environment actors, is considered in the product.  %Therefore, an environment component is  abstracted to the set of environment actors that directly communicate with the component. %Based on this description, we use interface components of a component in the product.

As previously mentioned, the initial traveling plan of a moving object imposes constraints on the arrival of the moving object at each area of its route. The moving object arrives at each area of its route at a pre-specified time and from a pre-specified direction. %, where each area is modeled by a CAM as a component. 
Regarding the arrivals of the moving objects at a component, the environment actors send messages at the pre-specified times to the component, and regarding the departures of the moving objects from the component, the environment actors  receive messages at the pre-specified times from the component. Therefore, the environment actors should abstract the environment of the component to the messages that are sent and received at the pre-specified times. To this end, in Magnifier an environment actor is replaced with an augmented environment actor that has the \textit{ERS} list. % we suppose that  an environment actor has a list of messages named \textit{ETRS}. 
An entry of \textit{ERS} either specifies that the augmented environment actor expects
to receive a message over an input port after an amount of time, or specifies that the actor intends to send a message over an output port after an amount of time. 
 %to a message that is expected to be received by the environment actor through an input port after an amount of time, or specifies a message that is intended to be sent by the environment actor over an output port after an amount of time. Let each entry of the list be an instance of $\mathit{En}=\MSG \times \R \times \mathit{CH}$. 
As the sub-track corresponding to an (augmented) environment actor is a critical section, \textit{ERS} is an ordered list, where the first entry contains the message which should  be  sent or received by the actor first.  Suppose that the augmented environment actor has received or sent the message of the first entry of \textit{ERS} at time $t$. It will send or receive the message of the second entry at time $t+t'$ where $t'$ is kept as a delay value in the second entry of \textit{ERS}. %The difference between $t'$ and $t$ is at least equal to the traveling time of the moving object across a sub-track, . 
The same argument is valid for the rest of the entries. As \textit{ERS} in a model of a track-based system is calculated from the initial traveling plans of the moving objects, the schedules of the moving objects in \textit{ERS} do not lead to any conflicts between the moving objects. 

%Because of the mentioned abstraction, a new TIOA for an environment actor augmented with the \textit{ERTS} list is defined. We use the following auxiliary functions in the definition. The function $\mathit{timeProg}: 2^{\mathit{En}} \rightarrow \R$ returns the time after which the message of the first entry of the given list should be sent or received. The function $\mathit{in}: \AID \times 2^{\mathit{En}} \rightarrow \mathit{bool}$ returns true if the first entry of the given list has a channel connected to an input port of the given actor. Similarly, the function $\mathit{out}: \AID \times 2^{\mathit{En}} \rightarrow \mathit{bool}$ returns true if the first entry of the given list has a channel connected to an output port of the given actor. The function $\mathit{expCh}: 2^{\mathit{En}} \rightarrow \mathit{CH}$ returns the channel of the first entry of the given list. The function $\mathit{expMes}: 2^{\mathit{En}} \rightarrow \MSG$ returns the message of the first entry of the given list. The function $\mathit{hasMes(ETRS)}$, where  $\mathit{hasMes}:2^{\mathit{En}} \rightarrow \mathit{bool}$, returns true if the actor connected to the source of the channel of the first entry in \textit{ETRS} has the same message as the first entry of \textit{ETRS}. Finally, the function $\mathit{updateL}: 2^{\mathit{En}} \rightarrow 2^{\mathit{En}}$ gets a list, removes the first entry, and returns the rest of the list. The TIOA of an environment actor \COM{in the model@runtime} is defined as follows.

In the following, we define TIOA of an augmented environment actor $aa_j$ that is used in the model@runtime. As shown in Fig.~\ref{fig::automata2}\protect\subref{fig::automata::environment}, this automaton has an  \textit{ERS} variable that corresponds to the \textit{ERS} list of the actor. The automaton has a set of input actions $\mathit{interact_{ch}}$, where \textit{ch} is a channel bound to an input port of the actor. Similarly, for each channel \textit{ch} bound to an output port of the actor an output action $\mathit{interact_{ch}}$ is defined. This automaton has access to the global variable \textit{transP} that corresponds to the input argument of the interact handlers and  is  used  to  transfer  a  value  between  TIOAs  of  an actor of CoodAA and the augmented environment actor.  %Because of the mentioned abstraction, a new TIOA for an environment actor augmented with the \textit{ERTS} list is defined as follows. %The TIOA of an environment actor \COM{in the model@runtime} is defined as follows.

\begin{definition} \label{actor::enviActor}
	(TIOA of an Augmented Environment Actor) The TIOA associated with an augmented environment actor $\mathit{aa}_j$ is $TA=(\{q_0\}, \mathit{q_0},\mathit{Var},\{\mathit{clock}\},\mathit{Act_{\mathit{in}}},\mathit{Act_{out}},T,I)$, where $\mathit{Var}=\{\mathit{ERS}_j\}$, $\mathit{Act_{\mathit{in}}}=\{\mathit{interact_{ch}}|\mathit{ch} \in \mathit{inChS(j)}\}$, $\mathit{Act_{out}}=\{\mathit{interact_{ch}}| \mathit{ch} \in \mathit{outChS(j)}\}$, $I(\mathit{q_0})=\mathit{clock} \leq \mathit{timeProg(ERS_j)}$, and $T$ is defined as follows.
	\begin{multline*}
	\tag{Send}
	(q_0,\mathit{\mathit{clock}=timeProg(ERS_j)} \wedge \mathit{out(j,ERS_j)}, \\ \mathit{interact_{\mathit{expCh(ERS_j)}}},\emptyset,\{\mathit{transP=expMes(ERS_j)},\\\mathit{ERS_j=updateL(ERS_j)}\},q_0) 
	\end{multline*}
	\begin{multline*}
	\tag{Receive}
	(q_0,\mathit{\mathit{clock}=timeProg(ERS_j)} \wedge \mathit{in(j,ERS_j)} \wedge \\ \mathit{hasMes(ERS_j)},\mathit{interact_{expCh(ERS_j)}},\{\mathit{clock}\},\{\mathit{ERS_j=}\\\mathit{updateL(ERS_j)}\},q_0)
	\end{multline*} \qed  
\end{definition}

In the following, we describe edges of TIOA of $\mathit{aa}_j$. If a message is present for $\mathit{aa}_j$ over an input port, $\mathit{aa}_j$ receives the message, and the interact handler corresponding to that  input port is triggered. In this case, the input arguments of the handler are not null. Besides the \textit{init} method of the actor, an interact handler can trigger an interact handler by passing null values as input arguments.  The \textit{init} method of $\mathit{aa}_j$ is executed whenever the actor is initialized. This method triggers an interact handler of $\mathit{aa}_j$. %by setting the input arguments of the handler to null. 
Let \textit{e}=(\textit{msg}, \textit{t}, \textit{ch}) be the first entry of the \textit{ERS} list of the actor. By triggering an interact handler, the actor introduces a delay with the value of \textit{t} if the input arguments of the handler are null. This operation of the actor is formulated in the $q_0$ location of the automaton shown in Fig.~\ref{fig::automata2}\protect\subref{fig::automata::environment}.  The location $q_0$ ensures that the time progresses up to \textit{t}  at which the message \textit{msg} is sent or received. Let $\mathit{En}$ be the set of all entries of all \textit{ERS} lists in the model.  The function $\mathit{timeProg}(\mathit{ERS})$, where  $\mathit{timeProg}: 2^{\mathit{En}} \rightarrow \R$, returns $t$. After the time is progressed by  $t$ time units, the message \textit{msg} should be sent or received. 

%\COM{The readability of the following two paragraphs can be improved.}
If \textit{ch} is bound to an input port of the actor, the actor expects to receive the message \textit{msg}. In this case, the interact handler terminates and the actor waits to receive a message. If \textit{ch} is bound to an  output port of the actor the  message \textit{msg} should be sent over that port. %The actor sends the message over the output port bound to \textit{ch}. %The message \textit{msg} is send over an output port if \textit{ch} is bound to that output port of the actor. %The output port of the actor over which the message is sent is the port bound to the channel \textit{ch}
After sending the message, \textit{e} is removed from the \textit{ERS} list, and the actor invokes the current executing interact handler with null values as arguments. This way, the same process for the first entry of \textit{ERS} repeats. As shown in   Fig.~\ref{fig::automata2}\protect\subref{fig::automata::environment}, the \textit{Send}  edge is enabled if $\mathit{out}(j,\mathit{ERS}_j)$, where $\mathit{out}: \AID \times 2^{\mathit{En}} \rightarrow \mathit{bool}$, returns true. The function $\mathit{out}(j,\mathit{ERS}_j)$ determines if \textit{ch} is bound to an output port of $\mathit{aa_j}$. 
The automaton of $\mathit{aa}_j$ synchronizes over the action $\mathit{interact_{\mathit{expCh(ERS_j)}}}$ to send the message $\mathit{expMes(ERS_j)}$ over the channel ${\mathit{expCh(ERS_j)}}$. The function $\mathit{expMes(ERS_j)}$, where  $\mathit{expMes}: 2^{\mathit{En}} \rightarrow \MSG$, returns the message \textit{msg}. The function ${\mathit{expCh(ERS_j)}}$, where $\mathit{expCh}: 2^{\mathit{En}} \rightarrow \mathit{CH}$, returns the channel \textit{ch}. %The function $\mathit{out}: \AID \times 2^{\mathit{En}} \rightarrow \mathit{bool}$ returns true if the first entry of a  given list has a channel bound to an output port of a given actor.  %In fact,  $\mathit{out(j,ETRS_j)}$ determines whether the actor $a_j$ expects to send a message.
  This edge uses $\mathit{updateL(ERS_j)}$, where $\mathit{updateL}: 2^{\mathit{En}} \rightarrow 2^{\mathit{En}}$, to remove \textit{e} from  the \emph{ERS} list and returns the rest of the list.   %by removing its first entry. The function $\mathit{updateL}: 2^{\mathit{En}} \rightarrow 2^{\mathit{En}}$ gets a list, removes the first entry, and returns the rest of the list.
 
 %If the channel \textit{ch} is bound to an input port of the actor, the actor expects to receive the message \textit{msg}. 
 %If a message is present for the actor over an input port, the actor receives the message, and the message handler corresponding to that input port is triggered. Since the input arguments are not null, 
 By triggering an interact handler, if the input arguments of the handler are not null,  the actor checks whether it has received the message \textit{msg} over the  input port bound to \textit{ch}. %,  and whether the received message is equal to \textit{msg}. 
 An error is thrown if this condition does not hold. Otherwise, \textit{e} is removed from \textit{ERS}, and the actor invokes the current executing interact handler. This way, the same operation for the first entry of \textit{ERS} repeats. As shown in Fig.~\ref{fig::automata2}\protect\subref{fig::automata::environment}, the \textit{Receive} edge is enabled if the functions \textit{in} and \textit{hasMsg} return true.
 %
 %These operations are specified by the \textit{Receive} edge. The automaton synchronizes on $\mathit{interact_{expCh(ERS_j)}}$ with the  automaton of an actor that sends the message. %, if this is the time for $\mathit{aa}_j$ to receive the message  over an expected port from the actor. %whose bound channel is determined by the function \emph{expCh(ETRS)}.  
% The \textit{Receive} edge is enabled if the functions $\mathit{in}(j,\mathit{ERS}_j)$ and $\mathit{hasMes(ERS_j)}$ return true. 
The function $\mathit{in}(j,\mathit{ERS}_j)$, where $\mathit{in}: \AID \times 2^{\mathit{En}} \rightarrow \mathit{bool}$, returns true if \textit{ch} is bound to an input port of the actor. The function $\mathit{hasMes(ERS_j)}$, where $\mathit{hasMes}:2^{\mathit{En}} \rightarrow \mathit{bool}$, returns true if the message to be received is equal to \textit{msg}. The automaton synchronizes on $\mathit{interact_{expCh(ERS_j)}}$ to receive \textit{msg} over the channel \textit{ch}, which is determined by $\mathit{expCh(ERS_j)}$. %, if this is the time for $\mathit{aa}_j$ to receive the message  over an expected port from the actor. % the first entry in \textit{ERS} has the same message as the first entry of \textit{ERS}.  %the actor contains the message the environment actor expects, checked by the function \emph{hasMsg}. The functions used over the \textit{Receive} edge are described as follows.   
 %The function $\mathit{in}: \AID \times 2^{\mathit{En}} \rightarrow \mathit{bool}$ returns true if the first entry of a given list has a channel bound to an input port of a given actor. %In fact, the  function \textit{in(j,ETRS)} determines whether the actor $a_j$ expects to receive a message. 
 %The function $\mathit{hasMes(ERS_j)}$, where  $\mathit{hasMes}:2^{\mathit{En}} \rightarrow \mathit{bool}$, returns true if the actor connected to the source of the channel of the first entry in \textit{ERS} has the same message as the first entry of \textit{ERS}. 
 In the case of an error is thrown, TIOAs of $\mathit{aa_j}$ and the sender actor reach a deadlock state.

In Magnifier, an interface component is described  by a network of TIOAs where each TIOA models an augmented environment actor. Suppose that a change happens to a component, and the component is adapted to the change. If the network of TIOAs of the adapted component and the networks of TIOAs of its interface components are compatible, the change does not propagate to the environment components. Otherwise, the change is propagated to the environment components. %The structure of track-based systems enables us to focus on a component and define an interface component for each one of its environment components. 
Since each component has a well-defined interface, we are able to focus on a component and define an interface component for each one of its environment components.
This way, we are able to find the direction where the change is propagated and to find the components affected by the change propagation. The change propagates from the component $C_i$ to the component $C_j \in \mathit{Env(C_i)}$ if there is an actor of $ C_j \downarrow_{C_i}$ such that no transition is performed in the location $q_0$ of its TIOA  (Definition.~\ref{actor::enviActor}) even after letting time progress. This   shows a deadlock in the product of TIOAs of the component $C_i$ and the interface component $ C_j \downarrow_{C_i}$, and means that this augmented environment actor is not able to either send a message over a pre-determined port at a pre-specified time, or receive an expected message from a pre-determined port at a pre-specified time. In our approach, by propagating the change, all  components affected by the change are composed to create a new component that is adapted (please note that an actor follows the semantics defined in Fig.~\ref{fig::automata2}(b) whenever it belongs to an interface component, but all actors have the semantics shown in Fig.~\ref{fig::automata2}(a) whenever two components are composed to make a new  component). It is then checked whether the network of TIOAs of the new component and networks of TIOAs of its interface components are compatible. 

%The reason for composing the changed components in our approach is the circular dependency between the components. The change propagated from a component to one of its environment components may propagate back to the component. This means that the component should be adapted again and compatibility of networks of its TIOAs and TIOAs of its interface components should be checked. This case shows the circular dependency between two components. %The compositional verification is not applicable when the circular dependency appears among the components. 
%By composing the components to create a new component, all changes circulating between two components happen inside of the new component and their effects are considered.

%\textbf{Example.} For a better understanding of the approach, consider the following example.
%
\begin{example}
For a better understanding of the approach, consider the following example.
%A model with five components is shown in Fig.~\ref{interactiveComponents}. 
%The interactions between the components are denoted by arrows. 
Suppose that a change in the component $C_1$ of Fig.~\ref{interactiveComponents}(a) is detected and this component (its model$@$runtime) is adapted. If  $\N_{C_{a,1}} \bowtie \N_{C_2 \downarrow_{C_1}} \bowtie \N_{C_3 \downarrow_{C_1}} \bowtie \N_{C_4 \downarrow_{C_1}}$ %$\Pi_i \updownarrow_{ \mathit{Act_{I_1}} \cup \mathit{Act_{Out_1}}} \bowtie \Pi_{a,1} \updownarrow_{ \mathit{Act_{I_i}} \cup \mathit{Act_{Out_i}}}$ 
does not hold, the change propagates into some of the environment components. Let the change propagates to the components $C_2$ and $C_3$. It means that $C_1$ with its current adaptation is not able to either receive  messages from $C_i, i=2,3,$ or send messages to $C_i$ at the pre-specified times. %For instance, if $\Pi_3 \updownarrow_{\mathit{Act_{I_1}} \cup \mathit{Act_{O_1}}} \bowtie \Pi_{a,1} \updownarrow_{ \mathit{Act_{I_3}} \cup \mathit{Act_{O_3}}}$ %does not reach the all-stop state, does not hold, $C_1$ is not able to either receive messages from $C_3$ or send messages to $C_3$. 
Consequently, the adaptation for $C_2$ and $C_3$ is triggered.  %Suppose that $\Pi_3 \updownarrow_{\mathit{Act_{O_1}} \cup \mathit{Act_{I_1}}}  \bowtie  \Pi_{a,1} \updownarrow_{\mathit{Act_{I_3}} \cup \mathit{Act_{O_3}}}$ and  $ \Pi_2 \updownarrow_{ \mathit{Act_{O_1}} \cup \mathit{Act_{I_1}}} \bowtie \Pi_{a,1} \updownarrow_{\mathit{Act_{I_2}} \cup \mathit{Act_{O_2}}}$ do not hold, %do not reach the all-stop states, and the change is propagated into $C_3$ and $C_2$. 
The components $C_1$, $C_2$, and $C_3$ are adapted and are composed to provide the new  component $C_{1,2,3}$. %,  consisting of a coordinated actor model for each component.
If  $\N_{C_5\downarrow_{C_3}}  \bowtie \N_{C_4\downarrow_{C_1}}  \bowtie \N_{C_{1,2,3}}$ holds, %$\Pi_5 \updownarrow_{ \mathit{Act_{O_3}} \cup \mathit{Act_{I_3}}}  \bowtie  \Pi_{1,2,3} \updownarrow_{\mathit{Act_{I_5}} \cup \mathit{Act_{O_5}}}$ and $ \Pi_4 \updownarrow_{ \mathit{Act_{O_1}} \cup \mathit{Act_{I_1}}}  \bowtie \Pi_{1,2,3} \updownarrow_{\mathit{Act_{I_4}} \cup \mathit{Act_{O_4}}}$ %reach the all-stop states,hold,
the change propagation stops, and the change is not propagated further than $C_1$, $C_2$, and $C_3$.
\end{example}
%\begin{figure}
%	\centering
%	\includegraphics[width=3cm,height=3cm,keepaspectratio]{Fig/components}
%	\caption{A model consisting of 5 interactive components}
%	\label{interactiveComponents}
%\end{figure}

%The reason for composing the changed components in our approach is the circular dependency between the components. The change propagated from a component to one of its environment components may propagate back to the component. This means that the component should be adapted again and compatibility of networks of its TIOAs and TIOAs of its interface components should be checked. This case shows the circular dependency between two components. The compositional verification is not applicable when the circular dependency appears among the components. By composing the components to create a new component, all changes circulating between two components happen inside of the new component and their effects are considered. %An example of the circular dependency is shown in Fig.~\ref{ttcs}, where by propagating the change from $C_1$ to $C_2$, the change propagates back to $C_1$, since the blue aircraft arrives at $C_1$ at time 7 instead of time 9. 

%\begin{comment}

\subsection{Correctness of the Proposed Approach}

In the previous section, we defined interface components of an adapted component, where each interface component is an abstraction of an environment component. In the absence of a change, the correctness properties of the system are preserved as each component of the system preserves a set of local correctness properties, i.e. each component receives and sends messages at the pre-specified times and over the pre-specified ports. In the presence of a change, the correctness properties are satisfied if the adapted component can work with its environment, i.e. the adapted component and its environment satisfy their input and output assumptions. This is where the networks of TIOAs of the adapted component and its interface components are compatible. %The adapted component is not necessarily an individual component. It can be created by composing several components. %As previously mentioned, the set of correctness constraints in a \did{track-based system} includes the moving objects have to arrive at their destinations at the pre-specified times, the moving objects should not run out of fuel, the conflict in the system should be avoided, and the system should be deadlock-free. Furthermore, each component itself satisfies a set of local correctness constraints, i.e. each component should receive and send messages at the pre-specified times and over the pre-specified channels. In the case of a deadlock or a conflict occur in an area, or a moving object runs out of fuel, the component corresponding to the area does not send out the messages corresponding to the moving objects involved in the event. 
In this section, the correctness of the proposed approach is proved in Theorem.~\ref{correctness}. This theorem explains that reducing the environment components to the interface components is correct. 

\begin{theorem}
	The networks of TIOAs of the adapted component and its interface components are compatible if and only if the networks of TIOAs of the adapted component and its environment components are compatible.
	\label{correctness}
\end{theorem}
\begin{proof}
	"if": By contradiction. Suppose that the networks of TIOAs of the adapted component $C_i$ and its environment components are compatible, but $\N_{C_i} \bowtie \N_{C_{j_1} \downarrow_{C_i}} \bowtie \cdots \bowtie \N_{C_{j_n} \downarrow_{C_i}}$ , where $C_{j_k} \in \mathit{Env(C_i)}$ and $k=1,\cdots,|\mathit{Env(C_i)}|$, does not hold. %It means that  there exists $C_j \in \mathit{Env(C_i)}$ such that $\N_{C_j \downarrow_{C_i}} \bowtie \N_{C_i}$ does not hold. 
	It means that there exists an interface component $C_{j_l}\downarrow_{C_i},  C_{j_l} \in \mathit{Env(C_i)}$, and an augmented environment actor $\mathit{aa}_j \in {C_{j_l}\downarrow_{C_i}}$ such that this actor is  not able to either  receive an expected message from an expected port at a pre-specified time or send a message over a pre-specified port at a pre-specified time. Let $\mathit{aa_j}$ corresponds to the environment actor $a_j$. However, the actor $a_j$ belongs to $C_{j_l}$, e.g. $a_j \in A_{j_l}$. This means that $\N_{C_i} \bowtie \N_{C_{j_1}} \bowtie \cdots \bowtie \N_{C_{j_n}}$ does not hold, which contradicts the assumption. 
	
	"only if": By contradiction. Let the networks of TIOAs of the adapted component $C_i$ and its interface components be compatible, but $\N_{C_i} \bowtie \N_{C_{j_1}} \bowtie \cdots \bowtie \N_{C_{j_n}}$ , where $C_{j_k} \in \mathit{Env(C_i)}$ and $k=1,\cdots,|\mathit{Env(C_i)}|$, does not hold. We assumed that the adaptation results in a new network of compatible TIOAs for the component $C_i$. Furthermore, as each component $C_j \in \mathit{Env(C_i)}$ is not yet affected by a change, all TIOAs in $\N_{C_j}$ are compatible. Therefore, there exists $C_j \in \mathit{Env(C_i)}$ and an  environment actor $\mathit{a}_j \in A_j$ such that this actor is not able to either receive an expected message from an expected port at a pre-specified time or send a message over a pre-specified port at a pre-specified time.  However, this actor corresponds to an actor of $C_j \downarrow_{C_i}$. This means that $\N_{C_i} \bowtie \N_{C_{j_1} \downarrow_{C_i}} \bowtie \cdots \bowtie \N_{C_{j_n} \downarrow_{C_i}}$ does not hold, which contradicts the assumption. 
\end{proof}

%% file: 021-Magnifier.tex
\subsection{Overview} \label{sec::Mag}
In this section, we present an overview of the Magnifier approach. We first informally explain the approach on track-based systems. When a track-based system is designed, initial traveling plans of the moving objects are selected in a way that no conflict happens between the moving objects, and the moving objects  arrive at their destinations at the pre-specified times. In fact, the initial traveling plan of a moving object imposes constraints on the arrival of the moving object at each area of its route. When a change happens to an area,  the moving objects traveling across the area are rerouted  if there is an unavailable sub-track in their routes.
%
%the moving objects entering into the area with the unavailable sub-tracks in their routes are rerouted. 
This way the plan for the area is adapted. 
Note that the presence of a change (or its effect) in an area 
 may last for a while, so any change in the plan must consider the possible future effects. 
Under the following conditions, the correctness properties of the system are satisfied and the change does not propagate: 
\begin{enumerate}
	\item The moving objects traveling across the area depart from it based on their initial traveling plans.
	\item The moving objects intending to enter into the area at a future time, arrive at the area based on their initial traveling plans.
	\label{item2}
	\item The moving objects entering into the area at a future time, depart from the area based on their initial traveling plans.
	\label{item3}
\end{enumerate}
%First, the moving objects traveling across the area depart from it based on their initial traveling plans. Second, the moving objects intending to enter into the area at a future time, arrive at the area based on their initial traveling plans. Third, the moving objects of the second condition depart from the area based on their initial traveling plans. 
In the case of violating one of the above conditions, the change is propagated to the adjacent areas.

There is a set of %correctness 
properties for track-based systems that has to be satisfied.
Collision of moving objects is avoided by design, a sub-track can only contain one moving object at one moment in time. 
Moving objects must arrive at their destinations at the pre-specified times. This property is checked by Magnifier. We also check that the system is deadlock-free.
A deadlock in a track-based system happens whenever  moving objects are stuck in a traffic blockage and cannot find an available route towards their destinations. If a moving object is stuck in some area of the traveling space,  then it does not depart from the area. This causes propagation of the change  through the whole system, and the object does not arrive at its destination on time. Deadlock can be caught by Magnifier. 
We also check if the fuel of  moving objects go under a certain threshold.
We use separate functions to detect a deadlock or running out of fuel and to stop the analysis.

The adaptation for an adjacent area is triggered whenever the change propagates into the area. This means that the traveling plans of the moving objects entering into the area or traveling across the area may need to be adapted. Therefore,  in the case of propagating the change, all areas affected by the change are composed to form a new area. The traveling plans of the moving objects traveling across the new area are adapted. If the moving objects arrive at the new area and depart from it based on their initial traveling plans, the change propagation stops. 
As explained informally for track-based systems, the Magnifier approach uses an iterative algorithm to assure the correctness of the system by involving the least number of the components in the analysis. %Each component in the Magnifier approach is modeled by a coordinated actor model. 
%\marjan{The following paragraph somehow is very similar to the previous ones ... and what we said in the intro, or not?}
%\COM{Or maybe this one is good but remove the one short paragraph before this one.}
%
%Therefore, Magnifier thinks of each coordinated actor model as a component and abstracts t
The  environment of a component is abstracted to the external messages that are sent to the component at the pre-specified times. When a change occurs, Magnifier zooms-in on the component affected by the change. The model@runtime corresponding to the component is updated based on the snapshot taken from the system at the change point. The model@runtime is adapted based on an adaptation policy defined in the coordinator. The adaptation contains the change if the adapted component can work with its environment by satisfying the constraints on their interactions. In fact, the propagation of the change stops if the adapted component is able to receive messages of its environment at the pre-specified times and over pre-specified ports, and guarantees to deliver messages to its environment at the pre-specified times and over pre-specified ports. Otherwise, the change is propagated to the environment. In this case, Magnifier zooms-out and  concentrates on a new component resulted from composing all components affected by the change. The new component is adapted. The same procedure is repeated for the new component. In other words, the propagation of the change stops if the new component can work with its environment. Otherwise, Magnifier zooms-out and extends its verifying domain by composing the components affected by the change propagation and the previous ones. 

One can argue that in a compositional approach we can check the change in one component and then check its propagation to the neighborhood components one by one.
But a change may propagate back to the component which was the source of the change and develop a circular dependency. This situation is shown in the example of Section~\ref{section::problemDef}.
In Magnifier,
by composing the components and forming a new component, all changes circulating between two components happen inside of the new component and their effects are considered.

%% file: 15-Evaluation.tex
\section{Implementing and Evaluating Magnifier Using Ptolemy II} \label{sec::evaluation}

In this section, we briefly describe the implementation of Magnifier 
for an ATC case study with several control areas in Ptolemy II \cite{ptolemaeus2014system} as a proof of concept for effectiveness and efficiency of the work. %This prototype is implemented in Ptolemy II \cite{ptolemaeus2014system}. % using Ptolemy provides us two main advantages. The first is that an adaptation plan is automatically  developed and is applied to the components whenever it is needed. The second is that the verification domain can be automatically extended to bring more components into the analysis.
The main reason for using Ptolemy II instead of \textsc{Uppaal} is that Ptolemy enables us to automate the iterative and incremental process of Magnifier using its so called \textit{director}. The change propagated through the system can be automatically traced, and the  verification scope can be extended to bring more components into the analysis. Here, we first give a background on Ptolemy II. We then  describe our implementation and  compare the time consumption and the memory consumption between the compositional and non-compositional approaches. The Ptolemy model and implementations of the provided algorithms in this section are available online\footnote{ \url{http://www.ce.sharif.ir/~mbagheri/MagImp.zip}}.

\subsection{Background: Ptolemy Framework} \label{subsec::Ptolemy}

%Ptolemy II is an actor-oriented open source modeling and simulation framework. A Ptolemy model consists of actors that communicate via message passing. The semantics of communications of the actors in Ptolemy is defined by models of computation, implemented in a set of predefined director components. %Ptolemy provides a library of MoCs.

Ptolemy II \cite{ptolemaeus2014system} is an actor-based modeling and simulation  framework that provides different models of computation  with fully deterministic semantics. A model of computation defines the semantics of interactions of  actors, and is implemented in a Ptolemy director. In \cite{cooda}, we developed a Ptolemy template based on the coordinated adaptive actor model to model and analyze self-adaptive TTCSs. In this template, each sub-track is modeled by a Ptolemy actor, and the moving objects are considered as messages passing among the actors. The pathways between the sub-tracks are modeled by interconnections between the actors. The Ptolemy actors are connected through channels. An actor can read a message form an input channel and send a message over an output channel.  Furthermore, the controller (coordinator) is modeled by a Ptolemy director. In \cite{cooda}, we developed a director by extending the Discrete Event (DE) director. DE is one of the most commonly used models of computation in Ptolemy. 
%We use the Discrete Event (DE) model of computation in our implementation.

The DE director has a buffer of events. Each message communicated in the model is packaged in an event. An event besides a message  has a time tag and a reference to the receiver actor in the communication. In fact, instead of direct message passing among the actors, upon sending a message, an event is created and is placed into the internal buffer of the director. The director keeps the model time and stamps the event with the model time at which the message has been sent. 
An actor can execute the \textit{fireAt} instruction to ask the director to trigger the actor at a future time. In this case, an event with an empty message is generated for the director. 
%An event with empty message is also generated when an actor executes an instruction to ask the director to trigger the actor at a future time. 
The director labels the event with the requested future time. 
All the events tagged with the value of the current model time are enabled events.
The director takes an enabled event  from
its buffer, triggers the actor referred to by the event, and delivers the message to the actor. 
To choose between a set of enabled events, a deterministic policy (the so-called topological sort) is used.  
%In fact, it uses a deterministic policy (the so-called topological sort) to take all events with the time equal to the model time and to process them at the current time. 
If there is no enabled event, the model time progresses to reach the time of the event with the smallest time tag. 

\subsection{Ptolemy Implementation of Magnifier}

In \cite{cooda}, we implemented CoodAA and its MAPE-K architecture by extending the DE director in Ptolemy II. We modeled  ATC  and  railroad examples, and we checked different rerouting plans in different situations. The difference between the examples is in the number of ports of actors and the topology of their bindings. 
Our analysis in \cite{cooda} was based on the simulation engine provided by  
Ptolemy II.
In this paper, we use the Ptolemy II template proposed in \cite{cooda}, and extend the DE director to develop  the Magnifier director that supports formal verification. In each iteration, after replanning, Magnifier builds the state space to check the compatibility of components.

%because Ptolemy II with its deterministic models of computation does not support verification of a system. In this paper, we use the Ptolemy template proposed in \cite{cooda}, but we extend the DE director to develop  the Magnifier director, which supports verification of the system.
%
% Our Ptolemy implementation of the coordinated adaptive actor model
Our implementation is faithful to the TIOA semantics of Section~\ref{section::sem}. As described in Section~\ref{subsec::Ptolemy}, the Ptolemy director keeps a buffer of events and therefore it can adapt the network by manipulating the messages encapsulated in the events. The Magnifier director mimics the non-deterministic selection of actors for being executed. 
From a set of enabled events, the Magnifier director takes one of them non-deterministically. Finally, if there is no enabled event (i.e. with the  time tag equal to the current model time), the model time progresses. 
This non-determinism is shown in Fig.~\ref{fig::automata2}\protect\subref{fig::automata::simplified} since from a set of enabled edges in a network of TIOAs (i.e. multiple actors), an edge is selected non-deterministically. 

In our implementation, each actor (sub-track) is modeled by a Ptolemy actor. An actor is first triggered to receive the message corresponding to a moving object. Then, the actor executes the fireAt instruction and asks the director to trigger it again at a future time to send its message out. This way the traveling time of the moving object through the sub-track is modeled. %. This means that the moving object departs from the sub-track.
In the proposed semantics in Section~\ref{section::sem}, there is no centralized buffer as messages are directly delivered to the actors at the same model time.
In our implementation, like in our semantics, the events created per sending a message are processed at the same model time at which the message has been sent.  %The variable \textit{transP} allows two actors to communicate with each other. Each actor instead of requesting a director to trigger the actor at a future time, is associated with a clock to maintain the actor for a while at its current location. 
The global model time in Ptolemy mimics the synchronous progress of time of clocks in a network of TIOAs. 

The magnifier director  provides the assertion-based verification. It generates the state space of a given component, and performs the reachability analysis. 
For the sake of simplicity, we assume that all the coordinators of all components (the ATC controllers of all areas) have the same adaptation policy (rerouting algorithm). This way, we have only one coordinator (instead of a nested model and multiple coordinators). 
%Note that to have a simple model, we develop a model with a single coordinator, because ATC controllers of all areas in our case study have the same adaptation policy (rerouting algorithm). %The coordinator in this case has also the role of the top-level coordinator. 
The Magnifier director generates the state space of the model of 
an ATC example 
with several components, where the components are composed to create a new  component. The rerouting algorithm and the algorithm given in the following section are implemented in the director.  It is notable that designing the rerouting algorithm is not the concern of this paper. 

\subsubsection{Generating the State Space}
Here, we explain the algorithm to generate the state space. We also present the pseudocode of this algorithm. 
%\COM{ Add something like this: Here we explain the algorithms for generating the state space and its main function, we also include the pseudocode for these algoritms. These algorithms are based on standard model checking algorithms. }

\begin{comment}
\begin{algorithm}
\DontPrintSemicolon 
\SetKwFunction{algo}{All}
\KwIn{\emph{aircraft} as a set of aircraft}
\KwOut{\emph{plans} as initial plans of the aircraft}
\Begin{
$\mathit{plans} \gets \emptyset$\;
\ForEach{$a \in \mathit{aircraft}$}
{	
$(x_s,y_s) \gets \mathit{genSou()}$\;
$(x_d,y_d) \gets \mathit{genDes()}$\;
$\mathit{aT} \gets \mathit{depTime((x_s,y_s),(x_d,y_d),\mathit{plans})}$\;
$\mathit{plans} \gets \mathit{plans} \cup \{\mathit{route((x_s,y_s),(x_d,y_d),\mathit{dt},\mathit{plans})}$\;

}
\nl \KwRet $\mathit{timeStates}$\;
}
\caption{ALG1}
\label{algo::ALG1}
\end{algorithm}
\begin{algorithm}
\DontPrintSemicolon 
\SetKwFunction{algo}{All}
\KwIn{$(x_s,y_s)$ as the source of an aircraft, $(x_d,y_d)$ as its destination, $\mathit{aT}$ as its arrival time at the source, and $\mathit{plans}$ as the existing traveling plans}
\KwOut{\emph{route} as the route of the aircraft}
\Begin{
\If{$\mathit{containTime((x_s,y_s),\mathit{aT})}$}{
\nl \KwRet $\mathit{null}$\;
}
}
\caption{ALG1}
\label{algo::routing}
\end{algorithm}
\end{comment}

\textbf{Algorithm.} %The fire function of a coordinator works in several iterations. Each iteration starts with a time progress followed with firing of the actors whose corresponding events have time information equal to the current time of the model. Therefore, an iteration finishes if there is no actor to be fired at the current time of the model. The concurrent execution of the actors in an iteration leads to several execution traces. Each trace ends with a state that is the initial state of the next iteration if the model does not terminate. 
%The algorithm uses Depth-First Search (DFS) and Breadth-First Search  (BFS) to generate the state space of a given component. 
The algorithm to generate the state space of a component is shown in Algorithm.~\ref{algo::all}. Let the initial state of the component be a timed state. We call a state a timed state if a time transition is enabled at the state. The algorithm uses a queue to store the timed states (line 3).  %It  uses Breadth-First Search  (BFS) to process the timed states stored into the queue (line 4). 
It dequeues a timed state (line 5), and after progressing the time (line 6), uses  Depth-First Search (DFS) to generate all the traces starting with that state and ending with the new timed states (line 8). The algorithm terminates whenever no new timed state is generated (lines 9-12). Otherwise, the generated timed states are added to the end of the queue (line 13). %It again dequeues a timed state, uses  DFS to generate all the traces starting with that state and ending with the new timed states. 
The function \emph{timeProg} progresses the time to the smallest time of the events stored in the buffer of the director in state $s$, and the function \emph{deQueue} removes and returns the first state of the queue.

The main computation of the algorithm is performed by the function \emph{depthFS}, presented in Algorithm.~\ref{algo::dfs}. %This recursive function takes a state as input and generates all the traces starting with the state and ending with the timed states. The generated timed states are returned. 
Let \emph{buffer(s)} denotes the buffer of the events kept in the director in state $s$. The function \emph{depthFS} takes an event (line 11), and using the function \emph{trigger}, triggers the actor referred to by the event to generate the next state (line 12). Compared to state $s$, the buffer of the director in the next state stores the events that are possibly created by triggering the actor. Furthermore, the taken event is removed from the buffer in the next  state. %It triggers the actor which can be triggered in the initial state of the component to generate the next state. 
The state $s$ has several outgoing transitions (resp. several next states) if several actors can be triggered at the state. The next state is returned if it is a timed state (lines 7-10). Otherwise, the actors which can be triggered in the next state are triggered.  The algorithm to generate the state space terminates whenever one of the following conditions is fulfilled: all the moving objects supposed to travel through the component depart from it (reach their destinations), a disaster happens (i.e. the fuel of a moving object is zero), and the analysis time passes a threshold. These conditions are checked using the \emph{terminate} function over a state (lines 3-6). 

%Therefore, the algorithm uses Depth-First Search (DFS) to generate all the traces starting with the initial state and ending with the timed states. It then uses Breadth-First Search  (BFS) to process the timed states stored into the queue. It dequeues a timed state, and similar to the initial state, uses  DFS to generate all the traces starting with that state and ending with the new timed states.  
We use this algorithm to also generate the state space of a set of components in the non-compositional approach.
%
%The algorithm to generate the state space of a given component terminates whenever one of the following conditions is fulfilled: all the moving objects supposed to travel through the component depart from it (reach their destinations), a disaster happens (i.e. the fuel of a moving object is zero), and the analysis time passes a threshold. %These conditions also determine when the algorithm to generate the state space in the non-composition approch terminates. The first condition for the non-compositional approach is modified as follows. The algorithm terminates if all moving objects depart from their current sub-tracks at the pre-specified times. 
It is notable that the state space of the system does not have a Zeno behavior, since the travel of every moving object across a sub-track takes time.  The minimum progress of the time in our model is assumed to be one unit.
%We call a state a timed state if a time transition is enabled at the state. We provide a queue initiated with the initial state to store the timed states. Our algorithm to generate the state space pops a timed state from the queue and uses the iterative Depth First Search (DFS) algorithm to generate all the execution traces in the current iteration. The timed states resulted from the execution traces are pushed into the queue and are processed one by one in the next iterations. 

\begin{algorithm}
	\DontPrintSemicolon 
	\SetKwFunction{algo}{All}
	\KwIn{$s_0$ as the initial state of the component}
	\KwOut{\emph{stateSpace} that is the state space}
	\Begin{
		$\mathit{stateSpace} \gets \{s_0\}$\;
		$\mathit{queue} \gets \{s_0\}$\;
		\While{$\mathit{queue} \neq \emptyset$}{
			$s \gets \mathit{deQueue(queue)}$\;
			$s' \gets \mathit{timeProg(s)}$\;
			$\mathit{stateSpace} \gets \mathit{stateSpace} \cup \{s'\}$\;
			$\mathit{states} \gets \mathit{depthFS(s')}$\;
			\If{$\mathit{states}=\emptyset$}{
				\nl \KwRet $\mathit{stateSpace}$\;
			}
			%$\mathit{queue} \gets \mathit{queue}\setminus \{s\}$\;
			$\mathit{queue} \gets \langle \mathit{queue} | \mathit{depthFS(s')} \rangle$\;
		}
		\nl \KwRet $\mathit{stateSpace}$\;
	}
	\caption{Algorithm to generate the state space}
	\label{algo::all}
\end{algorithm}

\begin{algorithm}
	\DontPrintSemicolon 
	\SetKwFunction{algo}{All}
	\KwIn{$s$ as a state}
	\KwOut{\emph{timedStates} as a set of timed states}
	\Begin{
		$\mathit{timedStates} \gets \emptyset$\;
		\If{$\mathit{terminate(s)}$}{
			\nl \KwRet $\emptyset$\;
		}
		\If{$\mathit{timeStat(s)}$}{
			\nl \KwRet $\mathit{\{s\}}$\;
		}
		\ForEach{$ e \in \mathit{buffer(s)}$}
		{
			%$e \gets \mathit{take(buffer(s_0))}$\;
			$s_1 \gets \mathit{trigger(s,e)}$\;
			$\mathit{stateSpace} \gets \mathit{stateSpace} \cup \{s_1\}$\;
			\If{$depthFS(s_1)=\emptyset$}{\nl \KwRet $\emptyset$\;}
			$\mathit{timedStates} \gets \mathit{timedStates} \cup \mathit{depthFS(s_1)}$
		}
		\nl \KwRet $\mathit{timeStates}$\;
	}
	\caption{depthFS}
	\label{algo::dfs}
\end{algorithm}
%The coordinated actor model of a TTCS results in a TIOTS with a finite number of states, since execution of the system terminates at some point, when all moving objects have arrived at their destinations. An event in a coordinated actor model is tagged with a time, whose difference with the time at which the event is created is zero or greater than zero. Furthermore, the minimum progress of the time is one unit. Therefore, there is only one case in which TIOTS of the coordinated actor model exhibits a Zeno behavior. This case appears when there is a cycle in TIOTS and this cycle does not contain a time-progress transition. A Zeno behavior is a behavior in which an infinite number of internal transitions are performed during a limited amount of time. %As passage of each moving object through a sub-track in a TTCS takes an amount of time, time is always advanced, and there is no trace of execution in the model of a TTCS along which an infinite number of internal transitions are performed during a limited amount of time (TIOTS is Zeno-free).  

\textbf{Composition.} Assume that a storm happens at time $t$, and the component $C_1$ is affected by the storm. %We obtain the state of the real system by estimating positions of the aircraft in the traffic network at time $t$ using their flight plans. We set the initial state of $C_1$ to the state of the real system. We also set states of boundary actors of the environment components of $C_1$ to the information of the aircraft traveling through $C_1$ in the future.
Also, assume that flight plans of the aircraft are not necessarily the initial flight plans and have been adjusted based on the data monitored from the system. We use flight plans of the aircraft to extract the state of the system, describing positions of the aircraft in the traffic network at time $t$. We then set the initial state of $C_1$ to the state of the system. Similarly, we use flight plans of the aircraft to obtain states of environment actors of  $C_1$ by identifying the aircraft entering into $C_1$ in the future. These environment actors will directly send messages to $C_1$ at times $t'$, $t' \geq t$, and expect to receive messages from $C_1$ at times $t'$, $t' \geq t$. To generate the state space, we first compose $C_1$ with its environment actors to create a new component. %Note that this composition is performed at the level of the coordinated actor model. 
We then use the above algorithm to generate the state space of the new component. %If we could generate the state space, the composition of the state spaces of $C_1$ and the boundary actors is successfully performed. Otherwise, as mentioned in Section~\ref{Subsect::parallelComp}, the composition reaches an error state. 
If the composition of the state spaces of $C_1$ and the environment actors does not reach a deadlock state, the state space of the new component is successfully generated. Suppose the case in which the composition reaches a deadlock state. In this case, assume that an environment actor belonging to the environment component $C_2$ is not able to send its message at the pre-specified time $t'', t'' \geq t $, to $C_1$. This means that the change is propagated from $C_1$ to $C_2$ at time $t''$. We repeat the same procedure as $C_1$ to set the state of $C_2$ to the state of the real system at time $t''$. We also set states of those environment actors that send messages to $C_2$ at the times greater than $t''$.  Therefore, from $t''$ on, we have a component, composed of $C_1$, $C_2$, and several environment actors, whose state space is generated. %The same procedure repeats if a boundary actor of $r$ is not able to send its message to a neighboring region. 
This procedure terminates whenever the algorithm reaches a state in which all moving objects supposed to travel across the new component depart from it at their pre-specified times. In other words, the model has a trace during which all messages are received from the new component at the pre-specified times.

\subsection{Experimental Setting}
To compare the compositional and non-compositional approaches, we focused on an ATC example with a $\mathit{n \times n}$ mesh map, where the location of each sub-track is shown by the pair $\mathit{(x,y)}$ in the mesh. We also considered $\mathit{2\times (n-1)}$ source airports (each one is connected to a sub-track whose location is the pair $\mathit{(0,i)}$ or $\mathit{(i,0)}$, $\mathit{0 \leq i < n}$ ), and $\mathit{2 \times (n-1)}$ destination airports (each one is connected to a sub-track whose location is the pair  $\mathit{(n-1,i)}$ or $\mathit{(i,n-1)}$). We developed an algorithm to generate the initial flight plans of $m$ aircraft, and an algorithm (an adaptation policy) to reroute the aircraft as follows. The pseudo-codes of algorithms are given in the appendix.

\textbf{ALG1: Generating the initial plans.} %This algorithm attempts to find a route for an aircraft based on the XY-routing algorithm. 
%This algorithm searches all the possible routes between a source and a destination, and returns the first available route for the aircraft. A route is available if none of its sub-tracks is assigned to several aircraft at the same time.  The algorithm works as follows. 
This algorithm randomly generates the source $\mathit{(x_s,y_s)}$, the destination $\mathit{(x_d,y_d)}$, and a departure time from the source airport for each aircraft. The departure times follow an exponential distribution with the parameter $\lambda$. The time difference between two subsequent departures from a source airport should not be less than the flight time $\mathit{FD}$, which shows the traveling time of an aircraft across a sub-track. The aircraft $A$ can travel through the sub-track with the location $(x,y)$ if $A$ has no time conflict with the aircraft $B$, which is also supposed to travel across $(x,y)$. Similar to the XY routing algorithm \cite{chawade2012review}, ALG1 attempts to find a route from $\mathit{(x_s,y_s)}$ to $\mathit{(x_d,y_d)}$ by first traversing the X dimension and then traversing the Y dimension of the mesh. ALG1 switches its traversing direction from X to Y whenever the aircraft has a time conflict with another aircraft along the X dimension. 
If from the location $\mathit{(x,y)}$ the aircraft has a time conflict along  both dimensions, ALG1 
%\COM{If ALG1  cannot \rem{move  across} \marjan{find any route from - any available track? are you sure you can say the algorithm cannot move?} any  of the dimensions from the location $\mathit{(x,y)}$, 
backtracks  to a location before $\mathit{(x,y)}$ on the route, where it can switch its traversing direction to Y to travel through  a new location. 
%ALG1 backtracks \marjan{to where?}if it cannot move across any  of the dimensions from the location $\mathit{(x,y)}$. 
%It then moves across the Y dimension. 
This procedure continues until a route is discovered. ALG1 does not guarantee to find the most efficient (e.g. shortest) route.

\textbf{ALG2: Rerouting algorithm.} This algorithm attempts to find a route with the same length as the initial route of the aircraft. Assume that the aircraft is going to leave the location $\mathit{(x_0,y_0)}$ and the rest of its route is $\mathit{[(x_1,y_1),(x_2,y_2),\cdots,(x_n,y_n)]}$. Also, assume that the sub-track $T$ with the location $\mathit{(x_1,y_1)}$ is unavailable, and the moving object is not able to travel through it. In the ATC example, a sub-track is unavailable if it is stormy or is occupied by another aircraft. The algorithm finds a neighbor of $(x_0,y_0)$, e.g. the sub-track $T'$, that is available. %It then attempts to find a route with the length  equal to the length of the initial route from $T'$ in several steps.  %and neither of its $x$ and $y$ is equal to $T$.
%Then, the algorithm tries to find a route from $T'$ in several steps.At the first step, the algorithm
It then tries to find a route with the length 2 from $T'$ to  $\mathit{(x_2,y_2)}$. If there is no such route, it attempts to find a route with the length 3 from $T'$ to $\mathit{(x_3,y_3)}$, and so on. The algorithm uses the same procedure as ALG1 to find a route. It first traverses the X dimension and then traverses the Y dimension of the mesh. If a route from $T'$ to $\mathit{(x_i,y_i)}$, $2 \leq i \leq n$, is found, the route is concatenated with the rest of the route of the aircraft from $\mathit{(x_{i+1},y_{i+1})}$ to $\mathit{(x_n,y_n)}$, and the resulting route is returned.  %This approach adheres to the length of the initial route. 
It is possible that there is no available sub-track $T'$, or using the above approach, a route with the same length as the initial route is not found. 
 %If $(x_0,y_0)$ does not have an available neighbor, or using the above approach, a route from an available neighbor of $(x_0,y_0)$ is not found,   %with the same length as the initial route is not found, 
Then, the algorithm selects a neighbor of $(x_0,y_0)$ even if it is occupied, and returns a route from it to $(x_n,y_n)$. The aircraft will stay one more unit of time in $(x_0,y_0)$, and will fly based on its new route. If using the both above approaches no route is found, the moving object will stay one more unit of time in $(x_0,y_0)$, and then will fly based on its initial route if $T$ becomes available. Otherwise, the procedure of ALG2 repeats. %The rerouting algorithm uses the same procedure as ALG1 to find a route. It first traverses the X dimension and then traverses the Y dimension of the mesh. 
 Although ALG2 uses the same procedure as ALG1 to find a route, %In contrast to ALG1, ALG2 
 it avoids the stormy track, but
 it does not check the time conflict  with other
aircraft in the future.  If a potential conflict is detected, we will take care of it by rerouting the aircraft upon detecting the conflict. %and therefore backtracking is not needed.  However, ALG2 does not select a stormy sub-track as a part of its route.

\textbf{Scenarios.} Different parameters such as the rerouting algorithm, 
the time of the storm, the place of the storm, the network traffic volume, the amount of concurrency arisen from flight plans of the aircraft, and the network dimension change the results of  experiments. Because the traffic network of the ATC domain has a cascaded architecture, %the storm typically occurs in the center of the traffic network. 
the place of the storm can be typically approximated as the middle of the traffic network. Therefore, we select the middlemost sub-track of the network as the place of the storm.  We perform three sets of experiments; (ES1)
that is to compare the time and memory consumptions
between the compositional and non-compositional approaches, (ES2) that is to depict the variation of the time consumption in a set of experiments for each approach, and
(ES3) that is to compare the scalability of the approaches. The scenarios are described in the following.  
In our experiments, we assume that FD as the traveling time of an aircraft across a sub-track is one. We also assume that the aircraft consumes one unit of fuel per one unit of the traveling time. Furthermore, the fuel of each aircraft is more than the length of the longest path in the traveling network.

\noindent\textbf{(ES1)}. %This set of experiments is performed  to compare the time and memory consumptions between the compositional and non-compositional approaches.
We consider a $15 \times 15$ mesh structure, divided into $9$ regions of $5 \times 5$,  as the traffic networks in (ES1). %These meshes are respectively divided into $9$ regions of $5 \times 5$. 
The fuel of each aircraft  is set to $325$. We use ALG1 to generate 150 batches of flight plans per each $\lambda$ in $\{0.5, 0.25, 0.125\}$, where $\lambda$ is the parameter of the exponential distribution to generate departure times of the aircraft from source airports.  By increasing the value of $\lambda$, the mean interval time between two departures decreases. As a result, the network traffic volume  and subsequently the concurrency contained in the model might increase. We expect that  the compositional approach performs better than the non-compositional approach even in a low concurrency model.
%
%We change the network traffic volume and subsequently the concurrency contained in the model through different $\lambda$s (the inverse of the mean interval time between two departures).
%
Each generated batch contains flight plans of 2000 aircraft. Per each batch $P_i, 1\leq i \leq 150$, 
we generate 4 batches $P_{ij}, 1\leq j \leq 4$, such that $P_{i1}$ contains the first 500 flight plans of $P_i$, $P_{i2}$ contains the first 1000 flight plans of $P_i$, and so on.  We use  both approaches to analyze each batch $P_{ij}$ per each time of the storm in $\{100,200,400,600,800\}$. Obviously, whenever the storm occurs late, the most of the moving objects have arrived at their destinations. We remove the batch $P_i$ from the experiments of both approaches if for a batch $P_{ij}$ and a time of the storm, the model in one of the approaches is not deadlock-free, or its verification time passes the threshold (the results of experiments in which the models are not deadlock-free are investigated in (ES2)).  Table.~\ref{tab1} shows the number of experiments in which both approaches do not face a deadlock and both approaches analyze the model in less than a threshold (No Deadlock or TimeOver), the number of experiments in which both approaches face a deadlock (With Deadlock), and the number of experiments in which the verification time in both approaches passes a threshold (TimeOver). The threshold of the analysis time is set to an hour. In our experiments, per each $j$, we calculate the averages of the analysis time and the number of states of the batches $P_{ij}$.  
\begin{table}
	\caption{ The number of experiments in which the model in both approaches faces a deadlock (With Deadlock), in which the model does not face a deadlock and its analysis time is less than a threshold (No Deadlock), in which the verification time passes the threshold (TimeOver). The traffic network has a $n \times n$ mesh structure. $\lambda$ is the parameter of the exponential distribution to generate the departure times of the aircraft.}{
		\begin{tabular}{|m{1pc}|m{1.5pc}|m{4pc}|m{5pc}|m{3pc}|}
			%\cline{3-9}
			\hline
			n& $\lambda$ & With Deadlock & No Deadlock or TimeOver & TimeOver \\ \hline
			%	9 & 0.5 & 31 & 118& 1\\ \hline
			15 & 0.5& 27& 120 & 3\\ \hline 
			15 & 0.25& 37&110& 3\\ \hline 
			15 & 0.125& 56& 94& 0\\ \hline
		\end{tabular}
	}
	\label{tab1}
\end{table}

\noindent\textbf{(ES2)}. The traffic network in (ES2) has the same configuration as the traffic network in (ES1). %Similar to (ES1), In our experiments, the traffic networks in (ES1) and (ES2) have the same configuration.  We assume that the fuel of each aircraft is more than the length of the longest path in the traveling network, and is set to $325$ (ES2). The threshold of the analysis time in (ES2) is an hour.  
In (ES2), we use the batches of flight plans generated in (ES1) for $\lambda=0.5$. The reason for considering  $\lambda=0.5$ is that   the network might have the highest traffic volume for $\lambda=0.5$ compared to  $\lambda \in \{0.25, 0.125\}$. We use both approaches to analyze each batch $P_i$, containing the flight plans of 2000 aircraft. Since the possibility of propagating the change increases when the storm happens early, we suppose that the storm happens at time 100.   As shown in Table.~\ref{tab1}, the model in 120 experiments is deadlock-free and is analyzed in less than the predefined threshold. Compared to (ES1) that calculates the average of the analysis time for this set of experiments, (ES2) illustrates the variation of the analysis time in this set for each approach.  Furthermore, (ES2) depicts the variation of the time consumption to detect a deadlock in 27 experiments that are not deadlock-free.

\noindent\textbf{(ES3)}.  As the aim in (ES3) is to compare the scalability of the approaches, we consider a larger traffic network that is a $18 \times 18$ mesh structure with  $9$ regions of $6 \times 6$ in our experiments. The fuel of each aircraft  is set to $425$. We assume that the change happens at time 100. We use ALG1 to generate a batch $P$ of 5500 flight planes with $\lambda=0.5$.  
In (ES3), we start with the first 100 flight plans of $P$, and gradually increase the number of flight plans to compare the scalability of two approaches. The scalability of the approaches is measured by the number of the aircraft. To this end, We define a threshold for the verification time and set this threshold to 45 minutes. The approach that can analyze a model with more number of the aircraft in less than the defined threshold is more scalable.

\begin{figure*}
	\centering
	\subfloat[]{{\includegraphics[width=0.32\linewidth,height=4cm,keepaspectratio]{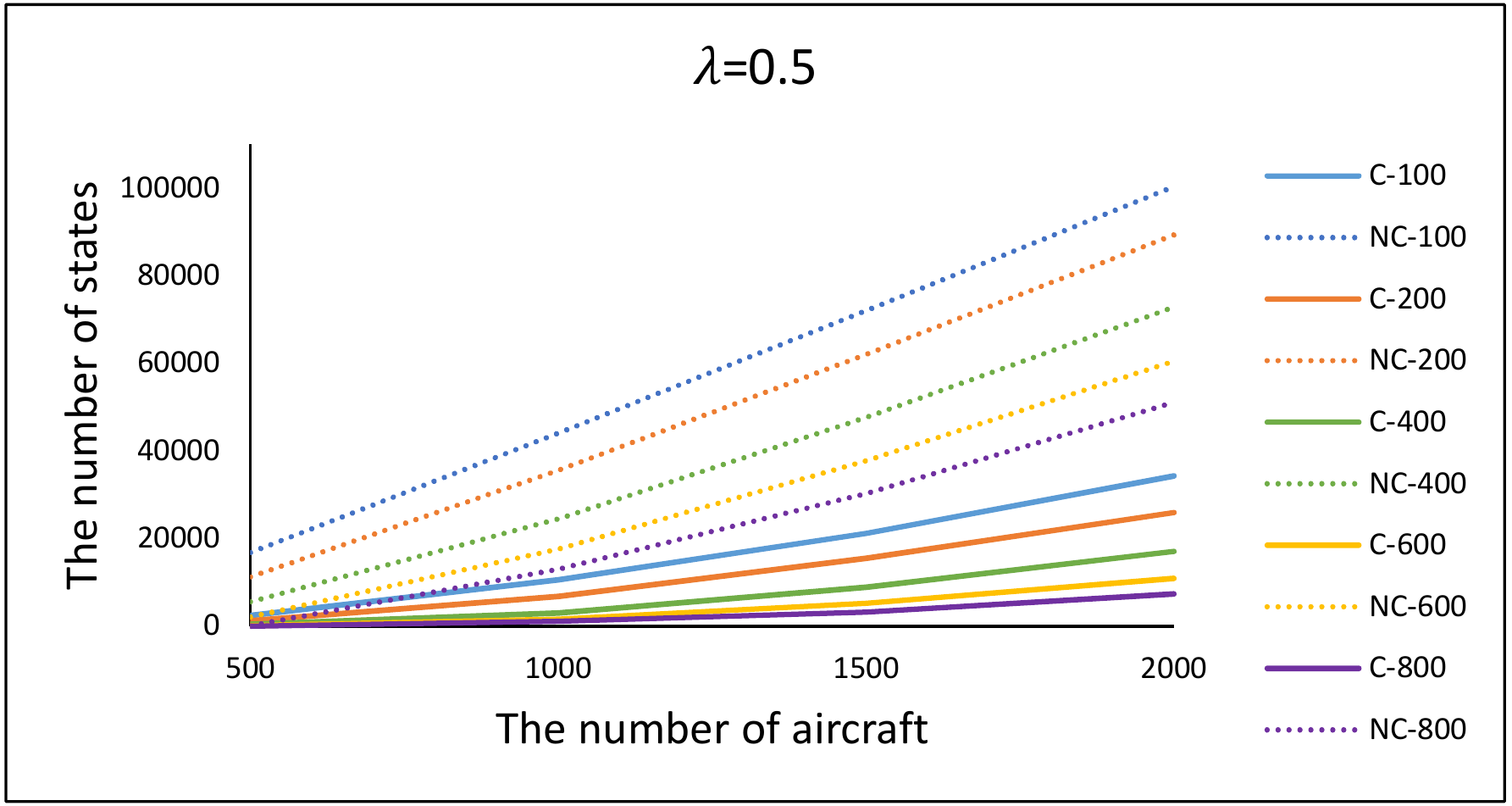} }}%
	%\qquad
	\subfloat[]{{\includegraphics[width=0.32\linewidth,height=4cm,keepaspectratio]{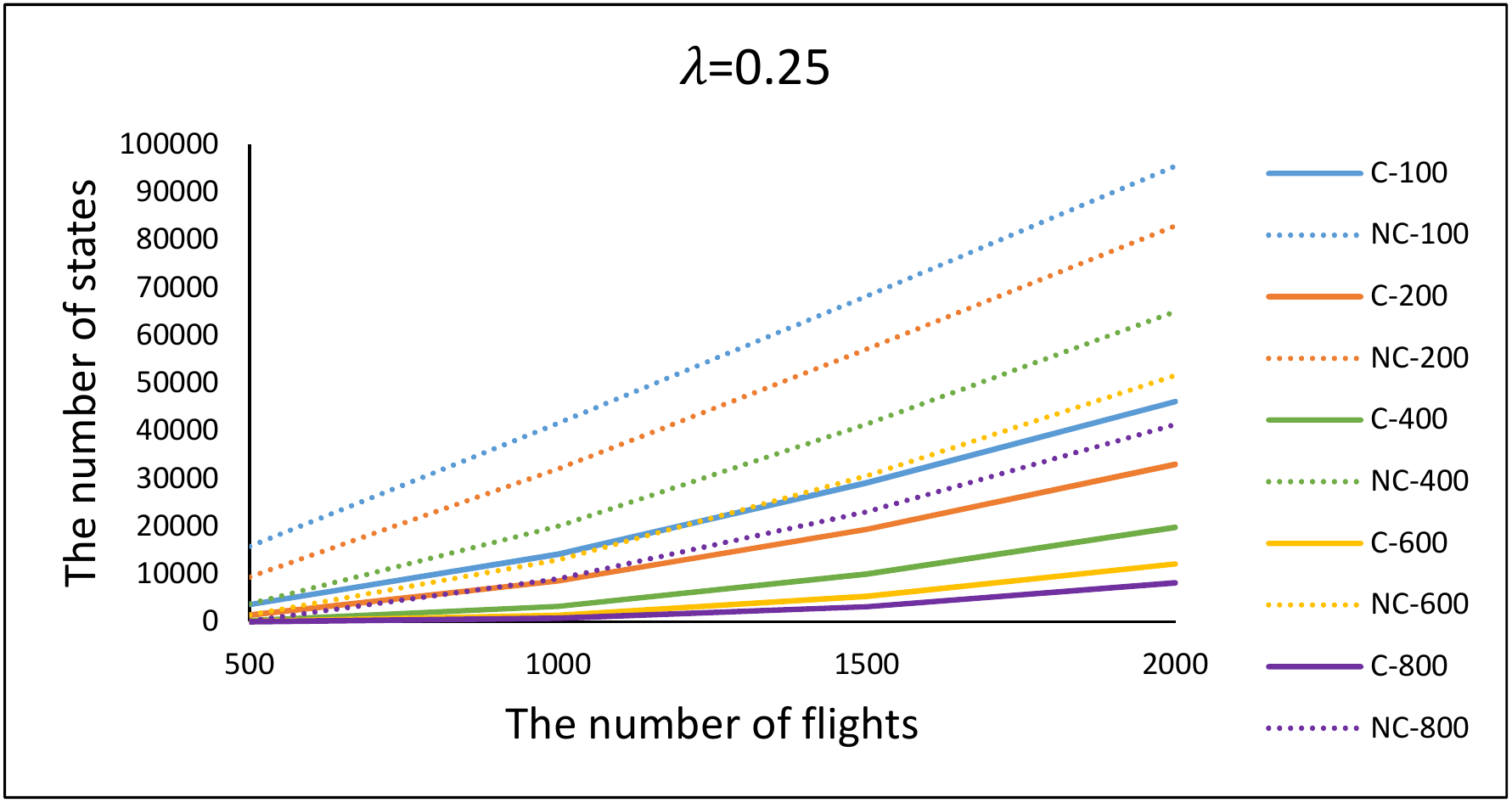} }}%
	%\qquad
	\subfloat[]{{\includegraphics[width=0.32\linewidth,height=4cm,keepaspectratio]{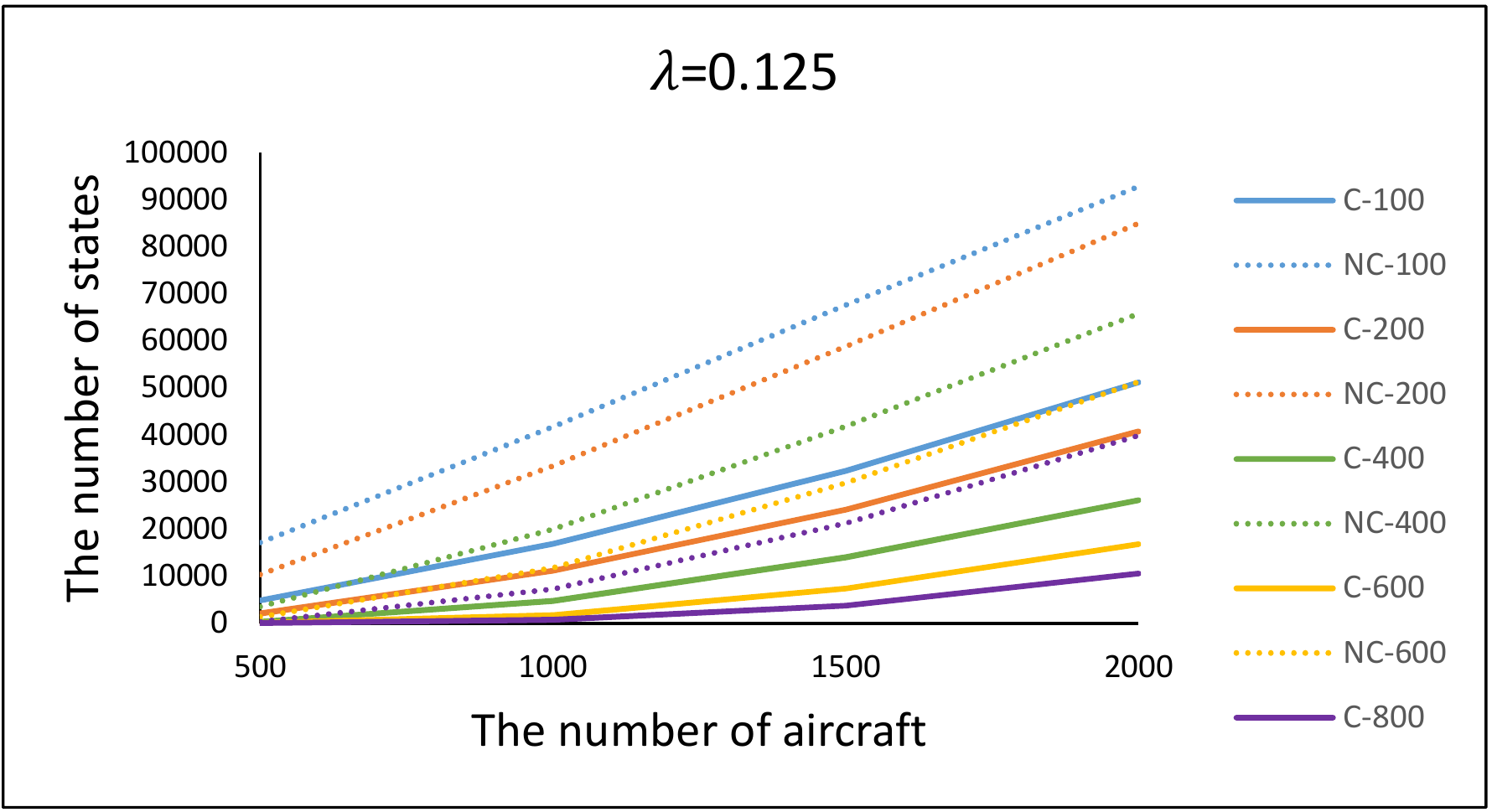} }}%
	\caption{The number of states in (ES1) for each value of $\lambda$ in $\{0.5,0.25,0.125\}$, where $\lambda$ is the parameter of the exponential distribution to generate the departure times of the aircraft. The notations $C$ and $\mathit{NC}$ refer to the compositional and non-compositional approaches, respectively. The time at which the storm happens varies in the set $\{100,200,400,600,800\}$. As an instance, $C-100$ depicts the results of the compsitional approach when a storm occurs at time 100.}%
	\label{15-0.5}%
\end{figure*}

\begin{figure*}
	\centering
	\subfloat[]{{\includegraphics[width=0.47\linewidth,height=4cm]{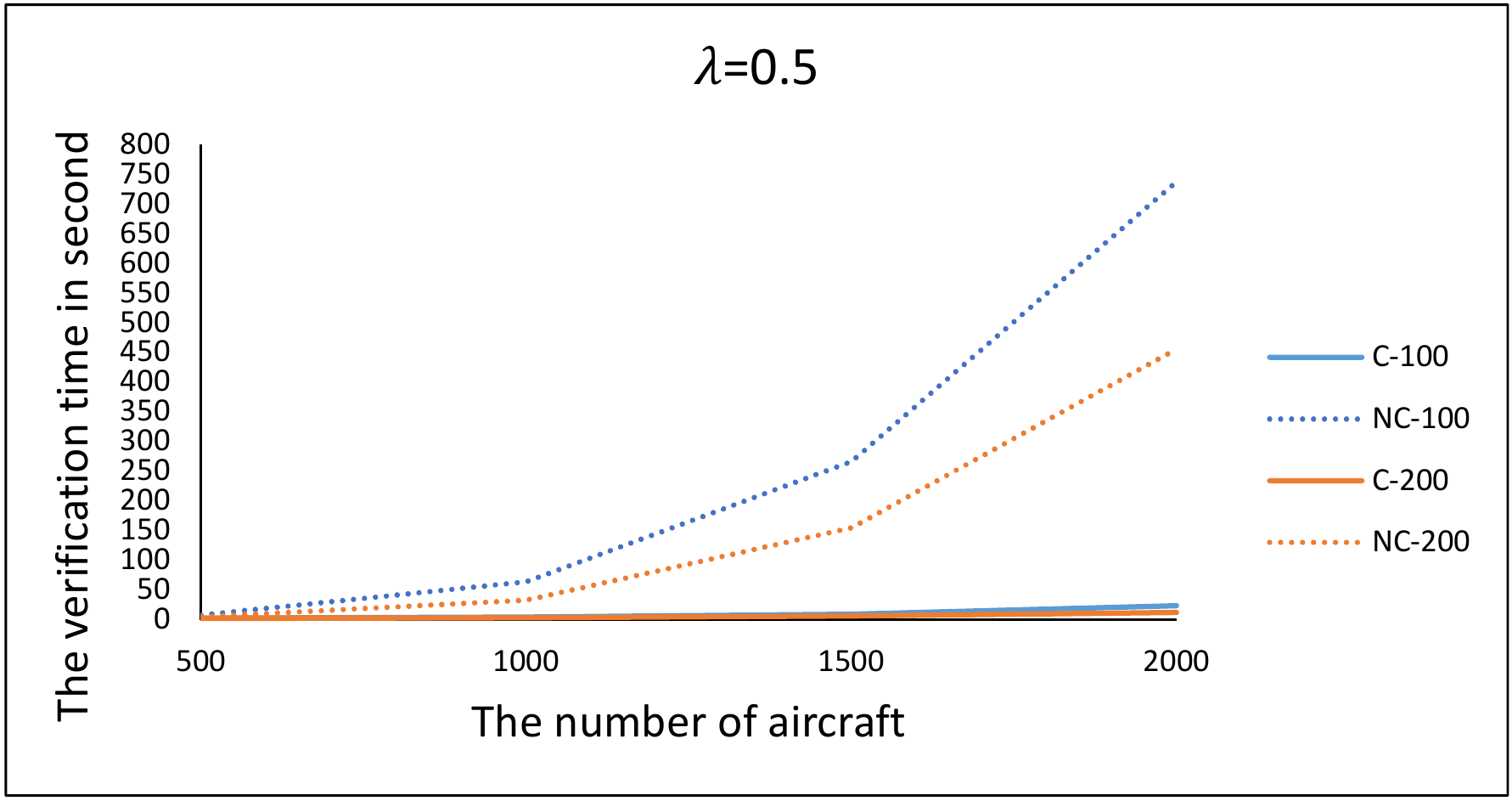} }}%
	%\qquad
	\subfloat[]{{\includegraphics[width=0.47\linewidth,height=4cm]{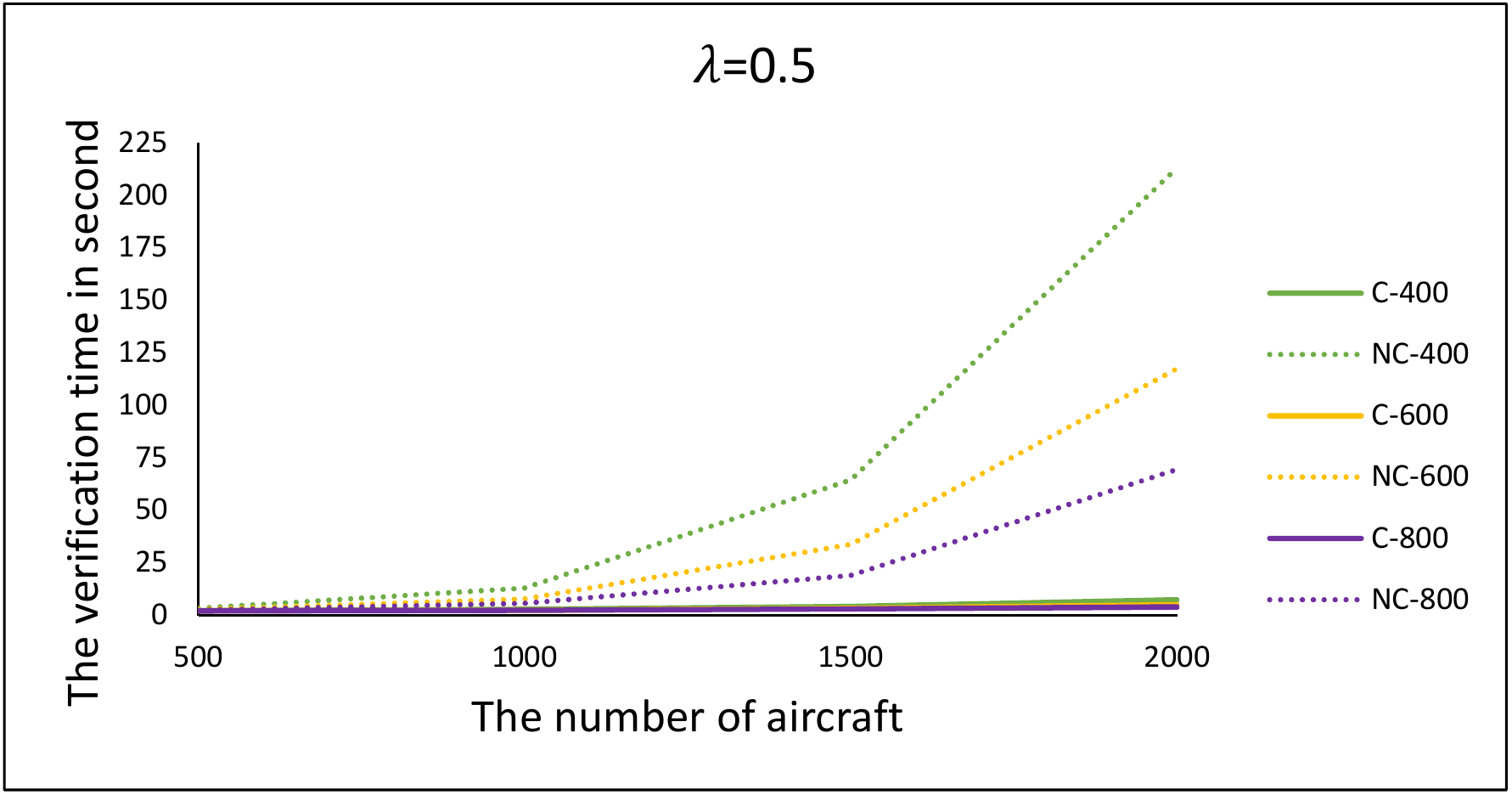} }}%
	%\qquad
	\caption{The verification time in (ES1) for $\lambda=0.5$. The left side depicts the verification time in the compositional (C) and non-compositional (NC) approaches when a storm occurs at a time in $\{100,200\}$. The right side depicts the verification time of each approach when a storm occurs at a time in $\{400,600,800\}$. The right and the left side figures show the verification time with different scales.}%
	\label{15-0.25}%
\end{figure*}
\begin{figure*}
	\centering
	\subfloat[]{{\includegraphics[width=0.32\linewidth,height=4cm,keepaspectratio]{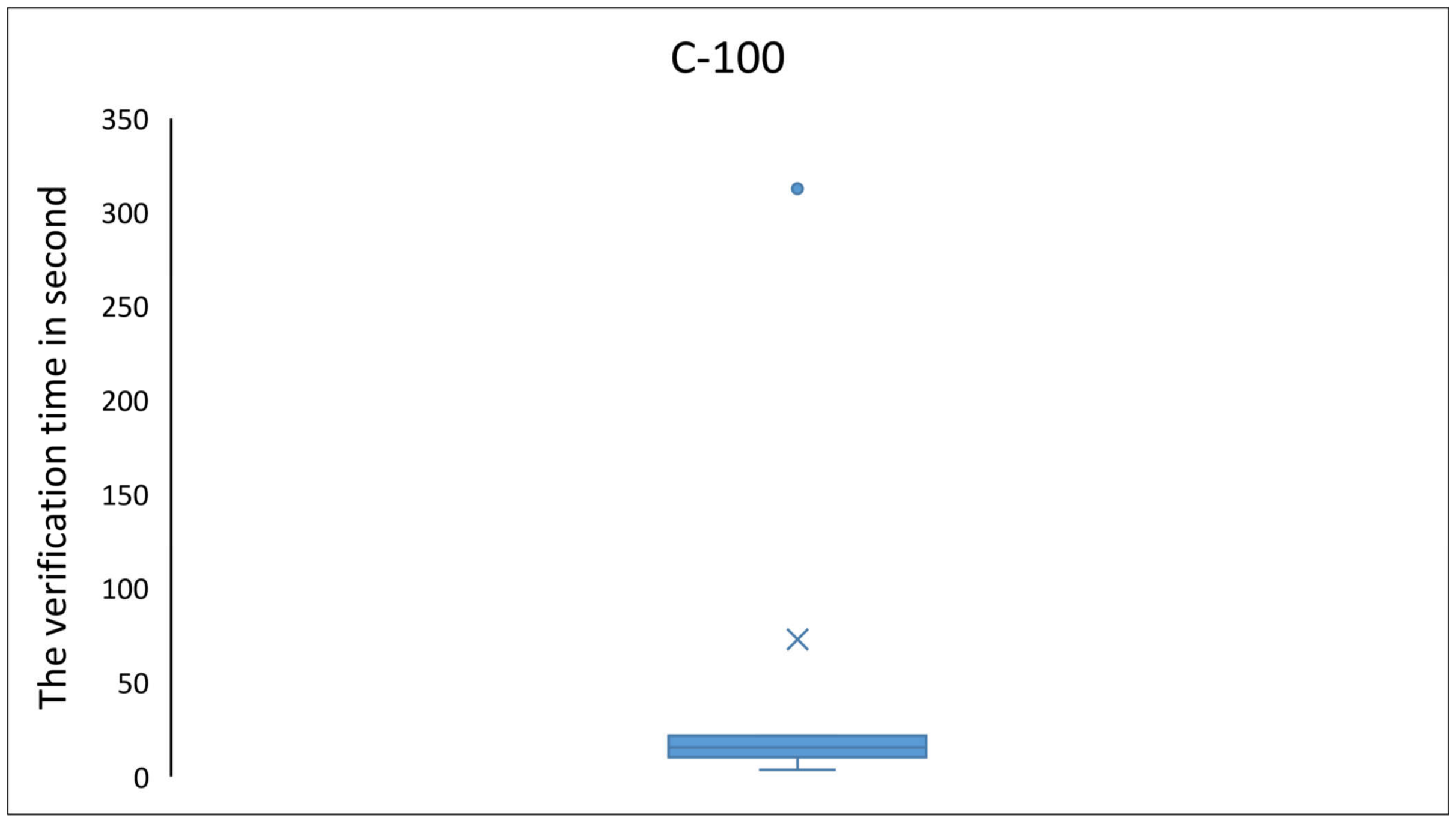} }}%
	%\qquad
	\subfloat[]{{\includegraphics[width=0.32\linewidth,height=4cm,keepaspectratio]{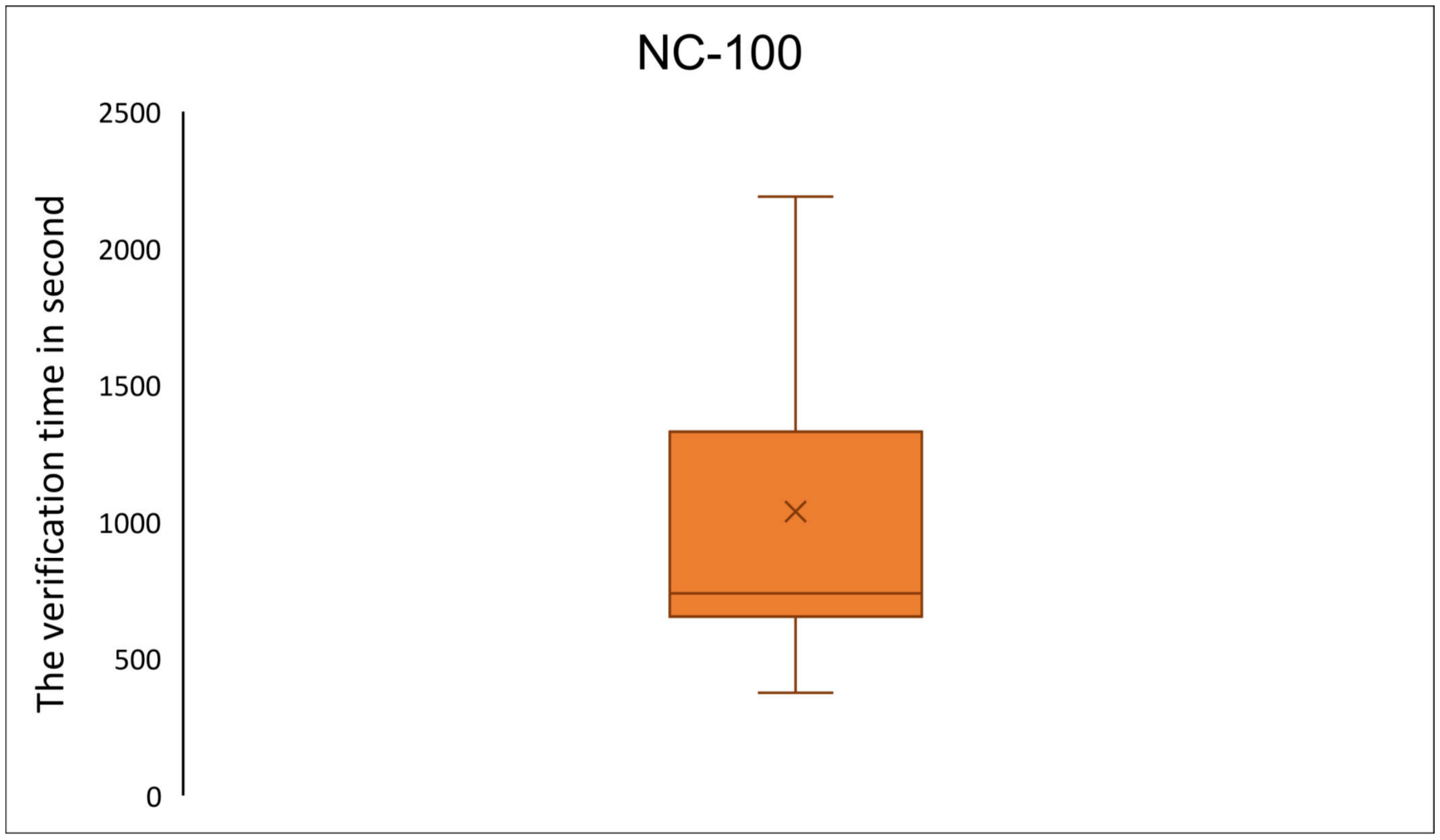} }}%
	%\qquad
	\subfloat[]{{\includegraphics[width=0.32\linewidth,height=4cm,keepaspectratio]{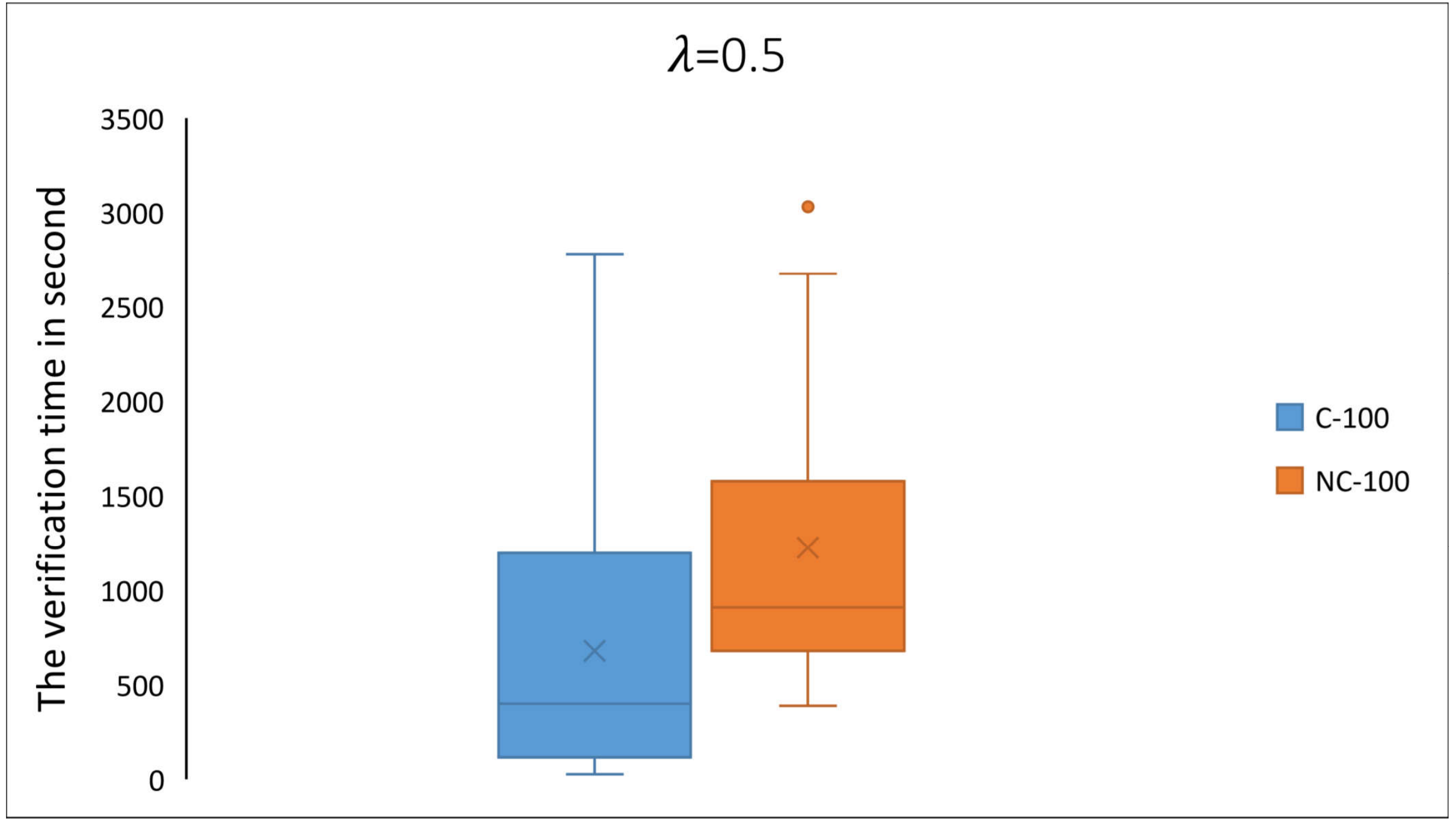} }}%
	
	\caption{The verification time in (ES2) for $\lambda=0.5$. The storm occurs at time 100. The variations of the time needed to verify the experiments with no deadlock using the compositional (C) and non-compositional approaches (NC) are depicted in parts (a) and (b), respectively. The variations of the time needed to detect a deadlock using both approaches are depicted in part (c).}%
	\label{15-0.125}%
\end{figure*}

\begin{figure*}
	\centering
	\includegraphics[width=\linewidth,height=4cm]{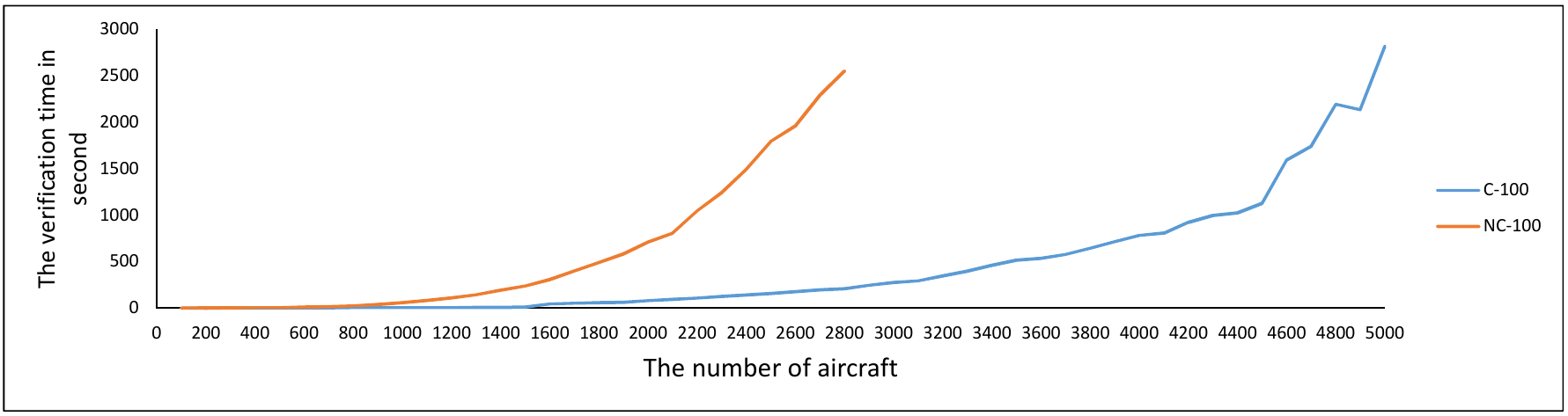}
	\caption{The  scalability of compositional (C) and non-compositional (NC) approaches in (ES3). Both approaches are run for the same  scenario with $\lambda=0.5$. The storm occurs at time 100. The scalability is measured in the number of aircraft, while the verification time is set to a threshold. The non-compositional approach is not able to verify a model with more than 2800 aircraft in a time less than the defined threshold.}
	\label{scalability}
\end{figure*}

\subsection{Comparison of the Magnifier approach and the Non-compositional Approach }

We run our experiments on an ubuntu 18.04 LTS amd64 machine with 67G memory and Intel (R) Xeon (R) CPU E5-2690 v2 @ 3.00GHZ. A part of our experimental results are shown in Figures  ~\ref{15-0.5}, ~\ref{15-0.25}, ~\ref{15-0.125} and ~\ref{scalability}. In these figures, "C" and "NC" refer to the compositional and non-compositional approaches, respectively. The legend entry $\mathit{C-i}, i \in \{100,200,400,600,800
\}$ depicts the experimental results of the compositional approach for the time $i$ at which the storm happens. The legend entry $\mathit{NC-i}$ depicts the results for the non-compositional approach. As shown in Fig.~\ref{15-0.5} and Fig.~\ref{15-0.25}, using the compositional approach results in decreasing the verification time and the number of states.  As expected, by increasing the number of aircraft, the number of states and accordingly, the verification time increase. %For instance, as a few number of the aircraft have arrived at their destinations at time 100, the number of states at this time is large.
The same results are valid for the smaller value of the time at which the storm occurs, since fewer aircraft have arrived at their destinations when the storm happens. By increasing the time at which the storm occurs, the differences between the results of the compositional and non-compositional approaches decrease. It is because most of the aircraft have arrived at their destinations when the storm happens late. To have a better representation of the verification time difference between the compositional and non-compositional approaches, we depict the results of the verification time for $\lambda=0.5$ in two diagrams with two different time scales, shown in Fig.~\ref{15-0.25}. As can be seen, the compositional approach is able to verify a model with 2000 aircraft in a few seconds for the smallest value of the time at which the storm happens. By increasing the  time interval between two departures from a source airport, the number of aircraft entering into the traffic network after the storm happens increases. Therefore, as shown in Fig.~\ref{15-0.5}, the number of states in the compositional approach increases whenever the value of $\lambda$ decreases. 
%
%By the way, ATC is a safety-critical domain, and the analysis to predicate any violations should be performed in a few seconds. 
%Since the traffic is more scattered through the time in $\lambda=0.125$, our approach has the largest number of the states for $\lambda=0.125$. In contrast to our approach, with two reasons the non-compositional approach has the largest number of the states in $\lambda=0.5$ in the most cases. The first reason is that it considers the whole model over which a large number of concurrences may happen. The second reason is that although a large number of the aircraft may arrive at their destination before the storm in $\lambda=0.5$, due to the heavy traffic the concurrency over the model increases. 
%More experiments on comparing the model checking time of the compositional and non-compositional approaches are available\footnote{ \url{www.ce.sharif.ir/~mbagheri/TTCSs.zip}}.

The results of our experiments in (ES2) are shown in Fig.~\ref{15-0.125}. The variation of the verification time in a set of experiments with no deadlock when the compositional approach is used is shown in Fig.~\ref{15-0.125}(a). The results of the same set of experiments for the case in which the non-compositional approach is used are depicted in Fig.~\ref{15-0.125}(b). We also depict the variation of the time  needed to detect a deadlock in a set of experiments using the compositional and non-compositional approaches in Fig.~\ref{15-0.125}(c). As shown in Fig.~\ref{15-0.125}(a), excluding the outliers, the model in our experiments is analyzed in less than 22 seconds using the compositional approach, while this time is around 2190 seconds in the non-compositional approach. Also, in our experiments, the average time for detecting a deadlock in the compositional approach is around 11 minutes, while this value in non-compositional approach is around 20 minutes.

We can define the latency of an adaptation policy by defining a threshold over the analysis time. The latency in our approach is defined as the time needed to adapt the model and to check it for correctness properties. We suppose that a human is involved in adapting the system if the  threshold is passed. As an instance, consider that the latency of the adaptation policy in our approach is 3 minutes. From 150 experiments, only 21 experiments (the verification time of 3 experiments pass the threshold, the outlier experiment in Fig.~\ref{15-0.125}(a), and 17 experiments of 27 experiments that face a deadlock) %whose analysis times are more than an hour 
need  human intervention in the compositional approach. %As shown in Fig.~\ref{dist}(a), even the outlier with the maximum value is analyzed in less than 3 minutes,
This is while all experiments in the non-compositional approach need more than 3 minutes analysis time. %The similar results can be seen in Fig.~\ref{dist}(c) and Fig.~\ref{dist}(d).

The results of our experiment in (ES3) are shown in Fig.~\ref{scalability}. To compare the scalabililty of both approaches, we run both approaches for the same scenario. Furthermore, we define a threshold for the verification time and set this threshold to 45 minutes. %In this experiment, we run a scenario for both approaches in an hour. %In this scenario, %while our approach is scalable to 4200 aircraft, 
As can be seen, the non-compositional approach is not scaled for more than 2800 aircraft. %In other words, it cannot model check a model with more than 1000 aircraft in less than an hour. 
The results of the compositional approach in Fig.~\ref{scalability} have fluctuations appeared between 4600 to 5000 aircraft. By adding new aircraft to the traffic network, some areas are congested, and consequently the concurrency of the model increases. This event results in some fluctuations and the fast growth of the "C" plot between 4600 to 5000 aircraft. Except for this range, this plot has a normal growth, since by adding the new aircraft, the behaviors of the congested areas has not sensibly changed. %By adding 200 aircraft to the traffic network, the aircraft are rerouted in a way that the congestion in some areas decreases. This event results in decreasing the number of states.  %Based on our investigations, the fluctuations appeared between $1600$ to $2200$ aircraft and also between $3200$ to $3600$ aircraft indicate that the aircraft added to the traffic network influence on the new routes selected for the other aircraft. In other words, in the presence of the new aircraft, the traffic network is adapted in a different way somehow the number of states decreases.

%% file: 06-RelatedWork.tex
\section{Related Work} \label{sec::relatedwork}
In this section, we concentrate on four classes of most related studies. The first class is concerned with the theory of interfaces. The second class is about  modeling and verifying traffic control systems. The third class describes the most closely related work that use compositional methods for the verification purpose, and the fourth class is about formal analysis of self-adaptive systems at runtime. %Finally, we briefly explain the major contributions of this work compared to the related work.

\textbf{Interface Theory}. The theory of interfaces is a widely studied topic. This theory describes the main features that each component-based design should obey, such as refinement, structural composition, and conjunction. The same as ours, the focus of \cite{10.1007/3-540-45828-X_9} is on the structural composition. The work of \cite{10.1007/3-540-45828-X_9} presents a theory of timed interfaces to explain the timing constraints on inputs and outputs of the components. A timed interface is encoded as a two-player timed game in which the environment as the input player provides inputs for the component and the component as the output player creates outputs. This work proposes an optimistic approach of composition that is two components can work together if there is a helpful environment to make them work together. In other words, two components are compatible if the environment has a winning strategy to avoid immediate and time error states in parallel product of two components. An immediate error state is reachable if a component sends an output that is not acceptable by the other component. A component blocks the progress of the time in a time error state. Similar to the approach of \cite{10.1007/3-540-45828-X_9}, \cite{david2010timed} proposes an optimistic approach for the structural composition of two interfaces specified by Timed Input Output Transition Systems (TIOTSs). In contrast to \cite{10.1007/3-540-45828-X_9}, \cite{david2010timed} %does not define time error states and 
assumes that the system is input-enabled. The input-enabled assumption is also considered when interfaces are specified by TIOAs \cite{1253264}. The approach of composition in \cite{1253264} is pessimistic. It means that two components should work together in all environments. Optimistic treatment of the composition is also considered in \cite{DBLP:phd/dnb/Bujtor18}, where an interface theory for Modal Input Output Automata is proposed. The approach of composition in \cite{DBLP:journals/corr/abs-1101-4731} is also pessimistic. In \cite{DBLP:journals/corr/abs-1101-4731}, two components are compatible if they do not reach an error state in their parallel product. The interfaces in \cite{DBLP:journals/corr/abs-1101-4731} are specified by Modal Input Output transition systems in which the timing constraints are not specified. 

Compared to the related work, we follow a pessimistic approach that is two components can work together if parallel  product of their TIOAs does not reach a deadlock state. Also the same as \cite{10.1007/3-540-45828-X_9}, we are able to express the input assumptions and there is no need to the input-enabled assumption. Since in our approach  the components do not block the progress of the time, we do not need to define time error states. 
%The compatibility of two TIOAs in \cite{1253264} is the composability described in Definition.~\ref{com::prec}.

\textbf{Modeling and Verifying Traffic Control Systems (TCSs)}. TCSs such as ATC and train control systems, due to the tight interconnection of the physical plant and the controller software, are mostly categorized as hybrid systems. There is a vast literature on verifying dynamic models of TCSs %, e.g. hybrid automata,
to detect the future conflicts among the moving objects \cite{Bu:2011:TOH:2000367.2000368,4177563,CHENG2016169}, to resolve the potential conflicts through the trajectory planning \cite{Loos:2013:FVD:2461328.2461350,6339013}, and to evaluate the correctness of the communication protocols among different entities of the system \cite{Damm2007,ZHAO2014337,CHEN2011505}. These approaches use the \textit{Lagrangian} models in which %each moving object, e.g. an aircraft or a train, along with its operational details is the concern of modeling \cite{10.1007/978-3-030-02146-7_1,menon2004new}. For instance, in a Lagrangian model of a train control system, each train can be modeled by a hybrid automata.
 %In these approaches, 
 the moving objects, e.g. aircraft or trains, along with their operational details are the concern of modeling \cite{menon2004new,10.1007/978-3-030-02146-7_1}. As an instance of this kind of modeling, the model of a train control system is created by composing a set of hybrid automata, where each automaton is the model of a train. Modeling the dynamic behaviors of each moving object in these approaches needs a set of differential equations, which due to the large number of the moving objects, makes the analysis of TCSs difficult \cite{menon2004new}. Furthermore, this approach of modeling is only necessary when we need to have a microscopic view of the traffic for our analysis purposes.
 %a troublesome task. 

In contrast to the Lagrangian-based approaches, our approach is based on \textit{Eulerian} models in which %the regions of the traveling space, e.g. the sub-tracks in TTCSs, are the concerns of modeling \cite{10.1007/978-3-030-02146-7_1,menon2004new}. Although the Eulerian modeling  may lose some operational details of the moving objects, it is more appropriate for modeling  
%In contrast to the mentioned approaches,
 the regions of the traveling space, e.g. the sub-tracks in \did{track-based systems}, are the concern of modeling \cite{menon2004new,10.1007/978-3-030-02146-7_1}. Although this kind of modeling  may lose some operational details of the moving objects, it is more appropriate for modeling  rerouting/rescheduling of the moving objects \cite{10.1007/978-3-030-02146-7_1}. In other words, the adaptation of \did{the system} in a macroscopic view is concerned with aggregate behaviors of the moving objects, and affects the whole traffic network by rerouting/rescheduling the moving objects \cite{10.1007/978-3-030-02146-7_1}.  Therefore, by modeling each sub-track as an actor, we develop a one-dimensional model of the traveling space instead of a complex multi-dimensional model of the moving objects. The properties of our interest such as preventing a moving object from running out of the fuel, and the arrival of a moving object at its destination at a pre-specified time are handled by adding a few features to the message corresponding to the moving object. For instance, each message carries information about the remaining fuel of the moving object for the rest of its travel, or carries the designated time for the arrival of the moving object at each sub-track in its route.  This approach of modeling not only provides an acceptable fidelity for the problem \cite{cooda}, but also relieves the analysis difficulties. It is notable that a few of the mentioned approaches such as \cite{Bu:2011:TOH:2000367.2000368} verify the system at runtime. However, the approach of \cite{Bu:2011:TOH:2000367.2000368} is not compositional.

Based on the related work, there is an increased interest towards the scheduling and path planning of moving objects in TCSs. A scheduling problem \cite{Behrmann:2005:OSU:1059816.1059823,Rasmussen2006} is the problem of efficiently assigning resources to a set of tasks such that some constraints are met, e.g. two tasks do not simultaneously use the same resource. Although the scheduling and path planning are not concerns of this paper, we briefly study the modeling approach of several work by putting the hybrid and dynamic modeling of TCSs aside. 
 
 In \cite{Behrmann:2005:OSU:1059816.1059823}, the same as a sub-track, a resource is in the idle state or in the use state. The resource maintains its use state until a clock reaches a usage time. In \cite{Behrmann:2005:OSU:1059816.1059823,Rasmussen2006}, Priced Timed Automata (PTA) is used for the scheduling and planning problem. %PTA is an extension of timed automata in which edges and locations are annotated with costs and cost rates, respectively.
  \cite{Rasmussen2006} uses PTA for the aircraft landing problem where a landing time and a runway should be assigned to each aircraft. %An aircraft should land within a pre-determined time window. Furthermore, 
  A minimum delay between two aircraft landing on the same runway should be preserved. 

%DRONA is a framework to build reliable applications of Distributed Mobile Robotics (DMR),  proposed in \cite{7945013}. In these applications, 
In \cite{7945013}, a grid-shaped workspace with static obstacles is shared between a set of robots, where tasks of the system to be performed by the robots are created dynamically. The aim of the paper is to compute an optimal collision-free motion plan for a robot whenever a new task is created. To program these applications, the authors use the P programming language. %that is for implementing event-driven, asynchronous programs, and has a model checker. 
The robots and also the plan generator are processes of the P language. %The correctness of the program is checked at design time. 
The concern of this paper is not reducing the model checking time, but decreasing the motion generation time.
A multi-robot system is modeled by a network of timed automata  
in \cite{1302413}. %In this system several robots start their journeys from an initial configuration on a grid-shaped workspace and want to reach their destination locations while avoiding conflicts with each other and with the static obstacles.  
The approach  of \cite{1302413} obtains all possible trajectories that move robots on a grid-shaped workspace from their initial locations to their destination locations. To this end, each obstacle, each robot, and each controller associated with a robot is modeled by a separate timed automaton. %The controller asks the robot to move into an adjacent cell. Each robot can remain in a cell within an interval. 
%Finally the strategies 
The trajectories are checked against Computational Tree Logic properties in   \textsc{Uppaal}.
%
%The approach of \cite{7810369} uses a timed automaton to capture the robot motions, where each edge shows the movement of the robot and locations are partitions of the space. The aim of this work is to find an optimal time path on the automaton such that a property described in Metric Interval Temporal Logic (MITL) is satisfied. 
In \cite{7810369}, a timed automaton for a robot and its environment is created, where each edge shows the movement of the robot and locations are partitions of the space. A  property of interest specified in Metric Interval Temporal Logic %Then, the MITL property 
is transformed into another automaton. The product of the robot automaton and the property automaton is %calculated. The resulting automaton is then 
given to \textsc{Uppaal}. Any execution that starts from an initial state and reaches a final accepting state is an accepting trajectory. 

Compared to the related work, we model sub-tracks instead of moving objects. This approach seems to be efficient and scalable when an enormous number of moving objects pass through a fixed workspace. Furthermore, we assume that plans are given or calculated using a planning algorithm, and %we check whether, considering the dynamic obstacles generated at runtime, the plans satisfy system's constraints.
obstacles are dynamically generated  at runtime.  However, obtaining optimal plans  regarding the harsh timing constraints is a difficult task.

\textbf{Compositional Methods}. %A compositional verification method is proposed in \cite{clarke1989compositional}, where each component of the model is supplied with a correctness property and an abstraction of its environment is modeled by interface processes. Then, by composing a component with its interface processes and verifying a property over the composition, the satisfaction of the property over the whole system is proved. %Consider two components $P_1$ and $P_2$ with the alphabet $\Sigma_{P_1}$ and $\Sigma_{P_2}$. The restriction of $P_1$ to $\Sigma_{P_2}$, shown by $P_1 \downarrow \Sigma_{P_2}$ , is the interface process $A_1$. Compositional verification using the  approach of \cite{clarke1989compositional} is formulated as: $A_1 \equiv P_1 \downarrow \Sigma_{P_2} \wedge \psi \in L(\Sigma_{P_2}) \wedge A_1 \parallel P_2 \models \psi \Rightarrow \, P_1 \parallel P_2 \, \models \psi$, where $L$ is a logic for reasoning.  
%Unlike this approach, we do not need to define the correctness properties for the components of the model. 
%
%
%The time constraints and the deadlock-freedom are preserved if TIOAs of the adapted component and its interface components are compatible. The conflict avoidance in our approach is built in the coordinated actor model, since each sub-track is a critical section \cite{cooda}. %Furthermore, if the fuel level becomes less than a threshold, a notification is raised in the model. 
%Therefore, there is no need to express the correctness properties with a logical formula. % we do not need to define any correctness properties.
%
In \cite{5069079}, an  Assume-Guarantee based approach for verifying self-adaptive systems at design time is proposed. In \cite{5069079}, the changed component is adapted. Then, a backward reasoning starts and re-generates a new assumption for the adapted component. If the new assumption is weaker than the previous assumption of the component, the adaptation is correct. Otherwise, the reasoning continues on the context of the changed component. If it reaches a null assumption on the context of the system, the adaptation is incorrect. The paper focuses on safety properties of the system, and does not consider the change propagation. The work in \cite{David2012} defines a refinement relation and a weakening operator to check the satisfaction of a property over a real-time system. Each property is divided into a set of subspecifications for which an assumption and a guarantee are  defined. The subspecifications, assumptions, and guarantees are defined by TIOAs. The assumption and guarantee are combined into a contract using the weakening operator. The property is satisfied if subspecifications refine their corresponding contracts and vise versa. This approach is not proposed for verifying self-adaptive systems at runtime and consequently does not consider the change and its effects on the system. %Unlike \cite{5069079} and \cite{David2012}, which define an assumption and a guarantee for each component,  we only define the interface components.

Magnifying-Lens Abstraction (MLA), presented in \cite{deAlfaro2007}, copes with the state space explosion in obtaining the maximal probabilities over a Markov Decision Process (MDP). 
It partitions the state space into regions, and calculates the upper and lower bounds for the maximal reachability or safety properties on the regions. It magnifies on a region at a time and obtains the values of mentioned parameters by calculating their values for each concrete state.  
Unlike MLA, the bounds of sub-properties in our approach are given through the interface components. % and are re-generated by adapting the components. 
Furthermore, the change propagation is not a concern in \cite{deAlfaro2007}.
The mechatronic UML (mUML) approach, %for model driven development of self-adaptive systems is 
proposed in \cite{Giese2013}, %. The mUML approach follows the Operator-Controller-Module (OCM) architectural model to develop the system components. Each main component, consisting of different elements, can be constructed by nesting different OCM hierarchies where each one defines an element.
 uses the refined UML model component and refined state charts to formally define components of a system, their interactions, and their timing and hybrid behaviors. %This approach uses compositional model checking to check the safety of the system. 
%The system is safe if each component of the system is safe and deadlock-free, and   constraints on interactions are satisfied. Therefore, 
Besides separately checking the safety of each component, mUML checks whether  interfaces of components are well-defined and  components refine their interfaces. In contrast to our approach, this approach is able to model hybrid behaviors, but change propagation is not  considered in  \cite{Giese2013}. Furthermore, as explicitly mentioned in the paper,  verification at runtime is not a concern in \cite{Giese2013}. %in \cite{Giese2013}. %This is a design time approach (However, it would be too late for the considered class of safety guarantees if the problems are detected at runtime. Only in cases where the runtime verification is not required in hard real-time (and thus remain outside the safe core) such an approach seems reasonable.)

% Different OCM hierarchies as components have peer-to-peer communications and collaborate by exchanging information. The mUML approach uses the refined UML model component to formally define OCM hierarchies. It also uses refined state charts to define timing constraints and hybrid behaviors of the components. The interactions are also specified by coordination patterns in UML. This approach uses compositional model checking to check the safety of the system. The system is safe if each component of the system is safe and deadlock-free, and  constraints of all patterns are fulfilled. Therefore, besides separately checking the safety of each component, mUML checks that the interfaces of the components are well-defined and the components refine their interfaces. In contrast to our approach, this approach is able to model hybrid behaviors, but change propagation is not a considered in  \cite{Giese2013}. Furthermore, as explicitly mentioned in the paper, the verification at runtime is not a concern in \cite{Giese2013}. %This is a design time approach (However, it would be too late for the considered class of safety guarantees if the problems are detected at runtime. Only in cases where the runtime verification is not required in hard real-time (and thus remain outside the safe core) such an approach seems reasonable.)

 The ACPS language to design and verify self-adaptive CPSs is proposed in \cite{8595374}. %The components of the system are modeled as processes of ACPS.  
 The components of the system are categorized in different groups such that each group affects on satisfaction of one requirement. %This is because different requirements are satisfied by different groups of components. For each group an adaptation is encoded such that the requirement is preserved. 
 An adaptation is encoded for each group such that the requirement is preserved. Finally, the ACPS definition of each group is translated to a verification tool and is separately verified. This work assumes that grouping the relevant components to a requirement is possible. %Finally, each group defined in ACPS is independently translated to a verification tool and is verified separately. This work assumes that the requirements allow to identify and group the relevant components to a requirement.  
 In contrast to \cite{8595374}, instead of considering a fixed number of components per each requirement, we increase the verification domain whenever it is needed. Furthermore, \cite{8595374} does not consider the change propagation phenomenon.
 
%Identifying components relevant to a requirement is manual. change propagation, tight interaction of the components, or even decomposing a requirement ... when requirements allow it... we dynamically increase the components. %There can be different adaptation procedure to preserve a requirement, whose performances can be compared. They use the FDR, a refinement checking tool for CPSs.
 %Different groups of components may be required to satisfy each requirement. They propose the Adaptive CPS (ACPS) language to modulary design self-adaptive CPSs and using compositional verification. The components are modeled as process of ACPS. Then components are grouped such that each group of components affects on satisfaction of each requirement. For each group an adaptation procedure is model such that the requirement is preserved. Each group of systems is independently verified. There can be different adaptation procedure to preserve a requirement, whose performances can be compared. They use the FDR, a refinement checking tool for CPSs. 

\textbf{Formal Analysis of Self-adaptive Systems at Runtime}. 
%A set of approaches such as \cite{IROPS,PVARFSAS,DWNRFSAS,LEE2018200,7968129,DBLP:conf/sigsoft/MorenoCGS15} use  state-based models (that are defined by states and transitions) to verify self-adaptive systems at runtime. 
Incremental runtime verification of MDPs, described in the PRISM language, is proposed in \cite{IROPS}, where runtime changes are limited to vary parameters of the PRISM model. An MDP is constructed incrementally by inferring a set of states needed to be rebuilt. The constructed MDP is then verified using an incremental verification technique. %The proposed QoSMOS in \cite{5611553}, integrates a set of pre-existing tools to model the MAPE-K feedback loop. The model$@$runtime in \cite{5611553} is expressed in PRISM language and is proactively adapted  by aim of optimizing the quality properties of the service-based systems. 
Runtime verification of parametric Discrete Time Markov Chains (DTMCs) is accomplished in \cite{PVARFSAS}. In this method, probabilities of transitions are given as variables. Then, the model is analyzed and a set of symbolic expressions is reported as the result. By substituting real values of the variables at runtime, verification is reduced to calculating the values of the symbolic expressions. 

In \cite{DWNRFSAS}, a self-adaptive software is designed as a dynamic software product line (DSPL). Then, an instance of DSPL is chosen at runtime considering the environmental changes. This approach uses parametric DTMCs to model common behaviors of the products and each variation point separately. Therefore, there is no need to verify each configuration separately. %RINGA, a self-adaptive framework for runtime verification of self-adaptive systems, is proposed in \cite{LEE2018200}. 
RINGA, introduced in \cite{LEE2018200},  uses Finite State Machines (FSM) to develop a design-time model of a system, and abstracts the model for using at runtime. Each state of the model implements a module, while a transition triggers an adaptation. Each transition is assigned an equation that is parameterized by  environmental variables. The value of the equation is calculated at runtime. Lotus$@$runtime \cite{7968129} %is a tool for verifying self-adaptive systems at runtime. It 
uses Probabilistic Labeled Transition Systems (PLTS) to develop a model@runtime. It monitors execution traces of the system and updates the probabilities in PLTS. The desirable properties in \cite{7968129} are explained through a source state, a target state, a condition to be satisfied, and the probability of satisfying the condition. %In \cite{8008800}, a complete methodology for developing trustworthy self-adaptive systems along with their assurance cases is proposed. An assurance case gives some evidence about safety of the system for a special application and in a special environment. The evidences are generated from combination of testing, verification, and simulation at design time and runtime. The authors present a tool-supported instance of their methodology. This instance uses UPPAL to verify the controller part at design time and PRISM to verify the controlled system and its environment at runtime. They exploit the ActiveForm approach \cite{Iftikhar:2014:AAF:2593929.2593944} and models the the controller through a network of Timed Automatan.  

In comparison to \cite{IROPS,PVARFSAS,DWNRFSAS,LEE2018200,7968129} which use  state-based models, an actor model is in a higher level of abstraction. Our actor-based approach besides decreasing the semantic gap between the model$@$runtime and  applications, %(there is a one-to-one mapping between elements of the system and the model), 
facilitates the modular analysis of the system. %In other words, not only the model of the whole system is decomposed into a set of components, but also the actor model of each component can be decomposed into the smaller components if the analysis yet suffers from the state-space explosion.
%There is a third class of the related work about the verification and formal analysis of self-adaptive systems at runtime. Since different adaptive systems have different requirements, there is no unique applicable approach  to model and verify all self-adaptive systems. In this part, the most closely related work to our verification approach and verification of self-adaptive systems are introduced. 

The failure propagation is studied in \cite{ Priesterjahn2013} that checks whether the structural adaptation of the system is fast enough to prevent a hazard. %Each adaptation takes an amount of time during which the failure is propagated through the system. 
After an adaptation, it is checked whether the remaining failures in the system lead to a hazard. Our approach, besides detecting a hazard, assures the satisfaction of the timing properties of the system. Based on the circumstances existing at the time the change occurs, %different values of latency for different adaptation policies can be imposed on the model by defining a
different thresholds over the analysis time can be imposed.
%in our approach. 
It is assumed that humans are involved if the adaptation cannot be handled during the expected time. The latency-aware adaptation is studied in \cite{DBLP:conf/sigsoft/MorenoCGS15}, where a probabilistic model checker proactively selects an adaptation strategy to maximize the utility of the system. Unlike  \cite{DBLP:conf/sigsoft/MorenoCGS15}, our focus is on effectively verifying the system behavior. 

%Unlike our approach, \cite{DBLP:conf/sigsoft/MorenoCGS15} uses probabilistic models, described in the PRISM language. However, our focus is on effectively verifying the system behavior. 

The work of \cite{6224401} investigates which state of the system is a safe state to update the implementation of the system whenever an environment assumption is changed. Furthermore, based on the old controller, a new controller is automatically synthesized for the software system. %A controller implements a specification consisting of requirements and environmental assumptions. A specification is described in the form of Modal Sequence Diagrams (MSDs). An MSD describes which sequence of interactions may, must, or must not occur during the system executions. Therefore, the synthesis of a new controller is based on the synthesis of MSDs, as specifications are updated by adding or removing MSDs. 
The approach of \cite{6224401} is applied on a RailCab system where an accident should be avoided before the RailCabs enter into a crossing. An infrastructure to deploy and execute new controllers on embedded devices is proposed in \cite{LaManna:2015:TED:2820489.2820499}. %The running example in \cite{LaManna:2015:TED:2820489.2820499} is a space exploration robot system, where several robots share a recharging station. To avoid a collision, a robot controller should check whether the charging station is available whenever a robot wants to recharge itself. 
\cite{6224401} does not verify the system after adapting it to a change. %in environment assumptions. 

%A multi-tier framework for modeling, planning, and enactment is proposed in \cite{D'Ippolito:2014:HBP:2568225.2568264}. Each tier has a goal and an environmental model that is used to build planning strategies. Each tier is also associated with a level of  functionality guaranteed through the tier. The highest level model has the highest level of  functionality. When an assumption of a higher model  is violated, the immediate lower level model is used, and this way the functionality of the system degrades. The system enhances the functionality by using the higher level whenever the problem of the higher level model is solved. The running example of this work is an automatic warehouse where several robots receive demands, take the supplies, fulfill orders, and deliver completed orders to the customers. Although \cite{D'Ippolito:2014:HBP:2568225.2568264} deals with the adaptation whenever an environment assumption is violated,  it does not attend to the verification of the system.

%% file: 07-Conclusion.tex
\section{Discussion and Future Work} \label{sec::conclude}
We proposed Magnifier, a compositional approach that iteratively detects the propagation of a change and  incrementally involves the components affected by a change into the analysis. An adaptation policy may contain the change and prevent the change to be propagated. %When a change happens to a component, different adaptation policies are investigated. If an adaptation policy contains the change, then the change does not propagate to the environment components. Therefore, it is possible that the change propagation stops. % at a time in the future.
In the worse case,  the change propagates to the whole system and Magnifier needs to compose all components of the model.
We compared the compositional approach of Magnifier and the non-compositional approach in Section \ref{sec::evaluation}.
The comparisons between our model, CoodAA, and other similar models on self-adaptive systems are presented in \cite{cooda,16_CoordinatedActorsForReliableSelfAdaptiveSystems}. In Section \ref{sec::relatedwork}, we included a comparison of  Magnifier with other analysis approaches for TTCS, and other compositional methods.  

Here we include  an observation regarding a comparison between the non-compositional approach and the worst case for Magnifier when the change propagates all the way to include the whole system. 
We argue that even in the worst case Magnifier performs better than the non-compositional approach. Our experiments testify this observation.
Looking more carefully into this comparison and building a formal proof is a part of our future work.
%}  
%\COM{Please rewrite the following and refer to the state space a graph, and how the graph for the compositional approach has lower branching factor in the beginning ... in the compositional approach, the enabled transitions in each state is limited to the events of actors along the routes of messages that are in the scope of the change  while in the noncompositional approach  all the actors in the route may have enabled events (which means outgoing transitions) ... when we go forward in generating the state space, further along the way, the branching factor becomes equal ...}
%
%
This observation can be justified as follows. Suppose that a change happens at time \textit{t}. The non-compositional approach involves all actors of the model that have a message at time \textit{t} into the analysis, and starts to generate the state space. In contrast, Magnifier focuses on the component affected by the change, and starts to generate the state space by involving those actors of the component having a message at time \textit{t}. Let the branching factor for a state be the number of outgoing transitions of the state or the number of actors that can be triggered at the state. At the beginning, Magnifier has a lower branching factor on all states of the state space compared to the non-compositional approach. At some point in future, e.g. at time \textit{t'}, when all components are affected by the change, both approaches involve the same number of actors into the analysis, and both approaches generates the same number of states and transitions. However, between \textit{t} and \textit{t'}, the graph of the state space in Magnifier is smaller than the graph of the state space in the  non-compositional approach. Therefore, even in the worst case, Magnifier performs better than the non-compositional approach in terms of the verification time and the memory consumption.

%Upon occurring a change, the non-compositional approach involves all actors along the routes of message 
%into the analysis, and the interleavings arisen from the concurrent executions of the actors lead the time and memory consumptions increase. Since Magnifier increases the involved actors at different iterations at different times, in overall, the amount of concurrency arisen during the analysis is less than the non-compositional approach, and as a consequence, Magnifier performs better than the non-compositional approach in overall. But Magnifier and the non-compositional approach may have equal performance from a time on, when all components of the model are composed. Maybe a solution to overcome this problem is periodical analysis of the model when all components are composed. Let $T$ denotes the period. We can check the reachable state space of the model for the next $T$ time units of the real-world in a negligible amount of time, and apply the adaptation plan to the system. We then analyze the model for the next period and obtain an adaptation plan. %This plan is only  valid for the next $T$ time units. 
%However, this approach needs to study many details such as detecting which components or which parts of the components are involved in the analysis in each period, how the adaptation should be performed, what happens for the circularity of the change, etc. Mixing the Magnifier approach with the periodical analysis of the system is performed as our future work.

Our Ptolemy II implementation of Magnifier is specialised for ATC. 
In \cite{10.1007/978-3-030-02146-7_1}, Lee and Sirjani show how CoodAA can capture TTCS applications in general. Here we consider a constant number of four for the ports of all actors, and the topology formed by connecting the ports through channels is a mesh. The extension to a dynamic number of ports and further than to dynamic bindings, seem like natural future work.
The general idea of Magnifier is not limited to TTCSs. It can be generalised for any control system with a modular design. We need to extend our model to include more general actors with different behaviors, and also different number of ports, and  different bindings among the ports (channels that form the topology).
To investigate the details of such extension is another future direction.
The possibility of analyzing actors in a compositional way is a consequence of their isolation discussed in  \cite{8958682} by Sirjani, Khamespanah and Ghassemi. 
Hence, we believe that CoodAA and Magnifier can be further extended and used in different areas and applications based on the foundations provided in this paper.
%}

%We propors wsed ill help .--a...n approach in which the model of the system is decomposed into a set of components. Upon encountering a change, the component affected by the change is adapted. If TIOAs of the adapted component and its interface components are compatible, the change does not propagate. Otherwise, the environment components affected by the change propagation are adapted. Then, all components affected by the change are composed to create a new component. The same process is repeated for the new component.  We model each component by a coordinated actor model, whose semantics is defined based on a network of TIOAs. 
 
%The correctness properties in our approach are built into the model. They are preserved if TIOAs of the adapted component and its interface components are compatible. When a change happens to a component, different adaptation policies are investigated. An adaptation policy may contain the change, and the change does not propagate to the environment components. Therefore, it is possible that the change propagation stops at a time in the future. However, in the worse case that the change propagates to the whole system, all components are composed in our approach.  We implemented our approach in the Ptolemy II framework. The results of our experiments confirm that our approach decreases the verification time and the number of states. To improve our approach, we will study the system behavior in several time windows in the case of propagating the change through the whole system as the future work.

%% file: appendix.tex
	
	In this section, the pseudo-codes for ALG1 and ALG2 are given. We only present the details that are necessary for understanding the idea of ALG1 and ALG2. 
	
	\textit{\textbf{Pseudo-code for ALG1}}.	Algorithm~\ref{algo::ALG1} generates the initial flight plans of $m$ aircraft. It calls Algorithm~\ref{algo::generateRoutes} to generate a route. The inputs, outputs, variables, and the functions used in Algorithm~\ref{algo::ALG1} are described as follows. 
	\begin{itemize}
	    \item $m$ is the number of moving objects.
	    \item $\lambda$ is the parameter of the exponential distribution.
	    \item $\mathit{FD}$  is the traveling time across a sub-track.
	    \item $\mathit{flightPlans}$  is the list of initial flight plans of $m$ aircraft.
	    \item \textit{genSource()} randomly generates the source of an aircraft.
	    \item \textit{genDes()} randomly generates the destination of an aircraft.
	    \item $\mathit{genDepartureTime(x_s,y_s,\lambda)}$ randomly generates the departure time of an aircraft from the  source airport connected to the given sub-track, where the departure time follows an exponential distribution with the parameter $\lambda$.
	    \item $\mathit{add(flightPlans, j,x_s,y_s,x_d,y_d,dTime,route)}$ adds the flight plan of a given aircraft with its source, destination, route, and  departure time to the list of initial flight plans and returns the resulting list.
	    \item $\mathit{generateRoute(x_s,y_s,dTime,x_d,y_d,[~])}$ generates the initial route of an aircraft. The pseudo-code of this  function is given in Algorithm~\ref{algo::generateRoutes}.
	\end{itemize}
	
	\begin{algorithm}
	\DontPrintSemicolon 
	\SetKwFunction{algo}{All}
	\KwIn{$m$, $\lambda$, and $\mathit{FD}$ }
	\KwOut{$\mathit{flightPlans}$}
	\Begin{
		$\mathit{flightPlans} \gets [~]$\;
		$j \gets 0$\;
	    \While{$j < m$}{
    	    $(x_s,y_s) \gets \mathit{genSource()}$\;
    	    $(x_d,y_d) \gets \mathit{genDes()}$\;
    	    $\mathit{dTime} \gets \mathit{genDepartureTime(x_s,y_s,\lambda)}$\;
    	    $\mathit{route} \gets \mathit{generateRoute(x_s,y_s,dTime,x_d,y_d,[~])}$\;
    	    \If{$\mathit{route} \neq [~]$}{
    	        $j \gets j+1$\;
    	        $\mathit{flightPlans} \gets \mathit{add(flightPlans, j,x_s,y_s,x_d,y_d,dTime,route)}$\;
    	    }
	    }
		
		\nl \KwRet $\mathit{flightPlans}$\;
	}
	\caption{ALG1}
	\label{algo::ALG1}
\end{algorithm}

The inputs, outputs, variables, and the functions used in Algorithm~\ref{algo::generateRoutes} are described as follows.
\begin{itemize}
    \item $(x_s,y_s)$  is the first sub-track (connected to the source airport) in the route of the aircraft.
    \item \textit{arrivalTime} is the arrival time of the aircraft at the first sub-track, which is equal to the departure time of the aircraft from the source airport.
    \item $(x_d,y_d)$  is the last sub-track (connected to the destination airport) in the route of the aircraft.
    \item \textit{flightPlans} is the list of flight plans of  aircraft (does not contain the flight plans of all $m$ aircraft when it is given as the input to the algorithm).
    \item \textit{route} is a route that is an empty list at first. To create the route, Algorithm~\ref{algo::generateRoutes} is recursively executed and sub-tracks are gradually added to the list.
    \item $\mathit{hasTimeConflict(x_s,y_s,arrivalTime,flightPlans)}$ returns true if  there exists an aircraft that either arrives at the sub-track $(x_s,y_s)$ at time $\mathit{arrivalTime}$ or departs from the sub-track at time $\mathit{arrivalTime+FD}$.
    \item $\mathit{add(route,(x_s,y_s))}$ adds the sub-track $(x_s,y_s)$ to the end of  $\mathit{route}$ and returns $\mathit{route}$.
    \item  $\mathit{increaseX}$ has a value of $\{0,1\}$. Let $(x,y)$ denotes the current position on a route. If  $\mathit{increaseX}$ is 1 then $x$ should be changed to $x+1$ (traversing the X dimension). Otherwise, $x$ does not change.
    \item $\mathit{doIncreaseX(x_s,x_d)}$ returns $1$ if $x_s < x_d$. Otherwise, it returns 0.
    \item  $\mathit{increaseY}$ has a value of $\{0,1,-1\}$. Let $(x,y)$ denotes the current position on a route. If  $\mathit{increaseY}$ is 1 then $y$ should be changed to $y+1$ and if $\mathit{increaseY}$ is -1 then $y$ should be changed to $y-1$ (both for traversing the Y dimension) . Otherwise, $y$ does not change.
    \item $\mathit{doIncreaseY(y_s,y_d)}$ returns 1 if $y_s < y_d$ and -1 if $y_s > y_d$. Otherwise, it returns 0.
    \item $\mathit{removeLastElement(route)}$ removes the last element from $\mathit{route}$ and returns the rest of it.
\end{itemize}
	\begin{algorithm*}
	\DontPrintSemicolon 
	\SetKwFunction{algo}{All}
	\KwIn{$(x_s,y_s)$, $\mathit{arrivalTime}$, $(x_d,y_d)$, $\mathit{flightPlans}$, $\mathit{route}$}
	\KwOut{$\mathit{route}$}
	\Begin{
	    \If{$\mathit{hasTimeConflict(x_s,y_s,arrivalTime,flightPlans)}$}{
	        \nl \KwRet $[~]$\;
	    }
	    $\mathit{route} \gets \mathit{add(route,(x_s,y_s))}$\;
		\If{$x_s=x_d$ and  $y_s=y_d$ }{
    	    \nl \KwRet $\mathit{route}$\;
    	  }
    	  
    	  $\mathit{increaseX} \gets \mathit{doIncreaseX(x_s,x_d)}$, $\mathit{increaseY} \gets doIncreaseY(y_s,y_d)$\;
    	   
    	   \eIf{$\mathit{increaseX}=1$}{
    	        \If{$\mathit{generateRoute(x_s+1, y_s,arrivalTime+FD, x_d, y_d,route)=[~]}$}{
    	            \If{$\mathit{increaseY}=1$}{
    	                \If{$\mathit{generateRoute(x_s, y_s+1,arrivalTime+FD, x_d, y_d,route)=[~]}$}{
    	                    $route \gets \mathit{removeLastElement(route)}$\;
    	                     \nl \KwRet $[~]$\;
    	            }
    	            }
    	            \If{$\mathit{increaseY}=-1$}{
    	                \If{$\mathit{generateRoute(x_s, y_s-1,arrivalTime+FD, x_d, y_d,route)=[~]}$}{
    	                    $route \gets \mathit{removeLastElement(route)}$\;
    	                     \nl \KwRet $[~]$\;
    	            }
    	        }
    	        \If{$\mathit{increaseY}=0$}{
    	                    $route \gets \mathit{removeLastElement(route)}$\;
    	                     \nl \KwRet $[~]$\;
    	        }
    	        }
    	        
    	   }
    	   {
    	    \eIf{$\mathit{increaseY}=1$}
    	        {
    	            \If{$\mathit{generateRoute(x_s, y_s+1,arrivalTime+FD, x_d, y_d,route)=[~]}$}{
    	                    $route \gets \mathit{removeLastElement(route)}$\;
    	                     \nl \KwRet $[~]$\;
    	            }
    	        }
    	        {
    	        \eIf{$\mathit{increaseY}=-1$}{
    	                \If{$\mathit{generateRoute(x_s, y_s-1,arrivalTime+FD, x_d, y_d,route)=[~]}$}{
    	                    $route \gets \mathit{removeLastElement(route)}$\;
    	                     \nl \KwRet $[~]$\;
    	            }
    	        }{
    	        \If{$\mathit{increaseY}=0$}{
    	                    $route \gets \mathit{removeLastElement(route)}$\;
    	                     \nl \KwRet $[~]$\;
    	            
    	        }
    	        }
    	        }
    	   }
    	    \nl \KwRet $\mathit{route}$\;
	}
	\caption{$\mathit{generateRoute}$}
	\label{algo::generateRoutes}
\end{algorithm*}
	
\textit{\textbf{Pseudo-code for ALG2}}. Assume that the aircraft is going to leave the location $\mathit{(x_0,y_0)}$ and the rest of its route is $\mathit{[(x_1,y_1),(x_2,y_2),\cdots,(x_n,y_n)]}$. Also, assume that the sub-track  with the location $(x_1,y_1)$ is unavailable, and the aircraft is not able to travel through it. Algorithm~\ref{algo::rerouting} generates a new route (a new flight plan) for the aircraft. The initial route in this algorithm is $\mathit{[(x_1,y_1),(x_2,y_2),\cdots,(x_n,y_n)]}$. The inputs, outputs, variables, and the functions used in Algorithm~\ref{algo::rerouting} are described as follows.
\begin{itemize}
    \item $\mathit{initialFlightPlan}$ is the initial flight plan of the aircraft.
    \item $\mathit{initialRoute}$ is the initial route of the aircraft.
    \item $(x_0,y_0)$ is the sub-track accommodating the aircraft. The aircraft is going to leave this sub-track.
    \item $(x_1,y_1)$ is the next sub-track in the route of the aircraft. This sub-track is unavailable.
    \item $\mathit{flightPlan}$ is a new flight plan for the aircraft.
    \item $\mathit{findAvaNeighbors(x_0,y_0,x_1,y_1)}$ returns available sub-tracks adjacent to $(x_0,y_0)$. 
    \item $\mathit{findNStormyNeighbors(x_0,y_0,x_1,y_1)}$ returns  sub-tracks adjacent to $(x_0,y_0)$ except for $(x_1,y_1)$, where these sub-tracks are not stormy.
    \item $\mathit{routeSelected}$ returns true if a route has been found. 
    \item $\mathit{needDelay}$ returns true if the aircraft needs to stay one more unit of time in  $(x_0,y_0)$.
    \item $\mathit{take()}$ takes an element from a given list, removes the element from the list, and returns the element.
    \item $\mathit{length}$ returns the length of a given list.
    \item $\mathit{concat(route,initialRoute,i)}$ concatenates \textit{route} with the rest of the route of
the aircraft from $(x_i+1 , y_i+1 )$ to $(x_n, y_n)$ in \textit{initialRoute}  and returns the resulting
route.
    \item $\mathit{generateRoute}$ uses an algorithm  simpler than Algorithm~\ref{algo::generateRoutes} to create a route based on the XY routing algorithm. It avoids the stormy track, but it does not check the time conflict with other aircraft in the future.
    \item $createFlightPlan$ creates a new flight plan based on the initial flight plan of the aircraft and the calculated route. The last input of this function is true if the aircraft needs to stay one more unit of time in  $(x_0,y_0)$.

\end{itemize}

	\begin{algorithm*}
	\DontPrintSemicolon 
	\SetKwFunction{algo}{All}
	\KwIn{$\mathit{initialFlightPlan}$, $\mathit{initialRoute}$, $(x_0,y_0)$, and $(x_1,y_1)$}
	\KwOut{$\mathit{flightPlan}$}
	\Begin{
	    $\mathit{avaNeighbors \gets findAvaNeighbors(x_0,y_0,x_1,y_1)}$\;
	    $\mathit{neighbors \gets findNStormyNeighbors(x_0,y_0,x_1,y_1)}$\;
	    $\mathit{routeSelected \gets False}$\;
	    $\mathit{needDelay \gets False}$\;
	    $i \gets 1$\;
	    $\mathit{route} \gets [~]$\;
        \While{$\mathit{avaNeighbors} \neq \emptyset$}{
	        $(x_s,y_s) \gets \mathit{take(avaNeighbors)}$\;
	        \While{$\mathit{i<length(initialRoute)}$}{
	            $(x_d,y_d) \gets \mathit{initialRoute[i]}$\;
	            $\mathit{route \gets generateRoute(x_s,y_s,x_d,y_d)}$\;
	            \If{$\mathit{route} \neq [~]$ and $\mathit{length(route)=(i+1)}$}{
	                $\mathit{route} \gets \mathit{concat(route,initialRoute,i)}$\;
	                $\mathit{routeSelected \gets True}$\;
	                $\mathit{break}$\;
	            }
	        }
	        \If{$\mathit{routeSelected=True}$}{
	            $\mathit{break}$\;
	        }
		}

        \If{$\mathit{routeSelected=False}$}{
		    \While{$\mathit{neighbors} \neq \emptyset$}{
		    	$(x_s,y_s) \gets \mathit{take(neighbors)}$\;
		    	$(x_d,y_d) \gets \mathit{initialRoute[length(initialRoute)-1]}$\;
	            $\mathit{route \gets generateRoute(x_s,y_s,x_d,y_d)}$\;
			    \If{$\mathit{route} \neq [~]$}{
			        $\mathit{routeSelected \gets True}$\;
			         $\mathit{needDelay \gets True}$\;
	                $\mathit{break}$\;
			    }
		}
		
    }
    
    \eIf{$\mathit{routeSelected=True}$}{
	            $\mathit{flightPlan} \gets \mathit{createFlightPlan(\mathit{initialFlightPlan},route,needDelay)}$\;
	   }
	   {
	            $\mathit{flightPlan} \gets \mathit{createFlightPlan(\mathit{initialFlightPlan},initialRoute,True)}$\;
	   }
		
		\nl \KwRet $\mathit{flightPlan}$\;
	}
	\caption{ALG2}
	\label{algo::rerouting}
\end{algorithm*}	